\let\oldvec\vec
\documentclass{llncs}
\let\vec\oldvec

\usepackage[utf8]{inputenc}
\usepackage{mdframed, float}
\usepackage{microtype,lmodern,amsmath,amsthm,amsfonts,longtable,etoolbox,tikz, colortbl, mathtools}
\usetikzlibrary{graphs,quotes,decorations.pathreplacing,petri,snakes}
\usepackage{csquotes}
\usepackage{multirow}
\usepackage{array}
\usepackage{pbox}
\usepackage{caption}
\usepackage[misc,geometry]{ifsym} % for Letter (corresponding author)
\newlength{\problemoffset}
\newcommand{\escale}[1]{\ensuremath{\textbf{\scalebox{0.8}{#1}}}}
\newcommand{\nscale}[1]{\ensuremath{\textbf{\scalebox{0.8}{#1}}}}
\newcommand{\escales}[1]{\ensuremath{\textbf{\scalebox{0.6}{#1}}}}
\newcommand{\nscales}[1]{\ensuremath{\textbf{\scalebox{0.7}{#1}}}}

\newcommand{\myEdge}[2]{ \tikz[baseline=-1pt]{
\draw[#2,line width=0.3pt] (0,0) -- ++(0.6,0) node[anchor=base, yshift=3pt, pos=0.5] {\escale{$#1$}};
}}

\newcommand{\Fbedge}[1]{ \tikz[baseline=-1pt]{
\draw[<->,line width=0.3pt] (0,0) -- ++(0.6,0) node[anchor=base, yshift=5pt, pos=0.5] {\escale{$#1$}};
}}

\newcommand{\ledge}[1]{ \tikz[baseline=-1pt]{
\draw[->,line width=0.3pt] (0,0) -- ++(0.85,0) node[anchor=base, yshift=4pt, pos=0.5] {\escale{$#1$}};
}}
\newcommand{\lledge}[1]{ \tikz[baseline=-1pt]{
\draw[->,line width=0.3pt] (0,0) -- ++(1,0) node[anchor=base, yshift=4pt, pos=0.5] {\escale{$#1$}};
}}

\newcommand{\edge}[1]{\myEdge{#1}{->}}

\newcommand{\fbedge}[1]{\myEdge{#1}{<->}}

\newcommand{\nop}{\ensuremath{\textsf{nop}}}
\newcommand{\inp}{\ensuremath{\textsf{inp}}}
\newcommand{\out}{\ensuremath{\textsf{out}}}
\newcommand{\set}{\ensuremath{\textsf{set}}}
\newcommand{\res}{\ensuremath{\textsf{res}}}
\newcommand{\swap}{\ensuremath{\textsf{swap}}}
\newcommand{\free}{\ensuremath{\textsf{free}}}
\newcommand{\used}{\ensuremath{\textsf{used}}}

\newcommand{\exit}{\ensuremath{\mathfrak{exit}}}
\newcommand{\enter}{\ensuremath{\mathfrak{enter}}}
\newcommand{\saveone}{\ensuremath{\mathfrak{save}_1}}
\newcommand{\savezero}{\ensuremath{\mathfrak{save}_0}}
\newcommand{\save}{\ensuremath{\mathfrak{save}}}

\begin{document}
\title{The Complexity of Boolean State Separation (Technical Report)} 
\author{Ronny Tredup(\Letter)\inst{1} \and
Evgeny Erofeev\thanks{Supported by DFG through \mbox{grant Be 1267/16-1} {\tt ASYST}.}\inst{2}}
\institute{
Universit\"at Rostock, Institut f\"ur Informatik, Theoretische Informatik, Albert-Einstein-Stra\ss e 22, 18059, Rostock 
(\email{ronny.tredup@uni-rostock.de})
\and
Department of Computing Science, Carl von Ossietzky Universit\"at Oldenburg,\\ D-26111 Oldenburg, Germany 
(\email{evgeny.erofeev@informatik.uni-oldenburg.de})
}
\maketitle

\begin{abstract}
For a Boolean type of nets $\tau$, a transition system $A$ is synthesizeable into a $\tau$-net $N$ if and only if distinct states of $A$ correspond to distinct markings of $N$, and $N$ prevents a transition firing if there is no related transition in $A$. 
The former property is called $\tau$-state separation property ($\tau$-SSP) while the latter -- $\tau$-event/state separation property ($\tau$-ESSP). 
$A$ is embeddable into the reachability graph of a $\tau$-net $N$ if and only if $A$ has the $\tau$-SSP. 
This paper presents a complete characterization of the computational complexity of \textsc{$\tau$-SSP} for all Boolean Petri net types.
\end{abstract}


%%%%%%%%%%%
\section{Introduction}%
%%%%%%%%%%%

Providing a powerful mechanism for the modeling of conflicts, dependencies and parallelism, Petri nets are widely used for studying and simulating concurrent and distributed systems. 
In system analysis, one aims to check behavioral properties of system models, and many of these properties are decidable~\cite{DBLP:journals/eatcs/EsparzaN94} for Petri nets and their reachability graphs, which represent systems' behaviors. 
The task of system synthesis is opposite: a (formal) specification of the system's behavior is given, and the goal is then to decide whether this behavior can be implemented by a Petri net. 
In case of a positive decision, such a net should be constructed. 

Boolean Petri nets form a simple yet rich and powerful family of Petri nets~\cite{DBLP:series/txtcs/BadouelBD15,DBLP:journals/acta/BadouelD95,DBLP:journals/acta/KleijnKPR13,DBLP:journals/acta/MontanariR95,DBLP:conf/apn/Pietkiewicz-Koutny97,DBLP:conf/ac/RozenbergE96,DBLP:conf/stacs/Schmitt96}, applied in asynchronous circuits design~\cite{10.1007/BFb0055644,DBLP:journals/integration/YakovlevKSK96}, concurrent constraint programs~\cite{10.1007/3-540-57787-4_18} and  analysis of biological systems models~\cite{DBLP:journals/nc/ChatainHKPT20}. 
In Boolean nets, each place contains at most one token, for any reachable marking. 
Hence, a place can be interpreted as a Boolean condition that is \emph{true} if marked and \emph{false} otherwise.
A place $p$ and a transition $t$ of such a net are related by one of the Boolean \emph{interactions} that define in which way $p$ and $t$ influence each other. 
The interaction $\inp$ ($\out$) defines that $p$ must be \emph{true} (\emph{false}) before and \emph{false} (\emph{true}) after $t$'s firing;
$\free$ ($\used$) implies that $t$'s firing proves that $p$ is \emph{false} (\emph{true});
$\nop$ means that $p$ and $t$ do not affect each other at all; 
$\res$ ($\set$) implies that $p$ may initially be both \emph{false} or \emph{true} but after $t$'s firing it is \emph{false} (\emph{true});
$\swap$ means that $t$ inverts $p$'s current Boolean value.
Boolean Petri nets are classified by the sets of interactions that can be applied. 
A set $\tau$ of Boolean interactions is called a \emph{type of net}, and a net $N$ is of type $\tau$ (a \emph{$\tau$-net}) if it applies at most the interactions of $\tau$. 
For a type $\tau$, the $\tau$-\emph{synthesis} problem consists in deciding whether a specification given in the form of a labeled transition system (TS) is isomorphic to the reachability graph of some $\tau$-net $N$, and in constructing $N$ if it exists.
The complexity of $\tau$-synthesis has been studied in different settings~\cite{tredup2019complexity,DBLP:conf/sofsem/Tredup20,DBLP:conf/tamc/TredupR19}, and varies substantially from polynomial~\cite{DBLP:conf/stacs/Schmitt96} to NP-complete~\cite{DBLP:journals/tcs/BadouelBD97}.

In order to perform synthesis, that is, to implement the behavior specified by the given TS with a $\tau$-net, two general problems have to be resolved: 
The $\tau$-net has to distinguish the global states of the TS, and the $\tau$-net has to prevent actions at states where they are not permitted by the TS. 
In the literature~\cite{DBLP:series/txtcs/BadouelBD15}, the former requirement is usually referred to as \emph{$\tau$-state separation property} ($\tau$-SSP), while the latter -- \emph{$\tau$-event/state separation property} ($\tau$-ESSP). 
Both $\tau$-SSP and $\tau$-ESSP define decision problems that ask whether a given TS fulfills the respective property.
The present work focuses exclusively on the computational complexity of $\tau$-SSP. 
The interest to state separation is motivated in several ways. 
First, many synthesis approaches are very sensitive to the size of the input's state space. 
This raises the question if some initial, so-called \emph{pre-synthesis} procedures~\cite{DBLP:journals/scp/BestD18,DBLP:journals/iandc/Wimmel20} can be employed as a quick-fail mechanism, i.e., techniques with a small computational overhead that would gain some helpful information for the main synthesis, or reject the input if exact (up to isomorphism) synthesis is not possible. 
Since $\tau$-synthesis allows a positive decision if and only if $\tau$-SSP and $\tau$-ESSP do~\cite{DBLP:series/txtcs/BadouelBD15}, an efficient decision procedure for $\tau$-SSP could serve as a quick-fail pre-process check.
Second, if exact synthesis is not possible for the given TS, one may want to have a simulating model, i.e., a $\tau$-net that over-approximates~\cite{DBLP:conf/lata/Schlachter18,DBLP:journals/corr/abs-2002-04841} the specified behavior with some possible supplement. 
Formally, the TS then has to be injectively embeddable into the reachability graph of a $\tau$-net. 
It is well known from the literature~\cite{DBLP:series/txtcs/BadouelBD15} that a TS can be embedded into the reachability graph of a $\tau$-net if and only if it has the $\tau$-SSP.
Finally, in comparison to $\tau$-ESSP, so far the complexity of $\tau$-SSP is known to be as hard \cite{DBLP:conf/tapsoft/BadouelBD95,DBLP:conf/stacs/Schmitt96,DBLP:conf/tamc/TredupR19,DBLP:conf/apn/TredupRW18} or actually less hard \cite{DBLP:conf/apn/Tredup19a,DBLP:conf/apn/Tredup19}. 
On the contrary, in this paper, for some types of nets, deciding the $\tau$-SSP is proven harder (NP-complete) than deciding the $\tau$-ESSP (polynomial), e.g., for $\tau = \{ \nop, \res, \set\}$. 
From the contribution perspective, the $\tau$-SSP has been previously considered only in the broader context of $\tau$-synthesis, and only for selected types~\cite{DBLP:conf/stacs/Schmitt96,DBLP:conf/apn/Tredup19,DBLP:conf/apn/TredupE20,DBLP:conf/tamc/TredupR19}. 
In this paper, we completely characterize the complexity of $\tau$-SSP for all 256 Boolean types of nets and discover 150 new hard types, cf.~\S1 - \S3, \S6, \S9 of ~Figure~\ref{fig:overview}, and 4 new tractable types, cf.~Figure~\ref{fig:overview}~\S10, and comprise the known results for the other 102 types as well, cf.~Figure~\ref{fig:overview}~\S4, \S5, \S7, \S8.
In particular, our characterization categorizes Boolean types with regard to their behavioral capabilities, resulting from the Boolean interactions involved. 
This reveals the internal organization of the entire class of Boolean nets, suggesting a general approach for reasoning about its subclasses.

This paper is organized as follows. 
In Section~\ref{sec:Preliminaries}, all the necessary notions and definitions will be given. 
Section~\ref{sec:nop_equipped} presents NP-completeness results for the types with $\nop$-interaction. 
$\tau$-SSP for types without $\nop$ is investigated in Section~\ref{sec:nop-free}. 
Concluding remarks are given in Section~\ref{sec:conclusion}. 
One proof is shifted to the appendix.

\begin{figure}[t!]\centering
\begin{tabular}{ p{0.5cm} | p{8cm}|  p{1.9cm}| p{1.25cm}}
$\S$ &Type of net $\tau$ & Complexity & Quantity
\\ \hline
\rowcolor[gray]{.8} 1 & $\{\nop,\res,\set,\swap\}\cup \omega$ with $\omega \subseteq \{\inp,\out,\used,\free\}$ & NP-complete & 16 \\ \hline
\rowcolor[gray]{.8}  2 & $\{\nop,\res,\swap\}\cup \omega$ with $\omega \subseteq \{\inp,\out,\used,\free\}$, \newline $\{\nop,\set,\swap\}\cup \omega$ with $\omega \subseteq \{\inp,\out,\used,\free\}$ & NP-complete & 32 \\ \hline
\rowcolor[gray]{.8}  3 & $\{\nop,\res,\set\}\cup \omega $ with $\omega\subseteq \{\inp,\out,\used,\free\}$, \newline $\{\nop,\out,\res\}\cup\omega$ with $\omega\subseteq \{\inp,\used,\free\}$, \newline $\{\nop,\inp,\set\}\cup\omega$ with $\omega\subseteq \{\out,\used,\free\}$ & NP-complete & 32 \\ \hline
4 & $\{\nop,\res\}\cup\omega$ with $\omega\subseteq \{\inp,\used,\free\}$,\newline $\{\nop,\set\}\cup\omega$ with $\omega\subseteq \{\out,\used,\free\}$ & polynomial & 16 \\ \hline
5 & $\{\nop,\inp,\out\}$ and $\{\nop,\inp,\out,\used\}$ & NP-complete & 2 \\ \hline 
\rowcolor[gray]{.8}  6 & $\{\nop,\inp,\out,\free\}$ and $\{\nop,\inp,\out,\used, \free\}$,\newline $\{\nop, \inp\}\cup\omega$ and $\{\nop, \out\}\cup\omega$ with $\omega\subseteq\{\used,\free\}$ & NP-complete & 10 \\ \hline 
7 & $\{\nop, \swap\}\cup\omega$ with $\omega\subseteq \{\inp,\out, \used, \free\}$, \newline $\{\nop\}\cup\omega$ with $\omega\subseteq \{\used, \free\}$ & polynomial & 20\\ \hline
8 & $\omega\subseteq \{\inp,\out,\res,\set,\used,\free\}$ & polynomial & 64\\ \hline
\rowcolor[gray]{.8}  9 & $\{\swap\}\cup \omega$ with $\omega\subseteq \{\inp,\out,\res,\set,\used,\free\}$ and $\omega\cap\{\res,\set,\used,\free\}\not=\emptyset$& NP-complete & 60\\ \hline
\rowcolor[gray]{.8}  10 &  $\{\swap\}\cup\omega$ with $\omega\subseteq\{\inp,\out\}$  & polynomial & 4\\
\end{tabular}
\caption{Overview of the computational complexity of $\tau$-SSP for Boolean types of nets $\tau$.
The gray area highlights the new results that this paper provides.}
\label{fig:overview}
\end{figure}

\section{Preliminaries}
\label{sec:Preliminaries}
In this section, we introduce necessary notions and definitions, supported by illustrations and examples, and some basic results that are used throughout the paper.

\textbf{Transition Systems.}
A (finite, deterministic) \emph{transition system} (TS, for short) $A=(S,E, \delta)$ is a directed labeled graph with the set of nodes $S$ (called \emph{states}), the set of labels $E$ (called \emph{events}) and partial \emph{transition function} $\delta: S\times E \longrightarrow S$. 
If $\delta(s,e)$ is defined, we say that $e$ \emph{occurs} at state $s$, denoted by $s\edge{e}$. 
By $s\edge{e}s'\in A$, we denote $\delta(s,e)=s'$.
This notation extends to paths, i.e., $q_0\edge{e_1}\dots\edge{e_n}q_n\in A$ denotes $q_{i-1}\edge{e_{i}}q_{i}\in A$ for all $i\in \{1,\dots, n\}$.
A TS $A$ is \emph{loop-free}, if $s\edge{e}s'\in A$ implies $s\not=s'$.
A loop-free TS $A$ is \emph{bi-directed} if $s'\edge{e}s\in A$ implies $s\edge{e}s'\in A$.
We say $s_0\fbedge{e_1}\dots \fbedge{e_n}s_n \in A$ is a simple bi-directed path if $s_i\not=s_j$ for all $i\not=j$ with $i,j \in \{0,\dots,n\}$.
An \emph{initialized} TS $A=(S,E,\delta, \iota)$ is a TS with a distinct \emph{initial} state $\iota\in S$, where every state $s\in S$ is \emph{reachable} from $\iota$ by a directed labeled path.
%We often refer to the ingredients of a TS $A$ by $S_A, E_A,\delta_A$ and $\iota_A$.
%
\begin{figure}[t!]
\begin{minipage}{1.0\textwidth}
\centering
\begin{tabular}{c|c|c|c|c|c|c|c|c}
$x$ & $\nop(x)$ & $\inp(x)$ & $\out(x)$ & $\res(x)$  & $\set(x)$ & $\swap(x)$ & $\used(x)$ & $\free(x)$\\ \hline
$0$ & $0$ & & $1$ & $0$ & $1$ & $1$ & & $0$\\
$1$ & $1$ & $0$ & & $0$ & $1$ & $0$ & $1$ & \\
\end{tabular}
\caption{
All interactions $i$ of $I$.
If a cell is empty, then $i$ is undefined on the respective $x$.
}\label{fig:interactions}
\end{minipage}

\begin{minipage}{1.0\textwidth}
\begin{center}
\begin{tikzpicture}[scale = 1.2]
\begin{scope}%nop, out, res, free, swap

\coordinate (1) at (1,0) ;
\foreach \i in {1} {\fill[red!20, rounded corners] (\i) +(-0.3,-0.4) rectangle +(0.45,0.4);}

\node (0) at (0,0) {\nscale{$0$}};
\node (1) at (1,0) {\nscale{$1$}};

\path (0) edge [->, out=-120,in=120,looseness=5] node[left, align =left] {\escale{$\nop$}} (0);
\path (1) edge [<-, out=60,in=-60,looseness=5] node[right, align=left] {\escale{$\nop$}  } (1);

\path (0) edge [<-] node[below] {\escale{$\inp$}} (1);
\node () at (0.0,-0.6) {\nscale{$ \tau $}};

\end{scope}
\begin{scope}[xshift=3cm]

\node (0) at (0,-0.6) {\nscale{$s_0$}};
\node (1) at (0,0.6) {\nscale{$s_1$}};
%\draw[-latex](-0.4,-0.6)to node[]{}(0);
\draw[-latex,out=60,in=-60](0)to node[right]{$a$}(1);
\draw[-latex,out=240,in=120](1)to node[left]{$a$}(0);
\node () at (0.5,-0.6) {\nscale{$ A_1$}};
\end{scope}
\begin{scope}[xshift=4.7cm]

\node (0) at (0,-0.6) {\nscale{$r_0$}};
\node (1) at (0,0.6) {\nscale{$r_1$}};
%\draw[-latex](-0.4,-0.6)to node[]{}(0);
\draw[-latex,out=60,in=-60](0)to node[right]{$b$}(1);
\draw[-latex,out=120,in=240](0)to node[left]{$c$}(1);
\node () at (0.5,-0.6) {\nscale{$ A_2$}};
\end{scope}
\begin{scope}[xshift=6.7cm]%nop, out, res, free, swap
\coordinate (1) at (1.5,0) ;
\foreach \i in {1} {\fill[red!20, rounded corners] (\i) +(-0.3,-0.4) rectangle +(0.45,0.4);}
\node (0) at (0,0) {\nscale{$0$}};
\node (1) at (1.5,0) {\nscale{$1$}};

\path (0) edge [->, out=-120,in=120,looseness=5] node[left, align =left] {\escale{$\nop$} } (0);
\path (1) edge [<-, out=60,in=-60,looseness=5] node[right, align=left] {\escale{$\nop$} \\ \escale{\used} \\ \escale{\set} } (1);

\path (0) edge [<-, bend right= 30] node[below] {\escale{$ \swap$}\  } (1);
\path (0) edge [->, bend left= 30] node[above] {\escale{$\set,\swap$}} (1);

\node () at (1.5,-0.6) {\nscale{$\tilde{\tau}$}};
\end{scope}
\end{tikzpicture}
\end{center}
\end{minipage}
\caption{%
Left: $ \tau = \{ \nop, \inp \}$. 
Right: $ \tilde{\tau} = \{ \nop,  \set, \swap, \used \}$.
The red colored area emphasizes the inside.
The only SSP atom of $A_1$ is $(s_0, s_1)$.
It is $\tilde{\tau} $-solvable by $R_1=(sup_1,sig_1)$ with $sup_1(s_0) = 0$, $sup_1(s_1) = 1$, $sig_1(a) = \swap$. 
Thus, $A_1$ has the $\tilde\tau$-separative set $\mathcal{R}=\{R_1\}$.
The SSP atom $(s_0,s_1)$ is not $\tau$-solvable.
The only SSP atom $(r_0,r_1)$ in $A_2$ can be solved by $\tilde{\tau} $-region $R_2=(sup_2, sig_2)$ with $sup_2(r_0) = 0$, $sup_2(r_1) = 1$, $sig_2(b) = \set$, $sig_2(c) = \swap$. 
Thus, $A_2$ has the $\tilde{\tau}$-SSP. 
The same atom can also be solved by $\tau$-region $R_3 = (sup_3,sig_3)$ with $sup_3(r_0) = 1$, $sup_3(r_1) = 0$, $sig_3(b) = sig_3(c) = \inp$. 
Hence, $A_2$ has the $\tau$-SSP, as well. 
}\label{fig:ex_types}
\begin{minipage}{1.0\textwidth}
\centering
\begin{tikzpicture}[new set = import nodes]
\begin{scope}[nodes={set=import nodes}]
%\node (0) at (0,0) {\nscale{$\iota$}};
\foreach \i in {0,...,3} {\node (\i) at (\i*1.1cm,0) {\nscale{$s_\i$}}; } 
%\draw[-latex](0,-0.4)to node[]{}(0);
\graph {
	(import nodes);
			0->["\escale{$a$}"]1->["\escale{$b$}"]2->["\escale{$c$}"]3;
			};
\node () at (1.8,-0.5) {\nscale{$ A_3$}};
\end{scope}
\begin{scope}[xshift=4.2cm,nodes={set=import nodes}]
\foreach \i in {0,...,3} {\coordinate (\i) at (\i*1.1cm,0); } 
\foreach \i in {0,1,3} {\fill[red!20] (\i) circle (0.2cm);}
\foreach \i in {0,1,3} {\node (\i) at (\i) {\nscale{$1$}}; } 
\node (2) at (2) {\nscale{$0$}}; 
\graph {
	(import nodes);
			0->["\escale{$\used$}"]1->["\escale{$\swap$}"]2->["\escale{$\set$}"]3;
			};
\node () at (1.8,-0.5) {\nscale{$ A_3^R$}};
\end{scope}
\begin{scope}[xshift=8.5cm,nodes={set=import nodes}]
\foreach \i in {0,...,3} {\coordinate (\i) at (\i*1.1cm,0); } 
\foreach \i in {0,1,3} {\fill[red!20] (\i) circle (0.2cm);}
\foreach \i in {0,1,3} {\node (\i) at (\i) {\nscale{$1$}}; } 
\node (2) at (2) {\nscale{$0$}}; 
\graph {
	(import nodes);
			0->["\escale{$\nop$}"]1->["\escale{$\swap$}"]2->["\escale{$\set$}"]3;
			};
\node () at (1.8,-0.5) {\nscale{$ A_3^{R'}$}};
\end{scope}
\end{tikzpicture}
\end{minipage}
\caption{%
Left: TS $A_3$, a simple directed path.
If $\tilde{\tau}$ is defined as in Figure~\ref{fig:ex_types}, then $sup(\iota)=1$, $sig(a)=\used$, $sig(b)=\swap$ and $sig(c)=\set$ implicitly defines the $\tilde{\tau}$-region $R=(sup, sig)$ of $A_3$ as follows:
$sup(s_1)=\delta_{\tilde{\tau}}(1,\used)=1$, $sup(s_2)=\delta_{\tilde{\tau}}(1,\swap)=0$ and $sup(s_3)=\delta_{\tilde{\tau}}(0,\set)=1$.
Middle: The image $A_3^R$ of $A_3$ (under R).
One easily verifies that $\delta_{A_3}(s,e)=s'$ implies $\delta_{\tilde{\tau}}(sup(s), sig(e))=sup(s')$, cf. Figure~\ref{fig:ex_types}.
In particular, $R$ is sound. 
For event $b$, the edge defined by $\delta_{A_3}(s_1,b)=s_2$ is mapped into $\delta_{\tilde{\tau}}(1,\swap)=0$ under $R$, i.e., $b$ makes a state change on the path; similar for $sig(c) = \set$. 
$R$ is not normalized, since $sup(s_0)=sup(s_1)$, but $sig(a)=\used \neq \nop$.
Right: The image $A_3^{R'}$ of $A_3$ under the normalized $\tilde{\tau}$-region $R'$ that is similar to $R$ but replaces $\used$ by $\nop$.
}\label{fig:image_of_path}
\end{figure}

\textbf{Boolean Types of Nets~\cite{DBLP:series/txtcs/BadouelBD15}.}
The following notion of Boolean types of nets allows to capture \emph{all} Boolean Petri nets in a \emph{uniform} way.
A \emph{Boolean type of net} $\tau=(\{0,1\},E_\tau,\delta_\tau)$ is a TS such that $E_\tau$ is a subset of the \emph{Boolean interactions}:
$E_\tau \subseteq I = \{\nop, \inp, \out, \res, \set, \swap, \used, \free\}$. 
Each interaction $i \in I$ is a binary partial function $i: \{0,1\} \rightarrow \{0,1\}$ as defined in Figure~\ref{fig:interactions}.
For all $x\in \{0,1\}$ and all $i\in E_\tau$, the transition function of $\tau$ is defined by $\delta_\tau(x,i)=i(x)$.
Notice that a type $\tau$ is completely determined by $E_\tau$.
Hence we often identify $\tau$ with $E_\tau$, cf. Figure~\ref{fig:overview}.
Moreover, since $I$ captures all meaningful Boolean interactions~\cite[p.~617]{DBLP:conf/tamc/TredupR19} and $\tau$ is defined by $E_\tau\subseteq I$, there are 256 Boolean types of nets at all.
For a Boolean type of net $\tau$, we say that its state $1$ is \emph{inside} and $0$ is \emph{outside}.
An interaction $i\in E_\tau$ \emph{exits} if $1\edge{i}0$, \emph{enters} if $0\edge{i}1$, \emph{saves 1} if $1\edge{i}1$ and \emph{saves 0} if $0\edge{i}0$. %EE: maybe better to define as $i(1) = 0$, etc ?
Accordingly, we group interactions together by $\exit=\{\inp,\res,\swap \}$, $\enter=\{\out,\set,\swap\}$, $\saveone=\{\nop,\set,\used\}$, $\savezero=\{\nop,\res,\free\}$ and $\save=\saveone\cup\savezero$.

For a net of type $\tau$ ($\tau$-net), the interactions of $\tau$ determine relations between \emph{places} and \emph{transitions} of the net. 
For instance, if a place $p$ and a transition $t$ are related via $\inp$, then $p$ has to be marked (\emph{true}) to allow $t$ to fire, and becomes unmarked (\emph{false}) after the firing (cf.~Figure~\ref{fig:interactions}). 
Since we are only concerned with state separation, we omit the formal definition of $\tau$-nets and rather refer to, e.g.,~\cite{DBLP:series/txtcs/BadouelBD15} for a comprehensive introduction to the topic.

\textbf{$\tau$-Regions.}
The following notion of $\tau$-regions is the key concept for state separation.
A $\tau$-region $R=(sup, sig)$ of TS $A=(S, E, \delta, \iota)$ consists of the mappings \emph{support} $sup: S \rightarrow \{0,1\}$ and \emph{signature} $sig: E \rightarrow E_\tau$, such that for every edge $s \edge{e} s'$ of $A$, the edge $sup(s) \ledge{sig(e)} sup(s')$ belongs to the type $\tau$;
we also say $sup$ \emph{allows} ($sig$ and thus the region) $R$.
If $P=q_0\edge{e_1}\dots \edge{e_n}q_n$ is a path in $A$, then $P^R=sup(q_0)\ledge{sig(e_1)}\dots\ledge{sig(e_n)}sup(q_n)$ is a path in $\tau$.
We say $P^R$ is the \emph{image} of $P$ (under $R$). 
For region $R$ and path $P$, event $e_i$ with $1 \le i \le n$ is called \emph{state changing} on the path $P^R$, if $sup(q_{i-1}) \neq sup(q_i)$ for $q_{i-1}\edge{e_i} q_i$ in $P$. 
Notice that $R$ is \emph{implicitly} defined by $sup(\iota)$ and $sig$:
Since $A$ is reachable, for every state $s\in S$, there is a path $\iota\edge{e_1}\dots \edge{e_n}s_n$ such that $s=s_n$.
Thus, since $\tau$ is deterministic, we inductively obtain $sup(s_{i+1})$ by $sup(s_{i})\lledge{sig(e_{i+1})}sup(s_{i+1})$ for all $i\in \{0,\dots, n-1\}$ and $s_0 = \iota$.
Hence, we can compute $sup$ and thus $R$ purely from $sup(\iota)$ and $sig$, cf.~Figure~\ref{fig:image_of_path}. 
If $\nop\in\tau$, then a $\tau$-region $R=(sup, sig)$ of a TS $A$ is called \emph{normalized} if $sig(e)=\nop$ for as many events $e$ as $sup$ allows: for all $e\in E$, $sig(e)\not\in \{\used,\free\}$ and if $sig(e)\in \exit\cup\enter$ then there is $s\edge{e}s'\in A$ such that $sup(s)\not=sup(s')$.  

\textbf{$\tau$-State Separation Property.}
A pair $(s, s')$ of distinct states of $A$ defines a \emph{state separation atom} (SSP atom).
A $\tau$-region $R=(sup, sig)$ \emph{solves} $(s,s')$ if $sup(s)\not=sup(s')$.
If $R$ exists, then $(s,s')$ is called \emph{$\tau$-solvable}.
If $s\in S_A$ and, for all $s'\in S_A\setminus\{s\}$, the atom $(s,s')$ is $\tau$-solvable, then $s$ is called $\tau$-solvable.
A TS has the \emph{$\tau$-state separation property} ($\tau$-SSP) if all of its SSP atoms are $\tau$-solvable.
A set $\mathcal{R}$ of $\tau$-regions of $A$ is called $\tau$-\emph{separative} if for each SSP atom of $A$ there is a $\tau$-region $R$ in $\mathcal{R}$ that solves it.
By the next lemma, if $\nop\in \tau$, then $A$ has the $\tau$-SSP if and only if it has a $\tau$-separative set of normalized $\tau$-regions:
\begin{lemma}\label{lem:irrelevance_of_test}
Let $A$ be a TS and $\tau$ be a \nop-equipped Boolean type of nets.
There is a $\tau$-separative set $\mathcal{R}$ of $A$ if and only if there is a $\tau$-separative set of normalized $\tau$-regions of $A$.
\end{lemma}
\begin{proof}
The \emph{if}-direction is trivial.
\emph{Only-if}:
Let $R=(sup, sig)$ be a non-normalized $\tau$-region, i.e., there is $e\in E_A$ such that $s\edge{e}s'\in A$ implies $sup(s)=sup(s')$. 
Since $\tau$ is $\nop$-equipped, $sup(s)\edge{\nop}sup(s')\in \tau$ for all $s\edge{e}s'\in A$. 
Thus, a $\tau$-region $R'=(sup, sig')$ can be constructed from $R$, where $sig'$ is equal to $sig$ except for $sig'(e)=\nop$. 
Since~$E_A$ is finite, a normalized region can be obtained from $R$ by inductive application of this procedure. 
\end{proof}
By the following lemma, $\tau$-SSP and $\tilde{\tau}$-SSP are equivalent if $\tau$ and $\tilde{\tau}$ are isomorphic:
\begin{lemma}[Without proof]
\label{lem:isomorphic_types}
If $\tau$ and $\tilde{\tau}$ are isomorphic types of nets then a TS $A$ has the $\tau$-SSP if and only if $A$ has the $\tilde{\tau}$-SSP.
\end{lemma}
In this paper, we consider the $\tau$-SSP also as decision problem that asks whether a given TS $A$ has the $\tau$-SSP.
The decision problem $\tau$-SSP is in NP:
By definition, $A$ has at most $\lvert S \rvert^2$ SSP atoms.
Hence, a Turing-machine can (non-deterministically) guess a $\tau$-separative set $\mathcal{R}$ such that $\vert \mathcal{R}\vert \leq \vert S\vert^2 $ and (deterministically) check in polynomial time its validity if it exists.

In what follows, some of our NP-completeness results base on polynomial-time reductions of the following decision problem, which is known to be NP-complete~\cite{DBLP:journals/dcg/MooreR01}:

\textbf{Cubic Monotone 1-in-3 3Sat.}~(\textsc{CM 1-in-3 3Sat})
The input is a Boolean formula $\varphi=\{\zeta_0,\dots, \zeta_{m-1}\}$ of negation-free three-clauses $\zeta_i=\{X_{i_0}, X_{i_1}, X_{i_2}\}$, where $i\in \{0,\dots, m-1\}$, with set of variables $X=\bigcup_{i=0}^{m-1}\zeta_i$;  
every variable $v\in X$ occurs in exactly three clauses, implying $\vert X \vert=m$.
The question to decide is whether there is a (one-in-three model) $M\subseteq X$ satisfying $\vert M\cap \zeta_i\vert =1$ for all $i\in \{0,\dots, m-1\}$.
\begin{example}[\textsc{CM 1-in-3 3Sat}]\label{ex:varphi}
The instance $\varphi=\{\zeta_0,\dots, \zeta_5\}$ of \textsc{CM 1-in-3 3Sat} with set of variables $X=\{X_0,\dots, X_5\}$ and clauses $\zeta_0=\{X_0,X_1,X_2\}$, $\zeta_1=\{X_0,X_2,X_3\}$, $\zeta_2=\{X_0,X_1,X_3\}$, $\zeta_3=\{X_2,X_4,X_5\}$, $\zeta_4=\{X_1,X_4,X_5\}$ and $\zeta_5=\{X_3,X_4,X_5\}$ has the one-in-three model $M=\{X_0,X_4\}$.
\end{example}

%%%%%%%%%%%%%%%%%%%%%%%%%%%%%%%%%%%%%%%%%%%%%%%%%%
\section{Deciding the State Separation Property for \nop-equipped Types}\label{sec:nop_equipped}%
%%%%%%%%%%%%%%%%%%%%%%%%%%%%%%%%%%%%%%%%%%%%%%%%%%

In this section, we investigate the computational complexity of \nop-equipped Boolean types of nets.
For technical reasons, we separately consider the types that include neither $\res$ nor $\set$ (\S5 - \S7 in Figure~\ref{fig:overview}) and the ones that have at least one of them (\S1 - \S4~in Figure~\ref{fig:overview}). 

First of all, the fact that $\tau$-SSP is polynomial for the types of~\S7 in Figure~\ref{fig:overview} is implied by the results of \cite{DBLP:conf/stacs/Schmitt96}\cite[p.~619]{DBLP:conf/tamc/TredupR19}. 
Moreover, for the types of Figure~\ref{fig:overview}~\S5 the NP-completeness of $\tau$-SSP has been shown in~\cite{DBLP:conf/apn/TredupRW18} ($\tau=\{\nop,\inp,\out\}$), and in~\cite{DBLP:conf/apn/Tredup19} ($\tau=\{\nop,\inp,\out,\used\}$, there referred to as $1$-bounded P/T-nets).
Thus, in order to complete the complexity characterization for the \nop-equipped types that neither contain \res\ nor \set, it only remains to ascertain the complexity of $\tau$-SSP for the types Figure~\ref{fig:overview}~\S6.
The following Subsection~\ref{sec:setres-free} proves that $\tau$-SSP is NP-complete for these types.

Then, we proceed with the types of~\S1 - \S4 in Figure~\ref{fig:overview}.
The fact that $\tau$-SSP is polynomial for the types of~\S 4 follows from~\cite[p.~619]{DBLP:conf/tamc/TredupR19}. 
The NP-completeness of $\tau$-SSP for the remaining types (\S1 - \S3) will be demonstrated in Subsection~\ref{sec:setres-derivatives}.

%%%%%%%%%%%%%%%%%%%%%%%%%%%%%%
\subsection{Complexity of $\tau$-SSP for \nop-equipped Types without $\res$ and $\set$}%
%%%%%%%%%%%%%%%%%%%%%%%%%%%%%%
\label{sec:setres-free}

The following theorem summarizes the complexity for the types of~\S5 - \S7 in Figure~\ref{fig:overview}. 
\begin{theorem}\label{the:res_set_free}
Let $\tau$ be a \nop-equipped Boolean type of nets such that $\tau\cap\{\res,\set\}=\emptyset$.
The $\tau$-SSP is NP-complete if $\tau\cap\{\inp,\out\}\not=\emptyset$ and $\swap\not\in\tau$, otherwise it is polynomial.
\end{theorem}

As just discussed in Section~\ref{sec:nop_equipped}, to complete the proof of Theorem~\ref{the:res_set_free} it remains to characterize the complexity of $\tau$-SSP for the types of Figure~\ref{fig:overview}~\S6.
Since $\tau$-SSP is NP-complete if $\tau=\{\nop,\inp,\out\}$ \cite{DBLP:conf/apn/TredupRW18}, by Lemma~\ref{lem:irrelevance_of_test}, the $\tau$-SSP is also NP-complete if $\tau=\{\nop,\inp,\out,\free\}$ and $\tau=\{\nop,\inp,\out,\free,\used\}$.

Thus, in what follows, we restrict ourselves to the types $\tau=\{\nop,\inp\}\cup\omega$ and $\tau=\{\nop,\out\}\cup\omega$, where $\omega\subseteq \{\used,\free\}$, and argue that their $\tau$-SSP are NP-complete.
To do so, we let $\tau=\{\nop,\inp\}$ and show the hardness of $\tau$-SSP by a reduction of \textsc{CM 1-in-3 3Sat}.
By Lemma~\ref{lem:irrelevance_of_test}, this also implies the hardness of $\tau\cup\omega$-SSP, where $\omega\subseteq \{\used,\free\}$.
Furthermore, by Lemma~\ref{lem:isomorphic_types}, the latter shows the NP-completeness of $\tau$-SSP if $\tau=\{\nop,\out\}\cup\omega$ and $\omega\subseteq \{\used,\free\}$.
The following paragraph introduces the intuition of our reduction approach.

\textbf{The Roadmap of Reduction.}
Let $\tau=\{\nop,\inp\}$ and $\varphi=\{\zeta_0,\dots, \zeta_{m-1}\}$ be an input of  \textsc{CM 1-in-3 3Sat} with variables $X=\{X_0,\dots, X_{m-1}\}$ and clauses $\zeta_i=\{X_{i_0}, X_{i_1}, X_{i_2}\}$ for all $i\in \{0,\dots, m-1\}$.
To show the NP-completeness of $\tau$-SSP, we reduce a given input $\varphi$ to a TS $A_\varphi$ (i.e., an input of $\tau$-SSP) as follows:
For every clause $\zeta_i=\{X_{i_0}, X_{i_1}, X_{i_2}\}$, the TS $A_\varphi$ has a directed labeled path $P_i$ that represents $\zeta_i$ by using its variables as events:
\begin{center} 
\begin{tikzpicture}[new set = import nodes]
\begin{scope}[nodes={set=import nodes}]% 
		\foreach \i in {0,...,4} { \coordinate (\i) at (\i*1.7cm,0) ;}
		\foreach \i in {0,...,3} { \node (\i) at (\i) {\nscale{$t_{i,\i}$}};}
		\node (p) at (-0.75cm, 0) {$P_i=$};
		\graph {
	(import nodes);
			0 ->["\escale{$X_{i_0}$}"]1->["\escale{$X_{i_1}$}"]2->["\escale{$X_{i_2}$}"]3;
		};
\end{scope}
\end{tikzpicture}
\end{center}
We ensure by construction that $A_\varphi$ has an SSP atom $\alpha$ such that if $R=(sup, sig)$ is a $\tau$-region solving $\alpha$, then $sup(t_{i,0})=1$ and $sup(t_{i,3})=0$ for all $i\in \{0,\dots, m-1\}$.
Thus, for all $i\in \{0,\dots, m-1\}$, the path $P_i^R$ is a path from $1$ to $0$ in $\tau$.
First, this obviously implies that there is an event $e\in \{X_{i_0}, X_{i_1}, X_{i_2}\}$ such that $sig(e)=\inp$.
Second, it is easy to see that there is no path in $\tau$ on which $\inp$ occurs twice, cf. Figure~\ref{fig:ex_types}.
The following figure sketches all possibilities of $P_i^R$, i.e., $sig(X_{i_0})=\inp$ and $sig(X_{i_1})=sig(X_{i_2})=\nop$, $sig(X_{i_1})=\inp$ and $sig(X_{i_0})=sig(X_{i_2})=\nop$, and $sig(X_{i_2})=\inp$ and $sig(X_{i_0})=sig(X_{i_1})=\nop$, respectively:
\begin{center} 
\begin{tikzpicture}[new set = import nodes]
\begin{scope}[nodes={set=import nodes}]% 
		\foreach \i in {0,...,3} { \coordinate (\i) at (\i*1.1cm,0) ;}
		\foreach \i in {0} {\fill[red!20] (\i) circle (0.2cm);}
		\foreach \i in {0} { \node (\i) at (\i) {\nscale{$1$}};}
		\foreach \i in {1,...,3} { \node (\i) at (\i) {\nscale{$0$}};}
		\graph {
	(import nodes);
			0 ->["\escale{$\inp$}"]1->["\escale{$\nop$}"]2->["\escale{$\nop$}"]3;
		};
\end{scope}
\begin{scope}[xshift=4.3cm,nodes={set=import nodes}]% 
		\foreach \i in {0,...,3} { \coordinate (\i) at (\i*1.1cm,0) ;}
		\foreach \i in {0,1} {\fill[red!20] (\i) circle (0.2cm);}
		\foreach \i in {0,1} { \node (\i) at (\i) {\nscale{$1$}};}
		\foreach \i in {2,3} { \node (\i) at (\i) {\nscale{$0$}};}
		\graph {
	(import nodes);
			0 ->["\escale{$\nop$}"]1->["\escale{$\inp$}"]2->["\escale{$\nop$}"]3;
		};
\end{scope}
\begin{scope}[xshift=8.6cm,nodes={set=import nodes}]% 
		\foreach \i in {0,...,3} { \coordinate (\i) at (\i*1.1cm,0) ;}
		\foreach \i in {0,1,2} {\fill[red!20] (\i) circle (0.2cm);}
		\foreach \i in {0,1,2} { \node (\i) at (\i) {\nscale{$1$}};}
		\foreach \i in {3} { \node (\i) at (\i) {\nscale{$0$}};}
		\graph {
	(import nodes);
			0 ->["\escale{$\nop$}"]1->["\escale{$\nop$}"]2->["\escale{$\inp$}"]3;
		};
\end{scope}
\end{tikzpicture}
\end{center}
Hence, the event $e$ is unique.
Since this is simultaneously true for all paths $P_0,\dots, P_{m-1}$, the set $M=\{e\in X \mid sig(e)=\inp\}$ selects exactly one variable per clause and thus defines a one-in-three model of $\varphi$.
Altogether, this approach shows that if $A_\varphi$ has the $\tau$-SSP, which implies that $\alpha$ is $\tau$-solvable, then $\varphi$ has a one-in-three model.

Conversely, our construction ensures that if $\varphi$ has a one-in-three model, then $\alpha$ and the other separation atoms of $A_\varphi$ are $\tau$-solvable, that is, $A_\varphi$ has the $\tau$-SSP.

\textbf{The Reduction of $A_\varphi$ for $\tau=\{\nop,\inp\}$.}
In the following, we introduce the announced TS $A_\varphi$, cf. Figure~\ref{fig:reduction_example}.
The initial state of $A_\varphi$ is $t_{0,0}$.
First of all, the TS $A_\varphi$ has the following path $P$ that provides the announced SSP atom $\alpha=(t_{m,0}, t_{m+1,0})$:
\begin{center} 
\begin{tikzpicture}[new set = import nodes]
\begin{scope}[nodes={set=import nodes}]% 
		\foreach \i in {0,...,7} { \coordinate (\i) at (\i*1.65cm,0) ;}
		\node (0) at (0) {\nscale{$t_{0,0}$}};
		\node (1) at (1) {$\dots$};
		\node (2) at (2) {\nscale{$t_{i,0}$}};
		\node (3) at (3) {$\dots$};
		\node (4) at (4) {\nscale{$t_{m-1,0}$}};
		\node (5) at (5) {\nscale{$t_{m,0}$}};
		\node (6) at (6) {\nscale{$t_{m+1,0}$}};
		\node (7) at (7) {\nscale{$\top$}};
		\graph {
	(import nodes);
			0 ->["\escale{$w_0$}"]1->["\escale{$w_{i-1}$}"]2->["\escale{$w_i$}"]3->["\escale{$w_{m-2}$}"]4->["\escale{$w_{m-1}$}"]5->["\escale{$k$}"]6->["\escale{$v$}"]7;
		};
\end{scope}
\end{tikzpicture}
\end{center}
Moreover, for every $i\in \{0,\dots, m-1\}$, the TS $A_\varphi$ has the following path $T_i$ that uses the variables of $\zeta_i=\{X_{i_0}, X_{i_1}, X_{i_2}\}$ as events and provides the subpath $P_i=t_{i,0}\edge{X_{i_0}}\dots\edge{X_{i_2}}t_{i,3}$:
\begin{center} 
\begin{tikzpicture}[new set = import nodes]
\begin{scope}[nodes={set=import nodes}]% 
		\foreach \i in {0,...,4} { \coordinate (\i) at (\i*1.7cm,0) ;}
		\foreach \i in {0,...,3} { \node (\i) at (\i) {\nscale{$t_{i,\i}$}};}
		\node (4) at (4) {\nscale{$\top$}};
		\graph {
	(import nodes);
			0 ->["\escale{$X_{i_0}$}"]1->["\escale{$X_{i_1}$}"]2->["\escale{$X_{i_2}$}"]3->["\escale{$u_i$}"]4;
		};
\end{scope}
\end{tikzpicture}
\end{center}
Finally, the TS $A_\varphi$ has, for all $i\in \{0,\dots, m-1\}$, the following path $G_i$:
\begin{center} 
\begin{tikzpicture}[new set = import nodes]
\begin{scope}[nodes={set=import nodes}]% 
		\foreach \i in {0,...,3} { \coordinate (\i) at (\i*1.7cm,0) ;}
		\foreach \i in {1,...,3} {\pgfmathparse{int(\i-1)} \node (\i) at (\i) {\nscale{$g_{i,\pgfmathresult}$}};}
		\node (0) at (0) {\nscale{$t_{0,0}$}};
		\graph {
	(import nodes);
			0 ->["\escale{$y_i$}"]1->["\escale{$u_i$}"]2->["\escale{$k$}"]3;
		};
\end{scope}
\end{tikzpicture}
\end{center}
\begin{figure}[t!]%full reduction example
\begin{center} 
\begin{tikzpicture}[new set = import nodes]

\begin{scope}[nodes={set=import nodes}]% T_0
		\foreach \i in {0,...,4} { \coordinate (\i) at (\i*1.4cm,0) ;}
		\foreach \i in {0} {\fill[red!20] (\i) circle (0.35cm);}
		\foreach \i in {0} { \node (t0\i) at (\i) {\nscale{$t_{0,\i}$}};}
		\foreach \i in {1,...,3} { \node (\i) at (\i) {\nscale{$t_{0,\i}$}};}
		\node (top) at (4*1.5cm, 3*1.2cm) {\nscale{$\top$}};
		\graph {
	(import nodes);
			t00 ->["\escale{$X_0$}"]1->["\escale{$X_1$}"]2->["\escale{$X_2$}"]3->[swap, bend right=25,"\escale{$u_0$}"]top;
		};
\end{scope}
\begin{scope}[yshift=1.2cm, nodes={set=import nodes}]% T_1
		\foreach \i in {0,...,3} { \coordinate (\i) at (\i*1.4cm,0) ;}
		\foreach \i in {0} {\fill[red!20] (\i) circle (0.35cm);}
		\foreach \i in {0} { \node (t1\i) at (\i) {\nscale{$t_{1,\i}$}};}
		\foreach \i in {1,...,3} { \node (\i) at (\i) {\nscale{$t_{1,\i}$}};}
		%\node (4) at (4) {\nscale{$\top$}};
		\graph {
	(import nodes);
			t10 ->["\escale{$X_0$}"]1->["\escale{$X_2$}"]2->["\escale{$X_3$}"]3->[swap, bend right=10,"\escale{$u_1$}"]top;
		};
\end{scope}
\begin{scope}[yshift=2.4cm, nodes={set=import nodes}]% T_2
		\foreach \i in {0,...,3} { \coordinate (\i) at (\i*1.4cm,0) ;}
		\foreach \i in {0} {\fill[red!20] (\i) circle (0.35cm);}
		\foreach \i in {0} { \node (t2\i) at (\i) {\nscale{$t_{2,\i}$}};}
		\foreach \i in {1,...,3} { \node (\i) at (\i) {\nscale{$t_{2,\i}$}};}
		%\node (4) at (4) {\nscale{$\top$}};
		\graph {
	(import nodes);
			t20 ->["\escale{$X_0$}"]1->["\escale{$X_1$}"]2->["\escale{$X_3$}"]3->["\escale{$u_2$}"]top;
		};
\end{scope}
\begin{scope}[yshift=3.6cm, nodes={set=import nodes}]% T_3
		\foreach \i in {0,...,3} { \coordinate (\i) at (\i*1.4cm,0) ;}
		\foreach \i in {0,1} {\fill[red!20] (\i) circle (0.35cm);}
		\foreach \i in {0} { \node (t3\i) at (\i) {\nscale{$t_{3,\i}$}};}
		\foreach \i in {1,...,3} { \node (\i) at (\i) {\nscale{$t_{3,\i}$}};}
		%\node (4) at (4) {\nscale{$\top$}};
		\graph {
	(import nodes);
			t30 ->["\escale{$X_2$}"]1->["\escale{$X_4$}"]2->["\escale{$X_5$}"]3->["\escale{$u_3$}"]top;
		};
\end{scope}
\begin{scope}[yshift=4.8cm, nodes={set=import nodes}]% T_4
		\foreach \i in {0,...,3} { \coordinate (\i) at (\i*1.4cm,0) ;}
		\foreach \i in {0,1} {\fill[red!20] (\i) circle (0.35cm);}
		\foreach \i in {0} { \node (t4\i) at (\i) {\nscale{$t_{4,\i}$}};}
		\foreach \i in {1,...,3} { \node (\i) at (\i) {\nscale{$t_{4,\i}$}};}
		%\node (4) at (4) {\nscale{$\top$}};
		\graph {
	(import nodes);
			t40 ->["\escale{$X_1$}"]1->["\escale{$X_4$}"]2->["\escale{$X_5$}"]3->["\escale{$u_4$}"]top;
		};
\end{scope}
\begin{scope}[yshift=6cm, nodes={set=import nodes}]% T_5
		\foreach \i in {0,...,3} { \coordinate (\i) at (\i*1.4cm,0) ;}
		\foreach \i in {0,1} {\fill[red!20] (\i) circle (0.35cm);}
		\foreach \i in {0} { \node (t5\i) at (\i) {\nscale{$t_{5,\i}$}};}
		\foreach \i in {1,...,3} { \node (\i) at (\i) {\nscale{$t_{5,\i}$}};}
		%\node (4) at (4) {\nscale{$\top$}};
		\graph {
	(import nodes);
			t50 ->["\escale{$X_3$}"]1->["\escale{$X_4$}"]2->["\escale{$X_5$}"]3->[bend left =10,"\escale{$u_5$}"]top;
		};
\end{scope}
\begin{scope}[yshift=7.2cm, nodes={set=import nodes}]% 
		\foreach \i in {0,1} { \coordinate (\i) at (\i*4cm,0) ;}
		\foreach \i in {0} {\fill[red!20] (\i) circle (0.35cm);}
		\node (t60) at (0) {\nscale{$t_{6,0}$}};
		\node (1) at (1) {\nscale{$t_{7,0}$}};
		\graph {
	(import nodes);
			t60 ->["\escale{$k$}"]1->[bend left =40,"\escale{$v$}"]top;
		};
\end{scope}
\begin{scope}[yshift=7.2cm, nodes={set=import nodes}]% 
		\graph {
	(import nodes);
			t00 ->["\escale{$w_0$}"]t10 ->["\escale{$w_1$}"]t20->["\escale{$w_2$}"]t30 ->["\escale{$w_3$}"]t40 ->["\escale{$w_4$}"]t50 ->["\escale{$w_5$}"]t60 ;
		};
\end{scope}
%%%%%%%%%%%%%%%%%%%%%%%%%%%%%%%%%%%%%%%%%%%%%%%%%
\begin{scope}[nodes={set=import nodes}]% G_0
		\foreach \i in {0,...,2} { \coordinate (\i) at (-\i*1.2cm-2.5cm,0) ;}
		\foreach \i in {0,1} {\fill[red!20] (\i) circle (0.35cm);}
		\foreach \i in {0} { \node (g0\i) at (\i) {\nscale{$g_{0,\i}$}};}
		\foreach \i in {1,2} { \node (\i) at (\i) {\nscale{$g_{0,\i}$}};}
		\graph {
	(import nodes);
			g00 ->[swap, "\escale{$u_0$}"]1->[swap, "\escale{$k$}"]2;
		};
\end{scope}
\begin{scope}[yshift=1.2cm,nodes={set=import nodes}]% G_1
		\foreach \i in {0,...,2} { \coordinate (\i) at (-\i*1.2cm-2.5cm,0) ;}
		\foreach \i in {0,1} {\fill[red!20] (\i) circle (0.35cm);}
		\foreach \i in {0} { \node (g1\i) at (\i) {\nscale{$g_{1,\i}$}};}
		\foreach \i in {1,2} { \node (\i) at (\i) {\nscale{$g_{1,\i}$}};}
		\graph {
	(import nodes);
			g10 ->[swap, "\escale{$u_1$}"]1->[swap, "\escale{$k$}"]2;
		};
\end{scope}
\begin{scope}[yshift=2.4cm,nodes={set=import nodes}]% G_2
		\foreach \i in {0,...,2} { \coordinate (\i) at (-\i*1.2cm-2.5cm,0) ;}
		\foreach \i in {0,1} {\fill[red!20] (\i) circle (0.35cm);}
		\foreach \i in {0} { \node (g2\i) at (\i) {\nscale{$g_{2,\i}$}};}
		\foreach \i in {1,2} { \node (\i) at (\i) {\nscale{$g_{2,\i}$}};}
		\graph {
	(import nodes);
			g20 ->[swap, "\escale{$u_2$}"]1->[swap, "\escale{$k$}"]2;
		};
\end{scope}
\begin{scope}[yshift=3.6cm,nodes={set=import nodes}]% G_3
		\foreach \i in {0,...,2} { \coordinate (\i) at (-\i*1.2cm-2.5cm,0) ;}
		\foreach \i in {0,1} {\fill[red!20] (\i) circle (0.35cm);}
		\foreach \i in {0} { \node (g3\i) at (\i) {\nscale{$g_{3,\i}$}};}
		\foreach \i in {1,2} { \node (\i) at (\i) {\nscale{$g_{3,\i}$}};}
		\graph {
	(import nodes);
			g30 ->[swap, "\escale{$u_3$}"]1->[swap, "\escale{$k$}"]2;
		};
\end{scope}
\begin{scope}[yshift=4.8cm,nodes={set=import nodes}]% G_4
		\foreach \i in {0,...,2} { \coordinate (\i) at (-\i*1.2cm-2.5cm,0) ;}
		\foreach \i in {0,1} {\fill[red!20] (\i) circle (0.35cm);}
		\foreach \i in {0} { \node (g4\i) at (\i) {\nscale{$g_{4,\i}$}};}
		\foreach \i in {1,2} { \node (\i) at (\i) {\nscale{$g_{4,\i}$}};}
		\graph {
	(import nodes);
			g40 ->[swap, "\escale{$u_4$}"]1->[swap, "\escale{$k$}"]2;
		};
\end{scope}
\begin{scope}[yshift=6cm,nodes={set=import nodes}]% G_5
		\foreach \i in {0,...,2} { \coordinate (\i) at (-\i*1.2cm-2.5cm,0) ;}
		\foreach \i in {0,1} {\fill[red!20] (\i) circle (0.35cm);}
		\foreach \i in {0} { \node (g5\i) at (\i) {\nscale{$g_{5,\i}$}};}
		\foreach \i in {1,2} { \node (\i) at (\i) {\nscale{$g_{5,\i}$}};}
		\graph {
	(import nodes);
			g50 ->[swap, "\escale{$u_5$}"]1->[swap, "\escale{$k$}"]2;
		};
\end{scope}
\begin{scope}[yshift=6cm,nodes={set=import nodes}]% G_5
		
		\graph {
	(import nodes);
			t00 ->[pos=0.85,"\escale{$y_0$}"]g00;%->[swap, "\escale{$k$}"]2;
		};
\end{scope}
\coordinate (c0) at (-1.7,0);
\coordinate (c01) at (-1.7,1.2);
\draw[->] (t00)--(c0)--(c01)--(g10) node [pos=0.4, above] {\escale{$y_1$}};
\coordinate (c1) at (-1.5,0);
\coordinate (c11) at (-1.5,2.4);
\draw[->] (t00)--(c1)--(c11)--(g20) node [pos=0.4, above] {\escale{$y_2$}};
\coordinate (c2) at (-1.3,0);
\coordinate (c21) at (-1.3,3.6);
\draw[->] (t00)--(c2)--(c21)--(g30) node [pos=0.4, above] {\escale{$y_3$}};
\coordinate (c3) at (-1.1,0);
\coordinate (c31) at (-1.1,4.8);
\draw[->] (t00)--(c3)--(c31)--(g40) node [pos=0.4, above] {\escale{$y_4$}};
\coordinate (c4) at (-0.9,0);
\coordinate (c41) at (-0.9,6);
\draw[->] (t00)--(c4)--(c41)--(g50) node [pos=0.4, above] {\escale{$y_5$}};

\end{tikzpicture}
\end{center}
\caption{The TS $A_\varphi$ originating from the input $\varphi$ of Example~\ref{ex:varphi}, which has the one-in-three model $M=\{X_0,X_4\}$.
The colored area sketches the region $R_M$ of Lemma~\ref{lem:res_set_free_model_implies_ssp} that solves $\alpha=(t_{6,0}, t_{7,0})$.}\label{fig:reduction_example}
\end{figure}
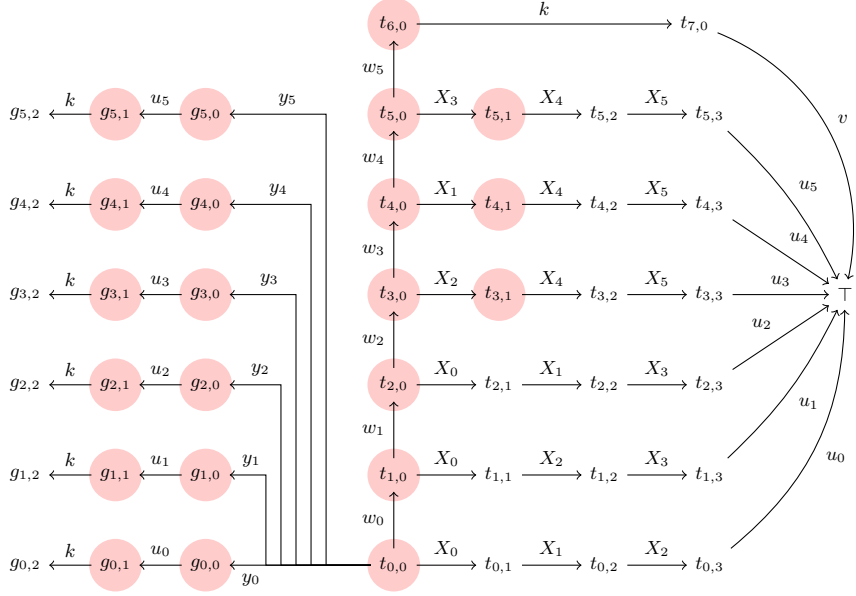
Notice that the paths $P$ and $T_0,\dots, T_{m-1}$ have the same \enquote{final}-state $\bot$.
Obviously, the size of $A_\varphi$ is polynomial in the size of $\varphi$.
The following Lemma~\ref{lem:res_set_free_ssp_implies_model} and Lemma~\ref{lem:res_set_free_model_implies_ssp} prove the validity of our reduction and thus complete the proof of Theorem~\ref{the:res_set_free}.
\begin{lemma}\label{lem:res_set_free_ssp_implies_model}
Let $\tau=\{\nop,\inp\}$.
If $A_\varphi$ has the $\tau$-SSP, then $\varphi$ has a one-in-three model.
\end{lemma}
\begin{proof}
Let $R=(sup, sig)$ be a $\tau$-region that solves $\alpha$, that is, $sup(t_{m,0})\not=sup(t_{m+1,0})$.
By $t_{m,0}\edge{k}t_{m+1,0}$ and $\tau=\{\nop,\inp\}$, this implies $sig(k)=\inp$, $sup(t_{m,0})=1$ and $sup(t_{m+1,0})=0$.
If $s_0\edge{e_0}\dots\edge{e_{i}}s_i\edge{k}s_{i+1}\edge{e_{i+2}}\dots\edge{e_n}s_n$ is a path in $A_\varphi$, then, by $sig(k)=\inp$, we get $sup(s_j)=1$ for all $j\in \{0,\dots, i\}$, $sup(s_j)=0$ for all $j\in \{i+1,\dots, n\}$ and $sig(e_j)=\nop$ for all $j\in \{1,\dots, i, i+2,\dots, n\}$.
This implies $sup(\top)=0$, $sig(v)=\nop$ as well as $sig(u_i)=\nop$ and $sup(t_{i,0})=1$ for all $i\in \{0,\dots, m-1\}$.
Furthermore, by $sup(\top)=0$ and $sig(u_i)=\nop$, we get $sup(t_{i,3})=0$ for all $i\in \{0,\dots, m-1\}$.
Hence, for all $i\in \{0,\dots, m-1\}$, the image $P_i^R$ of $P_i$ is a path from $1$ to $0$ in $\tau$.
Hence, as just discussed above, $M=\{e\in X\mid sig(e)=\inp\}$ selects exactly one variable per clause and defines a one-in-three model of $\varphi$.
\end{proof}

\begin{lemma}\label{lem:res_set_free_model_implies_ssp}
Let $\tau=\{\nop,\inp\}$.
If $\varphi$ has a one-in-three model, then $A_\varphi$ has the $\tau$-SSP.
\end{lemma}
\begin{proof}
Let $M$ be a one-in-three model of $\varphi$.

The following region $R_1=(sup, sig)$ solves $(s,s')$ for all $s\in \bigcup_{i=0}^{m-1} S(P_i)$ and all $s'\in \bigcup_{i=0}^{m-1}S(G_i)$, where $s'\not=t_{0,0}$:
$sup(t_{0,0})=1$;
for all $e\in E_{A_\varphi}$, if $e\in \{y_0,\dots, y_{m-1}\}$, then $sig(e)=\inp$, otherwise $sig(e)=\nop$.
Note Section~\ref{sec:res_set_free_model_implies_ssp}, Figure~\ref{fig:R_1} for an example of $R_1$.

Let $i\in \{0,\dots, m-1\}$ be arbitrary but fixed.
The following region $R^T_i=(sup, sig)$ solves $(s,s')$ for all $s\in \{t_{i,0}, \dots, t_{i,3}\}$ and all $s'\in \bigcup_{j=i+1}^{m-1}S(T_j)\cup\{t_{m,0}, t_{m+1,0}\}$:
$sup(t_{0,0})=1$;
for all $e\in E_{A_\varphi}$, if $e\in\{w_i\}\cup\{u_0,\dots, u_i\}\cup \{y_{i+1},\dots, y_{m-1}\}$, then $sig(e)=\inp$, otherwise $sig(e)=\nop$.
Note Section~\ref{sec:res_set_free_model_implies_ssp}, Figure~\ref{fig:RT_1} for an example of $R^T_1$.

The following region $R_2=(sup, sig)$ solves $(t_{m+1},\top)$ and $(g_{i,0}, g_{i,1})$ and $(g_{i,0}, g_{i,2})$ for all $i\in \{0,\dots, m-1\}$:
$sup(t_{0,0})=1$;
for all $e\in E_{A_\varphi}$, if $e\in \{v\}\cup \{u_0,\dots, u_{m-1}\}$, then $sig(e)=\inp$, otherwise $sig(e)=\nop$.
Note Section~\ref{sec:res_set_free_model_implies_ssp}, Figure~\ref{fig:R_2} for an example of $R_2$.

The following region $R_M=(sup, sig )$ uses the one-in-three model $M$ of $\varphi$ and solves $\alpha$ as well as $(g_{i,1}, g_{i,2})$ for all $i\in \{0,\dots, m-1\}$:
$sup(t_{0,0})=1$;
for all $e\in E_{A_\varphi}$, if $e\in \{k\}\cup M$, then $sig(e)=\inp$, otherwise $sig(e)=\nop$.
Note Figure~\ref{fig:reduction_example} for an example of $R_M$.

Let $i\in \{0,\dots, m-2\}$ be arbitrary but fixed.
The following region $R^G_i=(sup, sig)$ solves $(s, s')$ for all $s\in \{g_{i,0}, g_{i,1}, g_{i,2}\}$ and all $s'\in \bigcup_{j=i+1}^{m-1}S(G_j)$, where $s'\not=t_{0,0}$:
$sup(t_{0,0})=1$;
for all $e\in E_{A_\varphi}$, if $e\in \{y_{i+1}, \dots, y_{m-1}\}$, then $sig(e)=\inp$, otherwise $sig(e)=\nop$.
Note Section~\ref{sec:res_set_free_model_implies_ssp}, Figure~\ref{fig:RG_1} for an example of $R^G_1$.

By the arbitrariness of $i$ for $R^T_i$ and $R^G_i$, it remains to show that $t_{i,0},\dots, t_{i,3}$ are pairwise separable for all $i\in \{0,\dots, m-1\}$.
Let $i\in \{0,\dots, m-1\}$ be arbitrary but fixed.
We present only a region $R^X_{i_0}=(sup, sig)$ that solves $(t_{i,0}, s)$ for all $s\in \{t_{i,1}, t_{i,2}, t_{i,3}\}$.
It is then easy to see that the remaining atoms are similarly solvable. 
Let $j\not=\ell \in \{0,\dots, m-1\}\setminus\{i\} $ select the other two clauses of $\varphi$ that contain $X_{i_0}$, that is, $X_{i_0}\in \zeta_j\cap\zeta_\ell$.
$R^X_{i_0}=(sup, sig)$ is defined as follows:
$sup(t_{0,0})=1$;
for all $e\in E_{A_\varphi}$, if $e\in \{X_{i_0}\}\cup \{v\}\cup\{u_0,\dots,u_{m-1}\}\setminus\{u_i,u_j,u_\ell\}$, then $sig(e)=\inp$, otherwise $sig(e)=\nop$.
Note Section~\ref{sec:res_set_free_model_implies_ssp}, Figure~\ref{fig:RX_4_0} for an example of $R^X_{4,0}$.

Similarly, one gets regions where $X_{i_1}$ or $X_{i_2}$ has \inp-signature.
These regions solve the remaining atoms of $S(T_i)\setminus\{\top\}$.
Since $i$ was arbitrary, this completes the proof.
\end{proof}

%%%%%%%%%%%%%%%%%%%%%%%%%%%%%%%%%
\subsection{Complexity of $\tau$-SSP for \nop-equipped Types with $\res$ or $\set$.}%
%%%%%%%%%%%%%%%%%%%%%%%%%%%%%%%%%
\label{sec:setres-derivatives}

The next theorem states that $\tau$-SSP is NP-complete for the \nop-equipped types that have not yet been considered, cf~Figure~\ref{fig:overview}~\S1 - \S3.
Moreover, it summarizes the complexity of $\tau$-SSP for all types of~Figure~\ref{fig:overview}~\S1 - \S4: 
\begin{theorem}\label{the:res_set_derivatives}
Let $\tau$ and $\tilde{\tau}$ be Boolean type of nets and $\{\nop,\res\}\subseteq \tau$ and $\{\nop,\set\}\subseteq \tilde{\tau}$.
\begin{enumerate}
\item\label{the:res_set_derivatives_res}
The $\tau$-SSP is NP-complete if $\tau\cap \enter\not=\emptyset$, otherwise it is polynomial.
\item\label{the:res_set_derivatives_set}
The $\tilde{\tau}$-SSP is NP-complete if $\tilde{\tau}\cap \exit\not=\emptyset$, otherwise it is polynomial.
\end{enumerate}
\end{theorem}

In this section we complete the proof of Theorem~\ref{the:res_set_derivatives} as follows.
Firstly, we let $\tau_0=\{\nop,\inp,\out\}$ and, by a reduction of $\tau_0$-SSP, we show that the $\tau$-SSP is NP-complete if $\tau=\{\nop,\out,\res\}\cup \omega$ and $\omega\subseteq\{\inp,\used,\free\}$ or if $\{\nop,\res,\set\}\subseteq \tau$.
By Lemma~\ref{lem:isomorphic_types}, the former also implies the NP-completeness of $\tau$-SSP if $\tau=\{\nop,\inp,\set\}\cup \omega$ and $\omega\subseteq\{\out,\used,\free\}$.
Altogether, this proves the claim for all the types listed in \S1 and \S3 of Figure~\ref{fig:overview}.

Secondly, we let $\tau_1=\{\nop,\inp\}$ and reduce $\tau_1$-SSP to $\tau$-SSP, where $\tau=\{\nop,\res,\swap\}\cup\omega$ and $\omega\subseteq\{\inp,\used,\free\}$.
Again by Lemma~\ref{lem:isomorphic_types}, this also implies the NP-completeness of $\tau$-SSP if $\tau=\{\nop,\set,\swap\}\cup\omega$ and $\omega\subseteq\{\out,\used,\free\}$.
Hence, this proves the claim for all the types listed in \S2 of Figure~\ref{fig:overview} and thus completes the proof of Theorem~\ref{the:res_set_derivatives}.

For the announced reductions, we use the following extensions of a TS $A$, cf. Figure~\ref{fig:extensions}.
Let $A=(S_A,E_A,\delta_A,\iota_A)$ be a loop-free TS, and let $\overline{E_A}=\{\overline{e}\mid e\in E_A\}$ be the set containing for every event $e\in E_A$ the unambiguous and fresh event $\overline{e}$ that is \emph{associated with} $e$.
\emph{The backward-extension} $B=(S_A, E_A\cup \overline{E_A}, \delta_B, \iota_A)$ of $A$ extends $A$ by $\overline{E_A}$ and additional backward edges: for all $e\in E_A$ and all $s,s'\in S_A$, if $\delta_A(s,e)=s'$, then $\delta_B(s,e)=s'$ and $\delta_B(s',\overline{e})=s$.
The \emph{oneway loop-extension} $C=(S_A, E_A\cup\overline{E_A}, \delta_C,\iota_A)$ of a TS $A$ extends $B$ by some additional loops: for all $x\in E_A\cup \overline{E_A}$ and all $s\in S_A$, we define $\delta_C(s,e)=\delta_B(s,e)$ and, for all $e\in E_A$ and all $s,s'\in S_A$, if $\delta_A(s,e)=s'$, then $\delta_C(s',e)=s'$.
Finally, the \emph{loop-extension} $D=(S_A, E_A\cup\overline{E_A}, \delta_D,\iota_A)$ of $A$ is an extension of $C$, where for all $x\in E_A\cup \overline{E_A}$ and all $s\in S_A$, we define $\delta_D(s,x)=\delta_C(s,x)$ and, for all $e\in E_A$ and all $s,s'\in S_A$, if $\delta_A(s,e)=s'$, then $\delta_D(s,\overline{e})=s$.

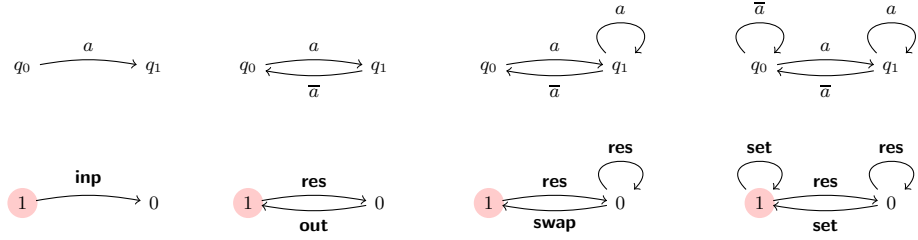
\begin{figure}[t!]
\begin{center} 
\begin{tikzpicture}[new set = import nodes]
\begin{scope}[nodes={set=import nodes}]% 
		\foreach \i in {0,1} { \coordinate (\i) at (\i*1.75cm,0) ;}
		\foreach \i in {0,1} { \node (\i) at (\i) {\nscale{$q_\i$}};}
		\graph {
	(import nodes);
			0 ->[bend left =10,"\escale{$a$}"]1;
		};
\end{scope}
\begin{scope}[xshift=3cm,nodes={set=import nodes}]% 
		\foreach \i in {0,1} { \coordinate (\i) at (\i*1.75cm,0) ;}
		\foreach \i in {0,1} { \node (\i) at (\i) {\nscale{$q_\i$}};}
		\graph {
	(import nodes);
			0 ->[bend left =10,"\escale{$a$}"]1;
			0 <-[bend right =10,swap, "\escale{$\overline{a}$}"]1;
		};
\end{scope}
\begin{scope}[xshift=6.2cm,nodes={set=import nodes}]% 
		\foreach \i in {0,1} { \coordinate (\i) at (\i*1.75cm,0) ;}
		\foreach \i in {0,1} { \node (\i) at (\i) {\nscale{$q_\i$}};}
		\path (1) edge [<-, in=135, out=45,looseness=5] node[above] {\nscale{$a$} } (1);
		\graph {
	(import nodes);
			0 ->[bend left =10,"\escale{$a$}"]1;
			0 <-[bend right =10,swap, "\escale{$\overline{a}$}"]1;
		};
\end{scope}
\begin{scope}[xshift=9.8cm,nodes={set=import nodes}]% 
		\foreach \i in {0,1} { \coordinate (\i) at (\i*1.75cm,0) ;}
		\foreach \i in {0,1} { \node (\i) at (\i) {\nscale{$q_\i$}};}
		\path (1) edge [<-, in=135, out=45,looseness=5] node[above] {\nscale{$a$} } (1);
		\path (0) edge [<-, in=135, out=45,looseness=5] node[above] {\nscale{$\overline{a}$} } (0);
		\graph {
	(import nodes);
			0 ->[bend left =10,"\escale{$a$}"]1;
			0 <-[bend right =10,swap, "\escale{$\overline{a}$}"]1;
		};
\end{scope}
\begin{scope}[yshift=-1.8cm]
\begin{scope}[nodes={set=import nodes}]% 
		\foreach \i in {0,1} { \coordinate (\i) at (\i*1.75cm,0) ;}
		\foreach \i in {0} {\fill[red!20] (\i) circle (0.2cm);}
		\foreach \i in {0,1} {\pgfmathparse{int(1-\i)} \node (\i) at (\i) {\nscale{$\pgfmathresult$}};}
		\graph {
	(import nodes);
			0 ->[bend left =10,"\escale{$\inp$}"]1;
		};
\end{scope}
\begin{scope}[xshift=3cm,nodes={set=import nodes}]% 
		\foreach \i in {0,1} { \coordinate (\i) at (\i*1.75cm,0) ;}
		\foreach \i in {0} {\fill[red!20] (\i) circle (0.2cm);}
		\foreach \i in {0,1} {\pgfmathparse{int(1-\i)} \node (\i) at (\i) {\nscale{$\pgfmathresult$}};}
		\graph {
	(import nodes);
			0 ->[bend left =10,"\escale{$\res$}"]1;
			0 <-[bend right =10,swap, "\escale{$\out$}"]1;
		};
\end{scope}
\begin{scope}[xshift=6.2cm,nodes={set=import nodes}]% 
		\foreach \i in {0,1} { \coordinate (\i) at (\i*1.75cm,0) ;}
		\foreach \i in {0} {\fill[red!20] (\i) circle (0.2cm);}
		\foreach \i in {0,1} {\pgfmathparse{int(1-\i)} \node (\i) at (\i) {\nscale{$\pgfmathresult$}};}
		\path (1) edge [<-, in=135, out=45,looseness=5] node[above] {\nscale{$\res$} } (1);
		\graph {
	(import nodes);
			0 ->[bend left =10,"\escale{$\res$}"]1;
			0 <-[bend right =10,swap, "\escale{$\swap$}"]1;
		};
\end{scope}
\begin{scope}[xshift=9.8cm,nodes={set=import nodes}]% 
		\foreach \i in {0,1} { \coordinate (\i) at (\i*1.75cm,0) ;}
		\foreach \i in {0} {\fill[red!20] (\i) circle (0.2cm);}
		\foreach \i in {0,1} {\pgfmathparse{int(1-\i)} \node (\i) at (\i) {\nscale{$\pgfmathresult$}};}
		\path (1) edge [<-, in=135, out=45,looseness=5] node[above] {\nscale{$\res$} } (1);
		\path (0) edge [<-, in=135, out=45,looseness=5] node[above] {\nscale{$\set$} } (0);
		\graph {
	(import nodes);
			0 ->[bend left =10,"\escale{$\res$}"]1;
			0 <-[bend right =10,swap, "\escale{$\set$}"]1;
		};
\end{scope}
\end{scope}
\end{tikzpicture}
\end{center}
\caption{Top, left to right: A TS $A$ consisting of a single edge; backward-extension $B$ of $A$;  oneway loop-extension $C$ of $A$; loop-extension $D$ of $A$.
Bottom, left to right: images of $A$ and its extensions $B,C,D$ under regions corresponding to the types of Lemma~\ref{lem:extensions} solving $(q_0,q_1)$: 
a $\{\nop,\inp\}$- ($\{\nop,\inp,\out\}$-) region of $A$;
a $\{\nop,\res, \out\}$-region of $B$;
a $\{\nop,\res, \swap\}$-region of $C$;
a $\{\nop,\res, \set\}$-region of $D$.
}\label{fig:extensions}
\end{figure}
Depending on the considered type $\tau$, we let $\tilde{\tau}=\tau_0$ or $\tilde{\tau}=\tau_1$ and reduce a loop-free TS $A=(S_A,E_A,\delta_A,\iota_A)$ either to its backward-, oneway loop- or loop-extension and show that $sup:S_A\rightarrow\{0,1\}$ allows a $\tilde{\tau}$-region of $A$ if and only if it allows a $\tau$-region of the extension:
\begin{lemma}\label{lem:extensions}
Let $\tau_0=\{\nop,\inp,\out\}$, $\tau_1=\{\nop,\inp\}$, $A$ 
a loop-free TS, $sup: S_A\rightarrow\{0,1\}$ and $B$, $C$ and $D$ the backward-, oneway loop- and loop-extension of $A$, respectively.
\begin{enumerate}
\item\label{lem:extensions_nop_res_out}
If $\tau=\{\nop,\out,\res\}\cup\omega$ and $\omega\subseteq \{\inp,\used,\free\}$, then $sup$ allows a $\tau_0$-region $R=(sup, sig)$ of $A$ if and only if it allows a normalized $\tau$-region $R'=(sup, sig')$ of $B$.
\item\label{lem:extensions_nop_res_set}
If $\tau\supseteq\{\nop,\res,\set\}$, then $sup$ allows a $\tau_0$-region $R=(sup, sig)$ of $A$ if and only if it allows a normalized $\tau$-region $R'=(sup, sig')$ of $D$.
\item\label{lem:extensions_nop_res_swap}
If $\tau=\{\nop,\res,\swap\}\cup\omega$ and $\omega\subseteq \{\inp,\used,\free\}$, then $sup$ allows a $\tau_1$-region $R=(sup, sig)$ of $A$ if and only if it allows a normalized $\tau$-region $R'=(sup, sig')$ of $C$.
\end{enumerate}
\end{lemma}
\begin{proof}
(1):
\emph{Only-if}:
Let $R=(sup, sig)$ be a $\tau_0$-region of $A$.
Recall that $sig(e)=\nop$, $sig(e)=\inp$, $sig(e)=\out$ imply $sup(s)=sup(t)$, $sup(s)=1$ and $sup(t)=0$, $sup(s)=0$ and $sup(t)=1$ for all edges $s\edge{e}t$ of $A$, respectively.
Thus, it is easy to see that $R$ induces a normalized $\tau$-region $R'=(sup, sig')$ of $B$ as follows, cf. Figure~\ref{fig:extensions}:
For all $e\in E_A$ and its associated event $\overline{e}\in \overline{E_A}$, if $sig(e)=\nop$, then $sig'(e)=sig'(\overline{e})=\nop$;
if $sig(e)=\inp$, then $sig'(e)=\res$ and $sig'(\overline{e})=\out$;
if $sig(e)=\out$, then $sig'(e)=\out$ and $sig'(\overline{e})=\res$.

\emph{If}:
Let $R'=(sup, sig')$ be a normalized $\tau$-region of $B$, and let $e\in E_A$ and $s\edge{e}t\in B$ and $s'\edge{e}t'\in B$ be arbitrary but fixed.
First of all, we argue that if $sup(s)\not=sup(t)$, then $sup(s)=sup(s')$ and $sup(t)=sup(t')$:
By definition of $B$, we have that $t\edge{\overline{e}}s$ and $t'\edge{\overline{e}}s'$ are present. 
If $sup(s)=1$ and $sup(t)=0$, then $s\edge{e}t$ and $t\edge{\overline{e}}s$ imply $sig'(e)\in \{\inp,\res\}$ and $sig'(\overline{e})=\out$.
By $\edge{e}t'$ and $\edge{\overline{e}}s'$, this immediately implies $sup(s')=1$ and $sup(t')=0$.
Similarly, if $sup(s)=0$ and $sup(t)=1$, then $sig'(e)=\out$ and $sig'(\overline{e})\in \{\inp,\res\}$, which implies $sup(s')=0$ and $sup(t')=1$.
Consequently, since $s'\edge{e}t'$ was arbitrary, the claim follows. 
Note that this implies in return that if $sup(s)=sup(t)$, then $sup(s')=sup(t')$.
Since both edges were arbitrary, it is easy to see that the following $\tau_0$-region $R=(sup, sig)$ of $A$ is well defined:
for all $e\in E_A$, if there is $s\edge{e}t\in B$ such that $sup(s)\not=sup(t)$, then $sig(e)=\inp$ if $sup(s)=1$ and $sup(t)=0$, else $sig(e)=\out$;
otherwise, $sig(e)=\nop$.

(2):
\emph{Only-if}:
Recall that $0\edge{\res}0$ and $1\edge{\set}1$ are present in $\tau$.
Consequently, if $R=(sup, sig)$ is a $\tau_0$-region of $A$, then the region $R'=(sup, sig')$, similarly defined to the one for the Only-if-direction of~(1), but replacing \out\ by \set, is a $\tau$-region of $D$, cf. Figure~\ref{fig:extensions}.

\emph{If}:
If $R'=(sup, sig')$ is a normalized $\tau$-region of $D$, then it holds $sig'(e)\in \{\nop, \res,\set\}$ for all $e\in E_D$.
This is due to the fact that if $s\edge{x}s'\in D$, then $s'\edge{x}s'\in D$ for all $x\in E_D$ and all $s,t\in S_D$.
Thus, $R=(sup, sig)$ is obtained from $R'$ by the same arguments as the ones presented for the If-direction of~(1). 

(3):
\emph{Only-if}:
Let $R=(sup, sig)$ be a $\tau_1$-region of $A$.
Recall that $sig(e)=\nop$ and $sig(e)=\inp$ imply $sup(s)=sup(t)$ and $sup(s)=1, sup(t)=0$ for all edges $s\edge{e}t$ of $A$, respectively.
Moreover, $1\edge{\res}0$, $0\edge{\res}0$ and $0\edge{\swap}1$ are present in $\tau$.
Thus, we get a normalized $\tau$-region $R'=(sup, sig')$ of $C$ as follows, cf. Figure~\ref{fig:extensions}:
For all $e\in E_A$ and its associated event $\overline{e}\in \overline{E_A}$, if $sig(e)=\nop$, then $sig'(e)=sig'(\overline{e})=\nop$;
if $sig(e)=\inp$, then $sig'(e)=\res$ and $sig'(\overline{e})=\swap$.

\emph{If}:
Let $R'=(sup, sig')$ be a normalized $\tau$-region of $C$, and let $e\in E_A$ and $s\edge{e}t\in C$ and $s'\edge{e}t'\in C$ be arbitrary but fixed.
We argue that $sup(s)\not=sup(t)$ implies $sup(s)=sup(s')=1$ and $sup(t)=sup(t')=0$:
By definition of $C$ and $s\edge{e}t\in C$, we get $t\edge{e}t\in C$. 
Thus, $sig'(e)\in \{\nop,\res\}$.
Thus, if $sup(s)\not=sup(t)$, then $sig'(e)=\res$, which implies $sup(s)=1$ and $sup(t)=0$.
Moreover, by $t\edge{\overline{e}}s\in C$, this also implies $sig'(\overline{e})=\swap$.
Finally, by $sig'(e)=\res$, $\edge{e}t'$, $sig'(\overline{e})=\swap$ and $t'\edge{\overline{e}}s'$, we get $sup(t')=0$ and $sup(s')=1$.
Consequently, the following definition of $sig$ yields a well-defined $\tau_1$-region $R=(sup, sig)$ of $A$:
for all $e\in E_A$, if $sig'(e)=\res$, then $sig(e)=\inp$, otherwise $sig(e)=\nop$.
\end{proof}

Notice that the TS $A_\varphi$ of Section~\ref{sec:setres-free} is loop-free.
Furthermore, in \cite{DBLP:conf/apn/TredupRW18}, it has been shown that $\{\nop,\inp,\out\}$-SSP is NP-complete even if $A$ is a simple directed path.
Moreover, the introduced extensions of $A$ are constructible in polynomial time. 
Thus, by Lemma~\ref{lem:irrelevance_of_test} and Lemma~\ref{lem:isomorphic_types}, the following corollary, which is easily implied by Lemma~\ref{lem:extensions}, completes the proof of Theorem~\ref{the:res_set_derivatives}. 
\begin{corollary}[Without Proof]\label{cor:extensions}
Let $\tau_0=\{\nop,\inp,\out\}$ and $\tau_1=\{\nop,\inp\}$, and let $A$ be a loop-free TS and $B$, $C$ and $D$ its backward-, oneway loop- and loop-extension, respectively.
\begin{enumerate}
\item
If $\tau=\{\nop,\out,\res\}\cup\omega$ and $\omega\subseteq \{\inp,\used,\free\}$, then $A$ has the $\tau_0$-SSP if and only if $B$ has the $\tau$-SSP.
\item
If $\tau \supseteq \{\nop,\res,\set\}$, then $A$ has the $\tau_0$-SSP if and only if $D$ has the $\tau$-SSP.
\item
If $\tau=\{\nop,\res,\swap\}\cup\omega$ and $\omega\subseteq \{\inp,\used,\free\}$, then $A$ has the $\tau_1$-SSP if and only if $C$ has the $\tau$-SSP.
\end{enumerate}
\end{corollary}

%%%%%%%%%%%%%%%%%%%%%%%%%%%%%%%%%%%
\section{Deciding the State Separation Property for \nop-free Types}%
%%%%%%%%%%%%%%%%%%%%%%%%%%%%%%%%%%%
\label{sec:nop-free}

The following theorem summarizes the complexity of $\tau$-SSP for \nop-free Boolean types:
\begin{theorem}\label{the:nop_free}
Let $\tau$ be a \nop-free type of nets and $A$ a TS.
\begin{enumerate}
\item\label{the:nop_free_polynomial}
If $\swap\not\in \tau$ or $\swap\in\tau$ and $\tau\cap\save=\emptyset$, then deciding if $A$ has the $\tau$-SSP is polynomial.
\item\label{the:nop_free_hardness}
If $\swap\in \tau$ and $\tau\cap\save\not=\emptyset$, then deciding if $A$ has the $\tau$-SSP is NP-complete.
\end{enumerate}
\end{theorem}

The tractability of $\tau$-SSP for \nop-free types that are also \swap-free has been shown in the broader context of $\tau$-synthesis in \cite{DBLP:conf/apn/TredupE20}.
Thus, restricted to Theorem~\ref{the:nop_free}.\ref{the:nop_free_polynomial}, it remains to argue that $\tau$-SSP is polynomial if $\tau=\{\swap\}\cup\omega$ and $\omega\subseteq \{\inp,\out\}$.
The following lemma states that separable inputs of $\tau$-SSP are trivial for these types and thus proves its tractability:
\begin{lemma}\label{lem:nop_free_polynomial}
Let $\tau=\{\swap\}\cup\omega$, where $\omega\subseteq \{\inp,\out\}$, and $A=(S,E,\delta,\iota)$ be a TS.
If $A$ has the $\tau$-SSP, then it has at most two states.
\end{lemma}
\begin{proof}
If there is $s\edge{e}s\in A$, then $A$ has no $\tau$-regions.
Hence, if such $A$ has more than one state, then it does not have the $\tau$-SSP.

Let's consider the case when $A$ is loop-free, that is, $s\edge{e}s'\in A$ implies $s\not=s'$.
First of all, note that there is at most one outgoing edge $\iota\edge{e}s$ at $\iota$, since if $\iota\edge{e}s$, $\iota\edge{e'}s'$ and $s\not=s'$, then $s$ and $s'$ are not separable.
This can be seen as follows:
If $R=(sup, sig)$ is a $\tau$-region that solves $(s,s')$, then $sup(s)=0$ and $sup(s')=1$ or $sup(s)=1$ and $sup(s')=0$.
If $sup(s)=0$ and $sup(s')=1$, then $sup(\iota)=0$ contradicts $sig(e)\in \tau$, and $sup(\iota)=1$ contradicts $sig(e')\in \tau$.
Similarly, $sup(s)=1$ and $sup(s')=0$ yields a contradiction.
Thus, a separating region $R$ does not exist, which proves the claim.
Secondly, if $\iota\edge{e}s\edge{e'}s''\in A$, then $\iota=s''$, since otherwise, $\iota$ and $s''$ are not separable.
This can be seen as follows:
If $R=(sup, sig)$ is a $\tau$-region that solves $(\iota,s'')$, then $sup(\iota)=0$ and $sup(s'')=1$ or $sup(\iota)=1$ and $sup(s'')=0$.
If $sup(\iota)=0$ and $sup(s'')=1$, then $sup(s')=0$ contradicts $sig(e)\in\tau$, and $sup(s')=1$ contradicts $sig(e')\in\tau$.
Similarly, $sup(\iota)=1$ and $sup(s'')=0$ yields a contradiction.
This implies again that a separating region does not exist, which proves the claim and, moreover, proves the lemma.
\end{proof}

To complete the proof of Theorem~\ref{the:nop_free}, it remains to prove the NP-completeness of $\tau$-SSP for the types listed in \S9 of Figure~\ref{fig:overview}, which are exactly covered by Theorem~\ref{the:nop_free}.\ref{the:nop_free_hardness}.
Thus, in the remainder of this section, if not stated explicitly otherwise, we let $\tau$ be a \nop-free type such that $\swap\in \tau$ and $\tau\cap\save\not=\emptyset$.
Moreover, we reduce \textsc{CM 1-in-3 3Sat} to $\tau$-SSP again.

\textbf{Basic Ideas of the Reduction.}
Similar to our previous approach we build a TS $A_\varphi$ that has for every clause $\zeta_i=\{X_{i_0}, X_{i_1}, X_{i_2}\}$, where $i\in \{0,\dots, m-1\}$, a (bi-directed) path $ P_i = \dots \fbedge{X_{i_0}}\dots \fbedge{X_{i_1}} \dots \fbedge{X_{i_2}} \dots$ on which the elements of $\zeta_i$ occur as events.
Together, the corresponding paths $P_0,\dots, P_{m-1}$ are meant to represent $\varphi$.
However, the current types are more diverse than $\{\nop,\inp\}$, since they also allow $\swap$ and other interactions.
Simultaneously, they are more restricted than $\{\nop,\inp\}$, since they lack of \nop.
One of the main obstacles that occur is that if $s\edge{e}s'\in A_\varphi$, then the current types basically allow  $sup(s)=1$ and $sup(s')=0$ as well as $sup(s)=0$ and $sup(s')=1$ for a $\tau$-region $R=(sup, sig)$ that solves $(s,s')$.
It turns out that this requires a second representation of $\varphi$.
To do so, we use a copy $\varphi'$ that originates from $\varphi$ by simply renaming its variables.
That is, $\varphi'$ originates from $\varphi$ by replacing every variable $v\in X$ of $\varphi$ by a unique and fresh variable $v'$.
\begin{example}[Renaming of $\varphi$]\label{ex:varphi_copy}
The instance $\varphi'=\{\zeta_0',\dots, \zeta_5'\}$ that originates from $\varphi$ of Example~\ref{ex:varphi} is defined by $X'=\{X_0',\dots, X_5'\}$ and $\zeta_0'=\{X_0',X_1',X_2'\}$, $\zeta_1'=\{X_0',X_2',X_3'\}$, $\zeta_2'=\{X_0',X_1',X_3'\}$, $\zeta_3'=\{X_2',X_4',X_5'\}$, $\zeta_4'=\{X_1',X_4',X_5'\}$ and $\zeta_5'=\{X_3',X_4',X_5'\}$.
\end{example}
It is immediately clear that $\varphi$ is one-in-three satisfiable if and only if $\varphi'$ is one-in-three satisfiable.
The TS $A_\varphi$ additionally has for every clause $\zeta_i'=\{X_{i_0}', X_{i_1}', X_{i_2}'\}$, where $i\in \{0,\dots, m-1\}$, of $\varphi'$ also a (bi-directed) path $ P_i' = \dots \Fbedge{X_{i_0}'}\dots \Fbedge{X_{i_1}'} \dots \Fbedge{X_{i_2}'} \dots$ on which the elements of $\zeta_i'$ occur as events.
Moreover, by the construction, the TS $A_\varphi$ has a SSP atom $\alpha=(s,s')$ such that if a $\tau$-region solves $\alpha$, then either the signatures of the variable events $X$ of $\varphi$ define a one-in-three model of $\varphi$ or the signatures of the variable events $X'$ of $\varphi'$ define a one-in-three model of $\varphi'$.
Obviously, both cases imply the one-in-three satisfiability of $\varphi$.

Conversely, the construction ensures, if $\varphi$ has a one-in-three model then $A_\varphi$ has the~$\tau$-SSP.

Similar to our approach for Theorem~\ref{the:res_set_free}, the TS $A_\varphi$ is a composition of several gadgets.
The next Lemma~\ref{lem:swap_basics} introduces some basic properties of $\tau$-regions in bi-directed TS for \nop-free types that we use to prove the functionailty of $A_\varphi$'s gadgets.
After that, Lemma~\ref{lem:nop_free_ingredients} introduces TS that are the (isomorphic) prototypes of the gadgets of  $A_\varphi$ and additionally proves the essential parts of their intended functionality.
\begin{lemma}\label{lem:swap_basics}
Let $\tau$ be a \nop-free Boolean type of nets and $A$ a bi-directed TS; 
let $s\edge{e}s'$ be an edge of $A$, 
$P_0=s_0\fbedge{e_1}\dots\fbedge{e_m}s_m$ and $P_1=q_0\fbedge{e_1}\dots\fbedge{e_m}q_m$ be two simple paths of $A$ that both apply the same sequence $e_1\dots e_m$ of events, 
and $R=(sup, sig)$ be a $\tau$-region of $A$.
\begin{enumerate}
\item\label{lem:swap_basics_keep}
If $sig(e)\in\save$, then $sup(s)=sup(s')=sup(q)=sup(q')$ for every edge $q\edge{e}q'\in A$. 
\item\label{lem:swap_basics_swap}
If $sup(s_m)\not=sup(q_m)$, then $sig(e_i)=\swap$ for all $i\in \{1,\dots, m\}$.
\item\label{lem:swap_basics_even}
If $sup(s_0)=sup(s_m)$, then $\vert \{e\in \{e_1,\dots, e_m\}\mid sig(e)=\swap\}\vert$ is even. 
\end{enumerate}
\end{lemma}
\begin{proof}
(1):
$A$ is bi-directed and $\edge{i}p\in \tau$ and $\edge{i}p'\in \tau$ imply $p=p'$ for all $i\in \save$.
 
(2):
By definition of $\tau$, if $sup(s_m)\not=sup(q_m)$, then, by $\edge{e_m}s_m$ and $\edge{e_m}q_m$, we get $sig(e_m)=\swap$.
Clearly, by $sup(s_m)\not=sup(q_m)$, this implies $sup(s_{m-1})\not=sup(q_{m-1})$.
Thus, the claim follows easily by induction on $m$.

(3):
Since $sup(s_0)=sup(s_m)$, the image $P_0^R$ of $P_0$ is a path of $\tau$ that starts and terminates at the same state.
Consequently, the number of changes between $0$ and $1$ on $P_0^R$ is even.
Since $A$ is bi-directed, $sup(s)\not=sup(s')$ if and only if $sig(e)=\swap$ for all $s\edge{e}s'\in P_0$.
Thus, the number of events of $P_0$ with a \swap-signature must be even. 
Hence, the claim.
\end{proof}

\begin{lemma}[Basic Components of $A_\varphi$]\label{lem:nop_free_ingredients}
Let $\tau$ be a $\nop$-free Boolean type and $A$ a bi-directed TS with the following paths $G$, $F$, $T$ and $Q$, and let $R=(sup, sig)$ be a $\tau$-region of $A$: 
\begin{center}
\begin{tikzpicture}[new set = import nodes]
\begin{scope}%G_i
\node (T) at (-0.5,0) {\nscale{$G=$}};
\foreach \i in {0,...,4} {\coordinate (\i) at (\i*1.1cm+0.2cm,0);}
\foreach \i in {0,...,4} {\node (t\i) at (\i) {\nscale{$g_\i$}};}
\graph { 
(t0) <->["\escale{$v$}"] (t1) <->["\escale{$w$}"] (t2) <->["\escale{$k_0$}"] (t3) <->["\escale{$k_1$}"] (t4);
};
\end{scope}
\begin{scope}[xshift=6.5cm]%F_i
\node (T) at (-0.5,0) {\nscale{$F=$}};
\foreach \i in {0,...,4} {\coordinate (\i) at (\i*1.1cm+0.2cm,0);}
\foreach \i in {0,...,4} {\node (t\i) at (\i) {\nscale{$f_\i$}};}
\graph { 
(t0) <->["\escale{$v$}"] (t1) <->["\escale{$w$}"] (t2) <->["\escale{$k_1$}"] (t3) <->["\escale{$k_0$}"] (t4); 
};
\end{scope}
\begin{scope}
\begin{scope}[yshift=-1.2cm]%T
\node (T) at (-0.5,0) {\nscale{$T=$}};
\foreach \i in {0,...,9} {\coordinate (\i) at (\i*1.15cm,0);}
\coordinate (10) at (1.15cm*9, -1.2cm);
\coordinate (11) at (1.15cm*8, -1.2cm);
\foreach \i in {0,...,11} {\node (t\i) at (\i) {\nscale{$t_{\i}$}};}
\graph { 
(t0) <->["\escale{$k_0$}"] (t1) <->["\escale{$v_0$}"] (t2) <->["\escale{$v_{1}$}"] (t3) <->["\escale{$X_{0}$}"] (t4) <->["\escale{$v_{2}$}"] (t5)<->["\escale{$X_{1}$}"](t6)<->["\escale{$v_{3}$}"](t7)<->["\escale{$X_{2}$}"](t8)<->["\escale{$v_{4}$}"](t9)<->["\escale{$v_{5}$}"](t10)<->["\escale{$k_0$}"](t11);
};
\end{scope}
\begin{scope}[yshift=-2.4cm]%Q
\node (T) at (-0.5,0) {\nscale{$Q=$}};
\foreach \i in {0,...,3} {\coordinate (\i) at (\i*1.15cm,0);}
\foreach \i in {0,...,3} {\node (t\i) at (\i) {\nscale{$q_{\i}$}};}
\graph { 
(t0) <->["\escale{$X_0$}"] (t1) <->["\escale{$v_6$}"] (t2) <->["\escale{$X_2$}"] (t3);
};
\end{scope}
\end{scope}
\end{tikzpicture}
\end{center}

\begin{enumerate}
\item\label{lem:nop_free_ingredients_alpha}
If $R$ solves $\alpha=(g_2,g_4)$, then either $sig(k_0)\in \save $ and $sig(k_1)=\swap$ or $sig(k_0)=\swap$ and $sig(k_1)\in \save $;
\item\label{lem:nop_free_ingredients_swaps}
If $sig(k_0)\in \save $ and $sig(k_1)=\swap$ or $sig(k_0)=\swap$ and $sig(k_1)\in \save $, then $sig(v)=sig(w)=\swap$;
\item \label{lem:nop_free_ingredients_model}
If $sig(k_0)\in \save$ and $sig(v_i)=\swap$ for all $i\in \{0,\dots, 6\}$, then there is exactly one $X\in \{X_0, X_1,X_2\}$ such that $sig(X)\not=\swap$.
\end{enumerate}
\end{lemma} 
\begin{proof}
Since $A$ is bi-directed, we have $sig(e)\not\in \{\inp,\out\}$ for all $e\in E_A$.

(1): $R$ solves $\alpha$, thus $sup(g_2)\not=sup(g_4)$.
If $sig(k_0)\not=\swap\not=sig(k_1)$ or $sig(k_0)=sig(k_1)=\swap$, then $sup(g_2)=sup(g_4)$, a contradiction.
Hence the claim, cf. Figure~\ref{fig:nop_free_ingredients}.

(2):
If $sig(k_0)\in \save$ and $sig(k_1)=\swap$, then we have $sup(g_2)=sup(f_3)\not=sup(f_2)$, cf. Figure~\ref{fig:nop_free_ingredients}.
By symmetry, the latter is also true if $sig(k_0)=\swap$ and $sig(k_1)\in \save$.
Hence, the claim follows from Lemma~\ref{lem:swap_basics}.\ref{lem:swap_basics_swap}.

(3):
Since $sig(k_0)\in \save $, we get $sup(t_1)=sup(t_{10})$.
By Lemma~\ref{lem:swap_basics}.\ref{lem:swap_basics_even}, this implies that $\vert \{e\in E(T) \mid sig(e)=\swap\}\vert $ is even.
Moreover, since $sig(v_0)=\dots sig(v_5)=\swap$, this implies $\vert \{e\in \{X_0,X_1,X_2\} \mid sig(e)=\swap\}\vert \in \{0,2\}$.
If $\vert \{e\in \{X_0,X_1,X_2\} \mid sig(e)=\swap\}\vert =0$ then, we get $sup(t_3)=sup(t_4)\not=sup(t_5)=sup(t_6)\not=sup(t_7)=sup(t_8)$ by Lemma~\ref{lem:swap_basics}.\ref{lem:swap_basics_keep}.
This particularly implies $sup(t_4)=sup(t_7)$ and, again by Lemma~\ref{lem:swap_basics}.\ref{lem:swap_basics_keep}, also $sup(t_4)=sup(q_1)=sup(q_2)$ and contradicts $sig(v_6)=\swap$, cf. Figure~\ref{fig:nop_free_ingredients_counterexample}.
Thus, we have $\vert \{e\in \{X_0,X_1,X_2\} \mid sig(e)=\swap\}\vert =2$, which proves the claim, cf. Figure~\ref{fig:nop_free_ingredients}.
\end{proof}
\begin{figure}[t!]%Illustrating figures for Lemma~\ref{lem:nop_free_ingredients}
\begin{minipage}{\textwidth}
\begin{center}
\begin{tikzpicture}[new set = import nodes]
\begin{scope}%G_i
\node (T) at (-0.5,0) {\nscale{$G^R=$}};
\foreach \i in {0,...,4} {\coordinate (\i) at (\i*1.1cm+0.2cm,0);}
\foreach \i in {1,4} {\fill[red!20] (\i) circle (0.2cm);}
\foreach \i in {0,2,3} {\node (t\i) at (\i) {\nscale{$0$}};}
\foreach \i in {1,4} {\node (t\i) at (\i) {\nscale{$1$}};   }
\graph { 
(t0) <->["\escale{$\swap$}"] (t1) <->["\escale{$\swap$}"] (t2) <->["\escale{$\free$}"] (t3) <->["\escale{$\swap$}"] (t4);
};
\end{scope}
\begin{scope}[xshift=6.5cm]%F_i
\node (T) at (-0.5,0) {\nscale{$F^R=$}};
\foreach \i in {0,...,4} {\coordinate (\i) at (\i*1.1cm+0.2cm,0);}
\foreach \i in {0,2} {\fill[red!20] (\i) circle (0.2cm);}
\foreach \i in {1,3,4} {\node (t\i) at (\i) {\nscale{$0$}};}
\foreach \i in {0,2} {\node (t\i) at (\i) {\nscale{$1$}};   }
\graph { 
(t0) <->["\escale{$\swap$}"] (t1) <->["\escale{$\swap$}"] (t2) <->["\escale{$\swap$}"] (t3) <->["\escale{$\free$}"] (t4); 
};
\end{scope}
\begin{scope}
\begin{scope}[yshift=-1.2cm]%T
\node (T) at (-0.5,0) {\nscale{$T^R=$}};
\foreach \i in {0,...,9} {\coordinate (\i) at (\i*1.15cm,0);}
\coordinate (10) at (1.15cm*9, -1.2cm);
\coordinate (11) at (1.15cm*8, -1.2cm);
\foreach \i in {2,5,7,9} {\fill[red!20] (\i) circle (0.2cm);}
\foreach \i in {0,1,3,4,6,8,10,11} {\node (t\i) at (\i) {\nscale{$0$}};}
\foreach \i in {2,5,7,9} {\node (t\i) at (\i) {\nscale{$1$}};   }
\graph { 
(t0) <->["\escale{$\free$}"] (t1) <->["\escale{$\swap$}"] (t2) <->["\escale{$\swap$}"] (t3) <->["\escale{$\free$}"] (t4) <->["\escale{$\swap$}"] (t5)<->["\escale{$\swap$}"](t6)<->["\escale{$\swap$}"](t7)<->["\escale{$\swap$}"](t8)<->["\escale{$\swap$}"](t9)<->["\escale{$\swap$}"](t10)<->["\escale{$\free$}"](t11);
};
\end{scope}
\begin{scope}[yshift=-2.4cm]%Q
\node (T) at (-0.5,0) {\nscale{$Q^R=$}};
\foreach \i in {0,...,3} {\coordinate (\i) at (\i*1.15cm,0);}
\foreach \i in {2} {\fill[red!20] (\i) circle (0.2cm);}
\foreach \i in {0,1,3} {\node (t\i) at (\i) {\nscale{$0$}};}
\foreach \i in {2} {\node (t\i) at (\i) {\nscale{$1$}};   }
\graph { 
(t0) <->["\escale{$\free$}"] (t1) <->["\escale{$\swap$}"] (t2) <->["\escale{$\swap$}"] (t3);
};
\end{scope}
\end{scope}
\end{tikzpicture}
\end{center}
\caption{Illustrations for Lemma~\ref{lem:nop_free_ingredients}.
The images $G^R$, $F^R$, $T^R$ and $Q^R$, where $R=(sup, sig)$ is a $\{\swap,\free\}$-region that solves $(g_2,g_4)$ and satisfies $sig(k_0)=sig(X_0)=\free$ and $sig(k_1)=sig(v)=sig(w)=sig(X_1)=sig(X_2)=sig(v_i)=\swap$ for all $i\in \{0,\dots, 6\}$.}\label{fig:nop_free_ingredients}
\end{minipage}

\vspace{0.25cm}

\begin{minipage}{\textwidth}
\begin{center}
\begin{tikzpicture}[new set = import nodes]
\begin{scope}%T
\node (T) at (-0.5,0) {\nscale{$T^R=$}};
\foreach \i in {0,...,9} {\coordinate (\i) at (\i*1.15cm,0);}
\coordinate (10) at (1.15cm*9, -1.2cm);
\coordinate (11) at (1.15cm*8, -1.2cm);
\foreach \i in {2,5,6,9} {\fill[red!20] (\i) circle (0.2cm);}
\foreach \i in {0,1,3,4,7,8,10,11} {\node (t\i) at (\i) {\nscale{$0$}};}
\foreach \i in {2,5,6,9} {\node (t\i) at (\i) {\nscale{$1$}};   }
\graph { 
(t0) <->["\escale{$\free$}"] (t1) <->["\escale{$\swap$}"] (t2) <->["\escale{$\swap$}"] (t3) <->["\escale{$\free$}"] (t4) <->["\escale{$\swap$}"] (t5)<->["\escale{$\used$}"](t6)<->["\escale{$\swap$}"](t7)<->["\escale{$\free$}"](t8)<->["\escale{$\swap$}"](t9)<->["\escale{$\swap$}"](t10)<->["\escale{$\free$}"](t11);
};
\end{scope}
\end{tikzpicture}
\end{center}
\caption{Illustration for Lemma~\ref{lem:nop_free_ingredients}.
A $\{\swap,\used,\free\}$-region $R=(sup, sig)$, restricted to $T$, where $sig(k_0)=sig(X_0)=sig(X_2)=\free$, $sig(X_1)=\used$ and $sig(v_i)=\swap$ for all $i\in \{0,\dots,5\}$ and, hence, $\vert \{e\in \{X_0,X_1,X_2\}\mid sig(e)=\swap\}\vert =0$.
$R$ is not extendable to a region of a TS that has $T$ and $Q$ such that $sig(v_6)=\swap$, since $sig(v_6)=\swap$ would contradict $sup(q_0)=\dots=sup(q_3)=0$, which would be required by $sig(X_0)=sig(X_2)=\free$.}\label{fig:nop_free_ingredients_counterexample}
\end{minipage}
\end{figure}

Let $\varphi$ be an instance of \textsc{CM 1-in-3 3Sat} with the set of variables $X=\{X_0,\dots, X_{m-1}\}$ and $\varphi'$ its renamed copy with event set $X'=\{X_0',\dots, X_{m-1}'\}$.
In the following, we introduce the construction of $A_\varphi$.

Firstly, for every $i\in \{0,\dots, 7m-1\}$, the TS $A_i$ has the following gadgets $G_i$, $F_i$, $G_i'$ and $F_i'$ with starting states $g_{i,0}$, $f_{i,0}$, $g_{i,0}'$ and $f_{i,0}'$, respectively, providing the atom $\alpha=(g_{0,2}, g_{0,4})$:
\begin{center}
\begin{tikzpicture}[new set = import nodes]
\begin{scope}%G_i
\node (T) at (-0.4,0) {\nscale{$G_i=$}};
\foreach \i in {0,...,4} {\coordinate (\i) at (\i*1.15cm+0.2cm,0);}
\foreach \i in {0,...,4} {\node (t\i) at (\i) {\nscale{$g_{i,\i}$}};}
\graph { 
(t0) <->["\escale{$v_i$}"] (t1) <->["\escale{$w_i$}"] (t2) <->["\escale{$k_0$}"] (t3) <->["\escale{$k_1$}"] (t4);% <->["\escale{$k_1$}"] (t5);
};
\end{scope}
\begin{scope}[xshift=6.2cm]%F_i
\node (T) at (-0.4,0) {\nscale{$F_i=$}};
\foreach \i in {0,...,4} {\coordinate (\i) at (\i*1.15cm+0.2cm,0);}
\foreach \i in {0,...,4} {\node (t\i) at (\i) {\nscale{$f_{i,\i}$}};}
\graph { 
(t0) <->["\escale{$v_i$}"] (t1) <->["\escale{$w_i$}"] (t2) <->["\escale{$k_1$}"] (t3) <->["\escale{$k_0$}"] (t4);
};
\end{scope}
\begin{scope}[yshift=-1.2cm]
\begin{scope}%G_i'
\node (T) at (-0.4,0) {\nscale{$G_i'=$}};
\foreach \i in {0,...,4} {\coordinate (\i) at (\i*1.15cm+0.2cm,0);}
\foreach \i in {0,...,4} {\node (t\i) at (\i) {\nscale{$g_{i,\i}'$}};}
\graph { 
(t0) <->["\escale{$v_i'$}"] (t1) <->["\escale{$w_i'$}"] (t2) <->["\escale{$k_1$}"] (t3) <->["\escale{$k_0$}"] (t4);% <->["\escale{$k_1$}"] (t5);
};
\end{scope}
\begin{scope}[xshift=6.2cm]%F_i'
\node (T) at (-0.4,0) {\nscale{$F_i'=$}};
\foreach \i in {0,...,4} {\coordinate (\i) at (\i*1.15cm+0.2cm,0);}
\foreach \i in {0,...,4} {\node (t\i) at (\i) {\nscale{$f_{i,\i}'$}};}
\graph { 
(t0) <->["\escale{$v_i'$}"] (t1) <->["\escale{$w_i'$}"] (t2) <->["\escale{$k_0$}"] (t3) <->["\escale{$k_1$}"] (t4);
};
\end{scope}
\end{scope}
\end{tikzpicture}
\end{center}
Secondly,  for every $i\in \{0,\dots, m-1\}$, the TS $A_\varphi$ has the following gadgets $T_{i,0}, T_{i,1}$ and $T_{i,0}', T_{i,1}'$ that use the elements of $\zeta_i=\{X_{i_0},X_{i_1}, X_{i_2}\}$ and $\zeta_i'=\{X_{i_0}',X_{i_1}', X_{i_2}'\}$ as events, respectively; 
their starting states are $t_{i,0,0}$, $t_{i,1,0}$, $t_{i,0,0}'$ and $t_{i,1,0}'$:
\begin{center}
\begin{tikzpicture}[new set = import nodes]
\begin{scope}%T_00
\node (T) at (-0.8, 0) {\nscale{$T_{i,0}=$}};
\foreach \i in {0,...,7} {\coordinate (\i) at (\i*1.5cm,0);}
\foreach \i in {8,...,11} {\pgfmathparse{\i-8}\coordinate (\i) at (10.5cm-\pgfmathresult*1.6cm,-1.3);}
%\foreach \i in {0,2,4,5,7,9} {\fill[red!20, rounded corners] (\i) +(-0.35,-0.3) rectangle +(0.3,0.3);}
\foreach \i in {0,...,11} {\node (t\i) at (\i) {\nscale{$t_{i,0,\i}$}};}
\graph { 
(t0) <->["\escale{$k_0$}"] (t1) <->["\escale{$v_{7i}$}"] (t2) <->["\escale{$v_{7i+1}$}"] (t3) <->["\escale{$X_{i_0}$}"] (t4) <->["\escale{$v_{7i+2}$}"] (t5)<->["\escale{$X_{i_1}$}"](t6)<->["\escale{$v_{7i+3}$}"](t7)<->[swap, "\escale{$X_{i_2}$}"](t8)<->["\escale{$v_{7i+4}$}"](t9)<->["\escale{$v_{7i+5}$}"](t10)<->["\escale{$k_0$}"](t11);
};
\end{scope}
\begin{scope}[ yshift=-1.3cm]%T_01
\node (T) at (-0.8,0) {\nscale{$T_{i,1}=$}};
\foreach \i in {0,...,3} { \coordinate (\i) at (\i*1.4cm,0) ;}
%\foreach \i in {1,3,4} {\fill[red!20, rounded corners] (\i) +(-0.35,-0.3) rectangle +(0.3,0.3);}
\foreach \i in {0,...,3} {\node (t\i) at (\i) {\nscale{$t_{i,1,\i}$}};}
\graph { 
(t0) <->["\escale{$X_{i_0}$}"] (t1) <->["\escale{$v_{7i+6}$}"] (t2) <->["\escale{$X_{i_2}$}"] (t3);%
};
\end{scope}
\end{tikzpicture}
\end{center}
\begin{center}
\begin{tikzpicture}[new set = import nodes]
\begin{scope}%T_00'
\node (T) at (-0.8, 0) {\nscale{$T_{i,0}'=$}};
\foreach \i in {0,...,7} {\coordinate (\i) at (\i*1.5cm,0);}
\foreach \i in {8,...,11} {\pgfmathparse{\i-8}\coordinate (\i) at (10.5cm-\pgfmathresult*1.7cm,-1.3);}
%\foreach \i in {0,2,4,5,7,9} {\fill[red!20, rounded corners] (\i) +(-0.35,-0.3) rectangle +(0.3,0.3);}
\foreach \i in {0,...,11} {\node (t\i) at (\i) {\nscale{$t_{i,0,\i}'$}};}
\graph { 
(t0) <->["\escale{$k_1$}"] (t1) <->["\escale{$v_{7i}'$}"] (t2) <->["\escale{$v_{7i+1}'$}"] (t3) <->["\escale{$X_{i_0}'$}"] (t4) <->["\escale{$v_{7i+2}'$}"] (t5)<->["\escale{$X_{i_1}'$}"](t6)<->["\escale{$v_{7i+3}'$}"](t7)<->[swap,"\escale{$X_{i_2}'$}"](t8)<->[ "\escale{$v_{7i+4}'$}"](t9)<->["\escale{$v_{7i+5}'$}"](t10)<->["\escale{$k_1$}"](t11);
};
\end{scope}
\begin{scope}[ yshift=-1.3cm]%T_01'
\node (T) at (-0.8,0) {\nscale{$T_{i,1}'=$}};
\foreach \i in {0,...,3} { \coordinate (\i) at (\i*1.4cm,0) ;}
%\foreach \i in {1,3,4} {\fill[red!20, rounded corners] (\i) +(-0.35,-0.3) rectangle +(0.3,0.3);}
\foreach \i in {0,...,3} {\node (t\i) at (\i) {\nscale{$t_{i,1,\i}'$}};}
\graph { 
(t0) <->["\escale{$X_{i_0}'$}"] (t1) <->["\escale{$v_{7i+6}'$}"] (t2) <->["\escale{$X_{i_2}'$}"] (t3);%
};
\end{scope}
\end{tikzpicture}
\end{center}
Finally, the gadgets are connected via their starting states to finally build $A_\varphi$ as follows:
\begin{center}
\begin{tikzpicture}[scale=0.9,new set = import nodes]
\begin{scope}[nodes={set = import nodes}]%G_i
\node (0) at (0,0) {\nscales{$\iota$}};
\node (1) at (1.2,0) {\nscales{$\top_0$}};
\node (2) at (2.4,0) {\nscales{$\top_1$}};
\node (dots) at (3.1,0) {$\dots$};
\node (4) at (4,0) {\nscales{$\top_{14m-2}$}};
\node (5) at (6.2,0) {\nscales{$\top_{14m-1}$}};
\node (10) at (1.2,-1.2) {\nscales{$G_0$}};
\node (20) at (2.4,-1.2) {\nscales{$F_0$}};
\node (dots0) at (3.1,-1.2) {$\dots$};
\node (40) at (4,-1.2) {\nscales{$G_{7m-1}$}};
\node (50) at (6.2,-1.2) {\nscales{$F_{7m-1}$}};
\node (11) at (1.2,1.2) {\nscales{$G_0'$}};
\node (21) at (2.4,1.2) {\nscales{$F_0'$}};
\node (dots1) at (3.1,1.2) {$\dots$};
\node (41) at (4,1.2) {\nscales{$G_{7m-1}'$}};
\node (51) at (6.2,1.2) {\nscales{$F_{7m-1}'$}};
\node (_1) at (-1.2,0) {\nscales{$\bot_0$}};
\node (_2) at (-2.4,0) {\nscales{$\bot_1$}};
\node (_dots) at (-3.1,0) {$\dots$};
\node (_4) at (-4,0) {\nscales{$\bot_{2m-2}$}};
\node (_5) at (-6.2,0) {\nscales{$\bot_{2m-1}$}};
\node (_10) at (-1.2,-1.2) {\nscales{$T_{0,0}$}};
\node (_20) at (-2.4,-1.2) {\nscales{$T_{0,1}$}};
\node (_dots0) at (-3.1,-1.2) {$\dots$};
\node (_40) at (-4,-1.2) {\nscales{$T_{m-1,0}$}};
\node (_50) at (-6.2,-1.2) {\nscales{$T_{m-1,1}$}};
\node (_11) at (-1.2,1.2) {\nscales{$T_{0,0}'$}};
\node (_21) at (-2.4,1.2) {\nscales{$T_{0,1}'$}};
\node (_dots1) at (-3.1,1.2) {$\dots$};
\node (_41) at (-4,1.2) {\nscales{$T_{m-1,0}'$}};
\node (_51) at (-6.2,1.2) {\nscales{$T_{m-1,1}'$}};
\graph { (import nodes);
 0 <->["\escales{$\otimes_0$}"] 1<->["\escales{$\otimes_1$}"] 2; 
  4 <->["\escales{$\otimes_{14m-1}$}"] 5;
 0 <->[swap,"\escales{$\oplus_0$}"] _1<->[swap, "\escales{$\oplus_1$}"] _2;
   _4 <->[swap, "\escales{$\oplus_{2m-1}$}"] _5;
1 <->[swap, "\escales{$\odot_0$}"] 10;
2 <->[swap, "\escales{$\odot_1$}"] 20;
4 <->[swap, "\escales{$\odot_{14m-2}$}"] 40;
5 <->[swap, "\escales{$\odot_{14m-1}$}"] 50;
1 <->[ "\escales{$\odot_0'$}"] 11;
2 <->["\escales{$\odot_1'$}"] 21;
4 <->["\escales{$\odot_{14m-2}'$}"] 41;
5 <->[ "\escales{$\odot_{14m-1}'$}"] 51;
_1 <->["\escales{$\ominus_0$}"] _10;
_2 <->[ "\escales{$\ominus_1$}"] _20;
_4 <->["\escales{$\ominus_{2m-2}$}"] _40;
_5 <->[ "\escales{$\ominus_{2m-1}$}"] _50;
_1 <->[ swap, "\escales{$\ominus_0'$}"] _11;
_2 <->[swap, "\escales{$\ominus_1'$}"] _21;
_4 <->[swap, "\escales{$\ominus_{2m-2}'$}"] _41;
_5 <->[swap,  "\escales{$\ominus_{2m-1}'$}"] _51;
};
\end{scope}
\end{tikzpicture}
\end{center}

The following Lemma~\ref{lem:nop_inp_ssp_implies_model} and Lemma~\ref{lem:nop_free_model_implies_ssp} prove the validity of our polynomial-time reduction and, hence, complete the proof of Theorem~\ref{the:nop_free}.

\begin{lemma}\label{lem:nop_inp_ssp_implies_model}
If $A_\varphi$ has the $\tau$-SSP, then $\varphi$ is one-in-three satisfiable.
\end{lemma}
\begin{proof}
Let $G$, $F$, $T$ and $Q$ be the paths defined in Lemma~\ref{lem:nop_free_ingredients}.
First of all, we observe that $G\cong G_{7i+j}\cong G_{7i+j}'$ and $F\cong F_{7i+j}\cong F_{7i+j}'$ and $T\cong T_{i,0} \cong T_{i,0}'$ and $Q\cong  T_{i,1}\cong T_{i,1}'$ for all $i\in \{0,\dots, m-1\}$ and $j\in \{0,\dots,6\}$.
Let $R=(sup, sig)$ be a $\tau$-region that solves $\alpha=(g_{0,2}, g_{0,4})$ (which exists, since $A_\varphi$ has the $\tau$-SSP) and let $i\in \{0,\dots, m-1\}$ be arbitrary but fixed.
By Lemma~\ref{lem:nop_free_ingredients}.\ref{lem:nop_free_ingredients_alpha}, we have either $sig(k_0)\in \save$ and $sig(k_1)=\swap$ or $sig(k_0)=\swap$ and $sig(k_1)\in \save$.
This implies $sig(u_{7i})=\dots=sig(u_{7i+6})=sig(u_{7i}')=\dots=sig(u_{7i+6}')=\swap$ by Lemma~\ref{lem:nop_free_ingredients}.\ref{lem:nop_free_ingredients_swaps}.
If $sig(k_0)\in \save$ and $sig(k_1)=\swap$, then by Lemma~\ref{lem:nop_free_ingredients}.\ref{lem:nop_free_ingredients_model}, this implies that there is exactly one event $e\in \{X_{i_0}, X_{i_1}, X_{i_2}\}$ such that $sig(e)\not=\swap$.
Consequently, since $i$ was arbitrary, if $sig(k_0)\in \save$ and $sig(k_1)=\swap$, then $M=\{e\in X \mid sig(e)\not=\swap\}$ selects exactly one variable of every clause $\zeta_i$ for all $i\in \{0,\dots, m-1\}$.
Thus, $M$ is a one in three-model of $\varphi$.
Otherwise, if $sig(k_0)=\swap$ and $sig(k_1)\in \save$, then we similarly obtain that $M'=\{e\in X' \mid sig(e)\not=\swap\}$ defines a one-in-three-model of $\varphi'$, which also implies the one-in-three satisfiability of $\varphi$.
\end{proof}
\begin{lemma}\label{lem:nop_free_model_implies_ssp}
If $\varphi$ is one-in-three satisfiable, then $A_\varphi$ has the $\tau$-SSP.
\end{lemma}
The proof of Lemma~\ref{lem:nop_free_model_implies_ssp} is shifted to the appendix, Section~\ref{sec:nop_free_model_implies_ssp}.

\section{Conclusion}
\label{sec:conclusion}
In this paper, we present the overall characterization of the computational complexity of the problem $\tau$-SSP for all 256 Boolean types of nets $\tau$.
Our presentation includes 154 new complexity results (Figure~\ref{fig:overview}: \S1 - \S3, \S6, \S9, \S10) and 102 known results (Figure~\ref{fig:overview}: \S4~\cite{DBLP:conf/tamc/TredupR19}, \S5~\cite{DBLP:conf/apn/Tredup19,DBLP:conf/apn/TredupRW18}, \S7~\cite{DBLP:conf/stacs/Schmitt96,DBLP:conf/tamc/TredupR19}, \S8~\cite{DBLP:conf/apn/TredupE20}) and classifies them in the overall context of boolean state separation.
Besides the new 150 hardness- and 4 tractability-results, this classification is one of the main contributions of this paper.
First of all, it becomes apparent that the distinction between \nop-free and \nop-equipped types is meaningful:
Within the class of \nop-free types, $\tau$-SSP turns out to be NP-complete if and only if $\swap\in \tau$ and $\tau\cap\save\not=\emptyset$.
Within the class of \nop-equipped types, a differentiation between types $\tau$ that satisfy $\tau\cap\{\res,\set\}=\emptyset$ and the ones with $\tau\cap\{\res,\set\}\not=\emptyset$ is useful:
$\tau$-SSP for the former ones is NP-complete if and only if $\tau\cap\{\inp,\out\}\not=\emptyset$ and $\swap\not\in\tau$.
In particular, $\{\swap\}\cup\tau$-SSP becomes polynomial for all these types.
On the other hand, for the latter ones, which include $i\in \{\res,\set\}$ such that $i\in \tau$, $\tau$-SSP is NP-complete as long as there is also an interaction in $\tau$ that opposes $i$, that is, $\tau\cap\exit\not=\emptyset$ if $i=\set$ and $\tau\cap\enter\not=\emptyset$ if $i=\res$.

Moreover, our proofs discover that, up to isomorphism, there are essentially four hard kernels that indicate the NP-completeness of $\tau$-SSP, namely $\{\nop,\inp\}$ and $\{\nop,\inp,\out\}$ for the \nop-equipped types and $\{\swap,\free\}$ and $\{\swap,\free,\used\}$ for the \nop-free types.
That means, for a \nop-equipped type $\tau$, the hardness of $\tau$-SSP can either be shown by a reduction of $\{\nop,\inp\}$-SSP or $\{\nop,\inp,\out\}$-SSP (or their isomorphic types), which is basically done in the proof of Lemma~\ref{lem:extensions}, or $\tau$-SSP is polynomial othwerwise.
Similarly, one finds out that for the \nop-free types $\tau$ in question, the hardness of $\tau$-SSP can be shown by a reduction of $\{\swap,\free\}$-SSP or $\{\swap,\free,\used\}$-SSP. 
Due to the space limitation, a reduction that covers \emph{all} hard \nop-free types on one blow is given instead of explicit proof.

For future work, it remains to completely characterize the computational complexity of deciding if a TS $A$ is isomorphic to the reachability graph of a Boolean Petri net, instead of only being embeddable by an injective simulation map.

\newpage
%%%%%%%%%%%%%%%%
\bibliography{myBibliography}%
%%%%%%%%%%%%%%%%

\newpage
\begin{appendix}

\section{Supporting Illustrations for the Proof of Lemma~\ref{lem:res_set_free_model_implies_ssp}. The depicted TS $A_\varphi$ originates from $\varphi$ of Example~\ref{ex:varphi}.}\label{sec:res_set_free_model_implies_ssp}

%This section graphically sketches concrete examples of the regions listed for the proof of Lemma~\ref{lem:res_set_free_model_implies_ssp}.

\begin{figure}[H]%full reduction example
\begin{minipage}{\textwidth}%R_1
\begin{center} 
\begin{tikzpicture}[new set = import nodes]

\begin{scope}[nodes={set=import nodes}]% T_0
		\foreach \i in {0,...,4} { \coordinate (\i) at (\i*1.4cm,0) ;}
		\coordinate (top) at (4*1.5cm, 3*1.2cm) ;
		\foreach \i in {0,...,3,top} {\fill[red!20] (\i) circle (0.35cm);}
		\foreach \i in {0} { \node (t0\i) at (\i) {\nscale{$t_{0,\i}$}};}
		\foreach \i in {1,...,3} { \node (\i) at (\i) {\nscale{$t_{0,\i}$}};}
		\node (top) at (top) {\nscale{$\top$}};
		\graph {
	(import nodes);
			t00 ->["\escale{$X_0$}"]1->["\escale{$X_1$}"]2->["\escale{$X_2$}"]3->[swap, bend right=25,"\escale{$u_0$}"]top;
		};
\end{scope}
\begin{scope}[yshift=1.2cm, nodes={set=import nodes}]% T_1
		\foreach \i in {0,...,3} { \coordinate (\i) at (\i*1.4cm,0) ;}
		\foreach \i in {0,...,3} {\fill[red!20] (\i) circle (0.35cm);}
		\foreach \i in {0} { \node (t1\i) at (\i) {\nscale{$t_{1,\i}$}};}
		\foreach \i in {1,...,3} { \node (\i) at (\i) {\nscale{$t_{1,\i}$}};}
		%\node (4) at (4) {\nscale{$\top$}};
		\graph {
	(import nodes);
			t10 ->["\escale{$X_0$}"]1->["\escale{$X_2$}"]2->["\escale{$X_3$}"]3->[swap, bend right=10,"\escale{$u_1$}"]top;
		};
\end{scope}
\begin{scope}[yshift=2.4cm, nodes={set=import nodes}]% T_2
		\foreach \i in {0,...,3} { \coordinate (\i) at (\i*1.4cm,0) ;}
		\foreach \i in {0,...,3} {\fill[red!20] (\i) circle (0.35cm);}
		\foreach \i in {0} { \node (t2\i) at (\i) {\nscale{$t_{2,\i}$}};}
		\foreach \i in {1,...,3} { \node (\i) at (\i) {\nscale{$t_{2,\i}$}};}
		\graph {
	(import nodes);
			t20 ->["\escale{$X_0$}"]1->["\escale{$X_1$}"]2->["\escale{$X_3$}"]3->["\escale{$u_2$}"]top;
		};
\end{scope}
\begin{scope}[yshift=3.6cm, nodes={set=import nodes}]% T_3
		\foreach \i in {0,...,3} { \coordinate (\i) at (\i*1.4cm,0) ;}
		\foreach \i in {0,...,3} {\fill[red!20] (\i) circle (0.35cm);}
		\foreach \i in {0} { \node (t3\i) at (\i) {\nscale{$t_{3,\i}$}};}
		\foreach \i in {1,...,3} { \node (\i) at (\i) {\nscale{$t_{3,\i}$}};}
		\graph {
	(import nodes);
			t30 ->["\escale{$X_2$}"]1->["\escale{$X_4$}"]2->["\escale{$X_5$}"]3->["\escale{$u_3$}"]top;
		};
\end{scope}
\begin{scope}[yshift=4.8cm, nodes={set=import nodes}]% T_4
		\foreach \i in {0,...,3} { \coordinate (\i) at (\i*1.4cm,0) ;}
		\foreach \i in {0,...,3} {\fill[red!20] (\i) circle (0.35cm);}
		\foreach \i in {0} { \node (t4\i) at (\i) {\nscale{$t_{4,\i}$}};}
		\foreach \i in {1,...,3} { \node (\i) at (\i) {\nscale{$t_{4,\i}$}};}
		\graph {
	(import nodes);
			t40 ->["\escale{$X_1$}"]1->["\escale{$X_4$}"]2->["\escale{$X_5$}"]3->["\escale{$u_4$}"]top;
		};
\end{scope}
\begin{scope}[yshift=6cm, nodes={set=import nodes}]% T_5
		\foreach \i in {0,...,3} { \coordinate (\i) at (\i*1.4cm,0) ;}
		\foreach \i in {0,...,3} {\fill[red!20] (\i) circle (0.35cm);}
		\foreach \i in {0,...,3} { \node (t5\i) at (\i) {\nscale{$t_{5,\i}$}};}
		\foreach \i in {1,...,3} { \node (\i) at (\i) {\nscale{$t_{5,\i}$}};}
		\graph {
	(import nodes);
			t50 ->["\escale{$X_3$}"]1->["\escale{$X_4$}"]2->["\escale{$X_5$}"]3->[bend left =10,"\escale{$u_5$}"]top;
		};
\end{scope}
\begin{scope}[yshift=7.2cm, nodes={set=import nodes}]% 
		\foreach \i in {0,1} { \coordinate (\i) at (\i*4cm,0) ;}
		\foreach \i in {0,1} {\fill[red!20] (\i) circle (0.35cm);}
		\node (t60) at (0) {\nscale{$t_{6,0}$}};
		\node (1) at (1) {\nscale{$t_{7,0}$}};
		\graph {
	(import nodes);
			t60 ->["\escale{$k$}"]1->[bend left =40,"\escale{$v$}"]top;
		};
\end{scope}
\begin{scope}[yshift=7.2cm, nodes={set=import nodes}]% 
		\graph {
	(import nodes);
			t00 ->["\escale{$w_0$}"]t10 ->["\escale{$w_1$}"]t20->["\escale{$w_2$}"]t30 ->["\escale{$w_3$}"]t40 ->["\escale{$w_4$}"]t50 ->["\escale{$w_5$}"]t60 ;
		};
\end{scope}
%%%%%%%%%%%%%%%%%%%%%%%%%%%%%%%%%%%%%%%%%%%%%%%%%
\begin{scope}[nodes={set=import nodes}]% G_0
		\foreach \i in {0,...,2} { \coordinate (\i) at (-\i*1.2cm-2.5cm,0) ;}
		%\foreach \i in {0,1} {\fill[red!20] (\i) circle (0.35cm);}
		\foreach \i in {0} { \node (g0\i) at (\i) {\nscale{$g_{0,\i}$}};}
		\foreach \i in {1,2} { \node (\i) at (\i) {\nscale{$g_{0,\i}$}};}
		\graph {
	(import nodes);
			g00 ->[swap, "\escale{$u_0$}"]1->[swap, "\escale{$k$}"]2;
		};
\end{scope}
\begin{scope}[yshift=1.2cm,nodes={set=import nodes}]% G_1
		\foreach \i in {0,...,2} { \coordinate (\i) at (-\i*1.2cm-2.5cm,0) ;}
		%\foreach \i in {0,1} {\fill[red!20] (\i) circle (0.35cm);}
		\foreach \i in {0} { \node (g1\i) at (\i) {\nscale{$g_{1,\i}$}};}
		\foreach \i in {1,2} { \node (\i) at (\i) {\nscale{$g_{1,\i}$}};}
		\graph {
	(import nodes);
			g10 ->[swap, "\escale{$u_1$}"]1->[swap, "\escale{$k$}"]2;
		};
\end{scope}
\begin{scope}[yshift=2.4cm,nodes={set=import nodes}]% G_2
		\foreach \i in {0,...,2} { \coordinate (\i) at (-\i*1.2cm-2.5cm,0) ;}
		%\foreach \i in {0,1} {\fill[red!20] (\i) circle (0.35cm);}
		\foreach \i in {0} { \node (g2\i) at (\i) {\nscale{$g_{2,\i}$}};}
		\foreach \i in {1,2} { \node (\i) at (\i) {\nscale{$g_{2,\i}$}};}
		\graph {
	(import nodes);
			g20 ->[swap, "\escale{$u_2$}"]1->[swap, "\escale{$k$}"]2;
		};
\end{scope}
\begin{scope}[yshift=3.6cm,nodes={set=import nodes}]% G_3
		\foreach \i in {0,...,2} { \coordinate (\i) at (-\i*1.2cm-2.5cm,0) ;}
		%\foreach \i in {0,1} {\fill[red!20] (\i) circle (0.35cm);}
		\foreach \i in {0} { \node (g3\i) at (\i) {\nscale{$g_{3,\i}$}};}
		\foreach \i in {1,2} { \node (\i) at (\i) {\nscale{$g_{3,\i}$}};}
		\graph {
	(import nodes);
			g30 ->[swap, "\escale{$u_3$}"]1->[swap, "\escale{$k$}"]2;
		};
\end{scope}
\begin{scope}[yshift=4.8cm,nodes={set=import nodes}]% G_4
		\foreach \i in {0,...,2} { \coordinate (\i) at (-\i*1.2cm-2.5cm,0) ;}
		%\foreach \i in {0,1} {\fill[red!20] (\i) circle (0.35cm);}
		\foreach \i in {0} { \node (g4\i) at (\i) {\nscale{$g_{4,\i}$}};}
		\foreach \i in {1,2} { \node (\i) at (\i) {\nscale{$g_{4,\i}$}};}
		\graph {
	(import nodes);
			g40 ->[swap, "\escale{$u_4$}"]1->[swap, "\escale{$k$}"]2;
		};
\end{scope}
\begin{scope}[yshift=6cm,nodes={set=import nodes}]% G_5
		\foreach \i in {0,...,2} { \coordinate (\i) at (-\i*1.2cm-2.5cm,0) ;}
		%\foreach \i in {0,1} {\fill[red!20] (\i) circle (0.35cm);}
		\foreach \i in {0} { \node (g5\i) at (\i) {\nscale{$g_{5,\i}$}};}
		\foreach \i in {1,2} { \node (\i) at (\i) {\nscale{$g_{5,\i}$}};}
		\graph {
	(import nodes);
			g50 ->[swap, "\escale{$u_5$}"]1->[swap, "\escale{$k$}"]2;
		};
\end{scope}
\begin{scope}[yshift=6cm,nodes={set=import nodes}]% G_5
		
		\graph {
	(import nodes);
			t00 ->[pos=0.85,"\escale{$y_0$}"]g00;%->[swap, "\escale{$k$}"]2;
		};
\end{scope}
\coordinate (c0) at (-1.7,0);
\coordinate (c01) at (-1.7,1.2);
\draw[->] (t00)--(c0)--(c01)--(g10) node [pos=0.4, above] {\escale{$y_1$}};
\coordinate (c1) at (-1.5,0);
\coordinate (c11) at (-1.5,2.4);
\draw[->] (t00)--(c1)--(c11)--(g20) node [pos=0.4, above] {\escale{$y_2$}};
\coordinate (c2) at (-1.3,0);
\coordinate (c21) at (-1.3,3.6);
\draw[->] (t00)--(c2)--(c21)--(g30) node [pos=0.4, above] {\escale{$y_3$}};
\coordinate (c3) at (-1.1,0);
\coordinate (c31) at (-1.1,4.8);
\draw[->] (t00)--(c3)--(c31)--(g40) node [pos=0.4, above] {\escale{$y_4$}};
\coordinate (c4) at (-0.9,0);
\coordinate (c41) at (-0.9,6);
\draw[->] (t00)--(c4)--(c41)--(g50) node [pos=0.4, above] {\escale{$y_5$}};

\end{tikzpicture}
\end{center}
\caption{The region $R_1$ of Lemma~\ref{lem:res_set_free_model_implies_ssp}.}\label{fig:R_1}
\end{minipage}
%%%%%%%%%%%%%%%
\begin{minipage}{\textwidth}
\begin{center} 
\begin{tikzpicture}[new set = import nodes]
\begin{scope}[nodes={set=import nodes}]% T_0
		\foreach \i in {0,...,3} { \coordinate (\i) at (\i*1.4cm,0) ;}
		\foreach \i in {0,...,3} {\fill[red!20] (\i) circle (0.35cm);}
		\foreach \i in {0} { \node (t0\i) at (\i) {\nscale{$t_{0,\i}$}};}
		\foreach \i in {1,...,3} { \node (\i) at (\i) {\nscale{$t_{0,\i}$}};}
		\node (top) at (4*1.5cm, 3*1.2cm) {\nscale{$\top$}};
		\graph {
	(import nodes);
			t00 ->["\escale{$X_0$}"]1->["\escale{$X_1$}"]2->["\escale{$X_2$}"]3->[swap, bend right=25,"\escale{$u_0$}"]top;
		};
\end{scope}
\begin{scope}[yshift=1.2cm, nodes={set=import nodes}]% T_1
		\foreach \i in {0,...,3} { \coordinate (\i) at (\i*1.4cm,0) ;}
		\foreach \i in {0,...,3} {\fill[red!20] (\i) circle (0.35cm);}
		\foreach \i in {0} { \node (t1\i) at (\i) {\nscale{$t_{1,\i}$}};}
		\foreach \i in {1,...,3} { \node (\i) at (\i) {\nscale{$t_{1,\i}$}};}
		\graph {
	(import nodes);
			t10 ->["\escale{$X_0$}"]1->["\escale{$X_2$}"]2->["\escale{$X_3$}"]3->[swap, bend right=10,"\escale{$u_1$}"]top;
		};
\end{scope}
\begin{scope}[yshift=2.4cm, nodes={set=import nodes}]% T_2
		\foreach \i in {0,...,3} { \coordinate (\i) at (\i*1.4cm,0) ;}
		%\foreach \i in {0} {\fill[red!20] (\i) circle (0.35cm);}
		\foreach \i in {0} { \node (t2\i) at (\i) {\nscale{$t_{2,\i}$}};}
		\foreach \i in {1,...,3} { \node (\i) at (\i) {\nscale{$t_{2,\i}$}};}
		\graph {
	(import nodes);
			t20 ->["\escale{$X_0$}"]1->["\escale{$X_1$}"]2->["\escale{$X_3$}"]3->["\escale{$u_2$}"]top;
		};
\end{scope}
\begin{scope}[yshift=3.6cm, nodes={set=import nodes}]% T_3
		\foreach \i in {0,...,3} { \coordinate (\i) at (\i*1.4cm,0) ;}
		%\foreach \i in {0,1} {\fill[red!20] (\i) circle (0.35cm);}
		\foreach \i in {0} { \node (t3\i) at (\i) {\nscale{$t_{3,\i}$}};}
		\foreach \i in {1,...,3} { \node (\i) at (\i) {\nscale{$t_{3,\i}$}};}
		\graph {
	(import nodes);
			t30 ->["\escale{$X_2$}"]1->["\escale{$X_4$}"]2->["\escale{$X_5$}"]3->["\escale{$u_3$}"]top;
		};
\end{scope}
\begin{scope}[yshift=4.8cm, nodes={set=import nodes}]% T_4
		\foreach \i in {0,...,3} { \coordinate (\i) at (\i*1.4cm,0) ;}
		%\foreach \i in {0,1} {\fill[red!20] (\i) circle (0.35cm);}
		\foreach \i in {0} { \node (t4\i) at (\i) {\nscale{$t_{4,\i}$}};}
		\foreach \i in {1,...,3} { \node (\i) at (\i) {\nscale{$t_{4,\i}$}};}
		\graph {
	(import nodes);
			t40 ->["\escale{$X_1$}"]1->["\escale{$X_4$}"]2->["\escale{$X_5$}"]3->["\escale{$u_4$}"]top;
		};
\end{scope}
\begin{scope}[yshift=6cm, nodes={set=import nodes}]% T_5
		\foreach \i in {0,...,3} { \coordinate (\i) at (\i*1.4cm,0) ;}
		%\foreach \i in {0,1} {\fill[red!20] (\i) circle (0.35cm);}
		\foreach \i in {0} { \node (t5\i) at (\i) {\nscale{$t_{5,\i}$}};}
		\foreach \i in {1,...,3} { \node (\i) at (\i) {\nscale{$t_{5,\i}$}};}
		\graph {
	(import nodes);
			t50 ->["\escale{$X_3$}"]1->["\escale{$X_4$}"]2->["\escale{$X_5$}"]3->[bend left =10,"\escale{$u_5$}"]top;
		};
\end{scope}
\begin{scope}[yshift=7.2cm, nodes={set=import nodes}]% 
		\foreach \i in {0,1} { \coordinate (\i) at (\i*4cm,0) ;}
		%\foreach \i in {0} {\fill[red!20] (\i) circle (0.35cm);}
		\node (t60) at (0) {\nscale{$t_{6,0}$}};
		\node (1) at (1) {\nscale{$t_{7,0}$}};
		\graph {
	(import nodes);
			t60 ->["\escale{$k$}"]1->[bend left =40,"\escale{$v$}"]top;
		};
\end{scope}
\begin{scope}[yshift=7.2cm, nodes={set=import nodes}]% 
		\graph {
	(import nodes);
			t00 ->["\escale{$w_0$}"]t10 ->["\escale{$w_1$}"]t20->["\escale{$w_2$}"]t30 ->["\escale{$w_3$}"]t40 ->["\escale{$w_4$}"]t50 ->["\escale{$w_5$}"]t60 ;
		};
\end{scope}
%%%%%%%%%%%%%%%%%%%%%%%%%%%%%%%%%%%%%%%%%%%%%%%%%
\begin{scope}[nodes={set=import nodes}]% G_0
		\foreach \i in {0,...,2} { \coordinate (\i) at (-\i*1.2cm-2.5cm,0) ;}
		\foreach \i in {0} {\fill[red!20] (\i) circle (0.35cm);}
		\foreach \i in {0} { \node (g0\i) at (\i) {\nscale{$g_{0,\i}$}};}
		\foreach \i in {1,2} { \node (\i) at (\i) {\nscale{$g_{0,\i}$}};}
		\graph {
	(import nodes);
			g00 ->[swap, "\escale{$u_0$}"]1->[swap, "\escale{$k$}"]2;
		};
\end{scope}
\begin{scope}[yshift=1.2cm,nodes={set=import nodes}]% G_1
		\foreach \i in {0,...,2} { \coordinate (\i) at (-\i*1.2cm-2.5cm,0) ;}
		\foreach \i in {0} {\fill[red!20] (\i) circle (0.35cm);}
		\foreach \i in {0} { \node (g1\i) at (\i) {\nscale{$g_{1,\i}$}};}
		\foreach \i in {1,2} { \node (\i) at (\i) {\nscale{$g_{1,\i}$}};}
		\graph {
	(import nodes);
			g10 ->[swap, "\escale{$u_1$}"]1->[swap, "\escale{$k$}"]2;
		};
\end{scope}
\begin{scope}[yshift=2.4cm,nodes={set=import nodes}]% G_2
		\foreach \i in {0,...,2} { \coordinate (\i) at (-\i*1.2cm-2.5cm,0) ;}
		%\foreach \i in {0} {\fill[red!20] (\i) circle (0.35cm);}
		\foreach \i in {0} { \node (g2\i) at (\i) {\nscale{$g_{2,\i}$}};}
		\foreach \i in {1,2} { \node (\i) at (\i) {\nscale{$g_{2,\i}$}};}
		\graph {
	(import nodes);
			g20 ->[swap, "\escale{$u_2$}"]1->[swap, "\escale{$k$}"]2;
		};
\end{scope}
\begin{scope}[yshift=3.6cm,nodes={set=import nodes}]% G_3
		\foreach \i in {0,...,2} { \coordinate (\i) at (-\i*1.2cm-2.5cm,0) ;}
		%\foreach \i in {0,1} {\fill[red!20] (\i) circle (0.35cm);}
		\foreach \i in {0} { \node (g3\i) at (\i) {\nscale{$g_{3,\i}$}};}
		\foreach \i in {1,2} { \node (\i) at (\i) {\nscale{$g_{3,\i}$}};}
		\graph {
	(import nodes);
			g30 ->[swap, "\escale{$u_3$}"]1->[swap, "\escale{$k$}"]2;
		};
\end{scope}
\begin{scope}[yshift=4.8cm,nodes={set=import nodes}]% G_4
		\foreach \i in {0,...,2} { \coordinate (\i) at (-\i*1.2cm-2.5cm,0) ;}
		%\foreach \i in {0,1} {\fill[red!20] (\i) circle (0.35cm);}
		\foreach \i in {0} { \node (g4\i) at (\i) {\nscale{$g_{4,\i}$}};}
		\foreach \i in {1,2} { \node (\i) at (\i) {\nscale{$g_{4,\i}$}};}
		\graph {
	(import nodes);
			g40 ->[swap, "\escale{$u_4$}"]1->[swap, "\escale{$k$}"]2;
		};
\end{scope}
\begin{scope}[yshift=6cm,nodes={set=import nodes}]% G_5
		\foreach \i in {0,...,2} { \coordinate (\i) at (-\i*1.2cm-2.5cm,0) ;}
		%\foreach \i in {0,1} {\fill[red!20] (\i) circle (0.35cm);}
		\foreach \i in {0} { \node (g5\i) at (\i) {\nscale{$g_{5,\i}$}};}
		\foreach \i in {1,2} { \node (\i) at (\i) {\nscale{$g_{5,\i}$}};}
		\graph {
	(import nodes);
			g50 ->[swap, "\escale{$u_5$}"]1->[swap, "\escale{$k$}"]2;
		};
\end{scope}
\begin{scope}[yshift=6cm,nodes={set=import nodes}]% 
		
		\graph {
	(import nodes);
			t00 ->[pos=0.85,"\escale{$y_0$}"]g00;
		};
\end{scope}
\coordinate (c0) at (-1.7,0);
\coordinate (c01) at (-1.7,1.2);
\draw[->] (t00)--(c0)--(c01)--(g10) node [pos=0.4, above] {\escale{$y_1$}};
\coordinate (c1) at (-1.5,0);
\coordinate (c11) at (-1.5,2.4);
\draw[->] (t00)--(c1)--(c11)--(g20) node [pos=0.4, above] {\escale{$y_2$}};
\coordinate (c2) at (-1.3,0);
\coordinate (c21) at (-1.3,3.6);
\draw[->] (t00)--(c2)--(c21)--(g30) node [pos=0.4, above] {\escale{$y_3$}};
\coordinate (c3) at (-1.1,0);
\coordinate (c31) at (-1.1,4.8);
\draw[->] (t00)--(c3)--(c31)--(g40) node [pos=0.4, above] {\escale{$y_4$}};
\coordinate (c4) at (-0.9,0);
\coordinate (c41) at (-0.9,6);
\draw[->] (t00)--(c4)--(c41)--(g50) node [pos=0.4, above] {\escale{$y_5$}};

\end{tikzpicture}
\end{center}
\caption{The region $R^T_1$ of Lemma~\ref{lem:res_set_free_model_implies_ssp}.}\label{fig:RT_1}
\end{minipage}
\end{figure}
%%%%%%%%%%%%%%%%%%%%%%%%%%%%%%%%%%%%%%
\begin{figure}[H]%full reduction example
\begin{minipage}{\textwidth}%R_2
\begin{center} 
\begin{tikzpicture}[new set = import nodes]

\begin{scope}[nodes={set=import nodes}]% T_0
		\foreach \i in {0,...,4} { \coordinate (\i) at (\i*1.4cm,0) ;}
		\coordinate (top) at (4*1.5cm, 3*1.2cm) ;
		\foreach \i in {0,...,3} {\fill[red!20] (\i) circle (0.35cm);}
		\foreach \i in {0} { \node (t0\i) at (\i) {\nscale{$t_{0,\i}$}};}
		\foreach \i in {1,...,3} { \node (\i) at (\i) {\nscale{$t_{0,\i}$}};}
		\node (top) at (top) {\nscale{$\top$}};
		\graph {
	(import nodes);
			t00 ->["\escale{$X_0$}"]1->["\escale{$X_1$}"]2->["\escale{$X_2$}"]3->[swap, bend right=25,"\escale{$u_0$}"]top;
		};
\end{scope}
\begin{scope}[yshift=1.2cm, nodes={set=import nodes}]% T_1
		\foreach \i in {0,...,3} { \coordinate (\i) at (\i*1.4cm,0) ;}
		\foreach \i in {0,...,3} {\fill[red!20] (\i) circle (0.35cm);}
		\foreach \i in {0} { \node (t1\i) at (\i) {\nscale{$t_{1,\i}$}};}
		\foreach \i in {1,...,3} { \node (\i) at (\i) {\nscale{$t_{1,\i}$}};}
		%\node (4) at (4) {\nscale{$\top$}};
		\graph {
	(import nodes);
			t10 ->["\escale{$X_0$}"]1->["\escale{$X_2$}"]2->["\escale{$X_3$}"]3->[swap, bend right=10,"\escale{$u_1$}"]top;
		};
\end{scope}
\begin{scope}[yshift=2.4cm, nodes={set=import nodes}]% T_2
		\foreach \i in {0,...,3} { \coordinate (\i) at (\i*1.4cm,0) ;}
		\foreach \i in {0,...,3} {\fill[red!20] (\i) circle (0.35cm);}
		\foreach \i in {0} { \node (t2\i) at (\i) {\nscale{$t_{2,\i}$}};}
		\foreach \i in {1,...,3} { \node (\i) at (\i) {\nscale{$t_{2,\i}$}};}
		\graph {
	(import nodes);
			t20 ->["\escale{$X_0$}"]1->["\escale{$X_1$}"]2->["\escale{$X_3$}"]3->["\escale{$u_2$}"]top;
		};
\end{scope}
\begin{scope}[yshift=3.6cm, nodes={set=import nodes}]% T_3
		\foreach \i in {0,...,3} { \coordinate (\i) at (\i*1.4cm,0) ;}
		\foreach \i in {0,...,3} {\fill[red!20] (\i) circle (0.35cm);}
		\foreach \i in {0} { \node (t3\i) at (\i) {\nscale{$t_{3,\i}$}};}
		\foreach \i in {1,...,3} { \node (\i) at (\i) {\nscale{$t_{3,\i}$}};}
		\graph {
	(import nodes);
			t30 ->["\escale{$X_2$}"]1->["\escale{$X_4$}"]2->["\escale{$X_5$}"]3->["\escale{$u_3$}"]top;
		};
\end{scope}
\begin{scope}[yshift=4.8cm, nodes={set=import nodes}]% T_4
		\foreach \i in {0,...,3} { \coordinate (\i) at (\i*1.4cm,0) ;}
		\foreach \i in {0,...,3} {\fill[red!20] (\i) circle (0.35cm);}
		\foreach \i in {0} { \node (t4\i) at (\i) {\nscale{$t_{4,\i}$}};}
		\foreach \i in {1,...,3} { \node (\i) at (\i) {\nscale{$t_{4,\i}$}};}
		\graph {
	(import nodes);
			t40 ->["\escale{$X_1$}"]1->["\escale{$X_4$}"]2->["\escale{$X_5$}"]3->["\escale{$u_4$}"]top;
		};
\end{scope}
\begin{scope}[yshift=6cm, nodes={set=import nodes}]% T_5
		\foreach \i in {0,...,3} { \coordinate (\i) at (\i*1.4cm,0) ;}
		\foreach \i in {0,...,3} {\fill[red!20] (\i) circle (0.35cm);}
		\foreach \i in {0,...,3} { \node (t5\i) at (\i) {\nscale{$t_{5,\i}$}};}
		\foreach \i in {1,...,3} { \node (\i) at (\i) {\nscale{$t_{5,\i}$}};}
		\graph {
	(import nodes);
			t50 ->["\escale{$X_3$}"]1->["\escale{$X_4$}"]2->["\escale{$X_5$}"]3->[bend left =10,"\escale{$u_5$}"]top;
		};
\end{scope}
\begin{scope}[yshift=7.2cm, nodes={set=import nodes}]% 
		\foreach \i in {0,1} { \coordinate (\i) at (\i*4cm,0) ;}
		\foreach \i in {0,1} {\fill[red!20] (\i) circle (0.35cm);}
		\node (t60) at (0) {\nscale{$t_{6,0}$}};
		\node (1) at (1) {\nscale{$t_{7,0}$}};
		\graph {
	(import nodes);
			t60 ->["\escale{$k$}"]1->[bend left =40,"\escale{$v$}"]top;
		};
\end{scope}
\begin{scope}[yshift=7.2cm, nodes={set=import nodes}]% 
		\graph {
	(import nodes);
			t00 ->["\escale{$w_0$}"]t10 ->["\escale{$w_1$}"]t20->["\escale{$w_2$}"]t30 ->["\escale{$w_3$}"]t40 ->["\escale{$w_4$}"]t50 ->["\escale{$w_5$}"]t60 ;
		};
\end{scope}
%%%%%%%%%%%%%%%%%%%%%%%%%%%%%%%%%%%%%%%%%%%%%%%%%
\begin{scope}[nodes={set=import nodes}]% G_0
		\foreach \i in {0,...,2} { \coordinate (\i) at (-\i*1.2cm-2.5cm,0) ;}
		\foreach \i in {0} {\fill[red!20] (\i) circle (0.35cm);}
		\foreach \i in {0} { \node (g0\i) at (\i) {\nscale{$g_{0,\i}$}};}
		\foreach \i in {1,2} { \node (\i) at (\i) {\nscale{$g_{0,\i}$}};}
		\graph {
	(import nodes);
			g00 ->[swap, "\escale{$u_0$}"]1->[swap, "\escale{$k$}"]2;
		};
\end{scope}
\begin{scope}[yshift=1.2cm,nodes={set=import nodes}]% G_1
		\foreach \i in {0,...,2} { \coordinate (\i) at (-\i*1.2cm-2.5cm,0) ;}
		\foreach \i in {0} {\fill[red!20] (\i) circle (0.35cm);}
		\foreach \i in {0} { \node (g1\i) at (\i) {\nscale{$g_{1,\i}$}};}
		\foreach \i in {1,2} { \node (\i) at (\i) {\nscale{$g_{1,\i}$}};}
		\graph {
	(import nodes);
			g10 ->[swap, "\escale{$u_1$}"]1->[swap, "\escale{$k$}"]2;
		};
\end{scope}
\begin{scope}[yshift=2.4cm,nodes={set=import nodes}]% G_2
		\foreach \i in {0,...,2} { \coordinate (\i) at (-\i*1.2cm-2.5cm,0) ;}
		\foreach \i in {0} {\fill[red!20] (\i) circle (0.35cm);}
		\foreach \i in {0} { \node (g2\i) at (\i) {\nscale{$g_{2,\i}$}};}
		\foreach \i in {1,2} { \node (\i) at (\i) {\nscale{$g_{2,\i}$}};}
		\graph {
	(import nodes);
			g20 ->[swap, "\escale{$u_2$}"]1->[swap, "\escale{$k$}"]2;
		};
\end{scope}
\begin{scope}[yshift=3.6cm,nodes={set=import nodes}]% G_3
		\foreach \i in {0,...,2} { \coordinate (\i) at (-\i*1.2cm-2.5cm,0) ;}
		\foreach \i in {0} {\fill[red!20] (\i) circle (0.35cm);}
		\foreach \i in {0} { \node (g3\i) at (\i) {\nscale{$g_{3,\i}$}};}
		\foreach \i in {1,2} { \node (\i) at (\i) {\nscale{$g_{3,\i}$}};}
		\graph {
	(import nodes);
			g30 ->[swap, "\escale{$u_3$}"]1->[swap, "\escale{$k$}"]2;
		};
\end{scope}
\begin{scope}[yshift=4.8cm,nodes={set=import nodes}]% G_4
		\foreach \i in {0,...,2} { \coordinate (\i) at (-\i*1.2cm-2.5cm,0) ;}
		\foreach \i in {0} {\fill[red!20] (\i) circle (0.35cm);}
		\foreach \i in {0} { \node (g4\i) at (\i) {\nscale{$g_{4,\i}$}};}
		\foreach \i in {1,2} { \node (\i) at (\i) {\nscale{$g_{4,\i}$}};}
		\graph {
	(import nodes);
			g40 ->[swap, "\escale{$u_4$}"]1->[swap, "\escale{$k$}"]2;
		};
\end{scope}
\begin{scope}[yshift=6cm,nodes={set=import nodes}]% G_5
		\foreach \i in {0,...,2} { \coordinate (\i) at (-\i*1.2cm-2.5cm,0) ;}
		\foreach \i in {0} {\fill[red!20] (\i) circle (0.35cm);}
		\foreach \i in {0} { \node (g5\i) at (\i) {\nscale{$g_{5,\i}$}};}
		\foreach \i in {1,2} { \node (\i) at (\i) {\nscale{$g_{5,\i}$}};}
		\graph {
	(import nodes);
			g50 ->[swap, "\escale{$u_5$}"]1->[swap, "\escale{$k$}"]2;
		};
\end{scope}
\begin{scope}[yshift=6cm,nodes={set=import nodes}]% G_5
		
		\graph {
	(import nodes);
			t00 ->[pos=0.85,"\escale{$y_0$}"]g00;%->[swap, "\escale{$k$}"]2;
		};
\end{scope}
\coordinate (c0) at (-1.7,0);
\coordinate (c01) at (-1.7,1.2);
\draw[->] (t00)--(c0)--(c01)--(g10) node [pos=0.4, above] {\escale{$y_1$}};
\coordinate (c1) at (-1.5,0);
\coordinate (c11) at (-1.5,2.4);
\draw[->] (t00)--(c1)--(c11)--(g20) node [pos=0.4, above] {\escale{$y_2$}};
\coordinate (c2) at (-1.3,0);
\coordinate (c21) at (-1.3,3.6);
\draw[->] (t00)--(c2)--(c21)--(g30) node [pos=0.4, above] {\escale{$y_3$}};
\coordinate (c3) at (-1.1,0);
\coordinate (c31) at (-1.1,4.8);
\draw[->] (t00)--(c3)--(c31)--(g40) node [pos=0.4, above] {\escale{$y_4$}};
\coordinate (c4) at (-0.9,0);
\coordinate (c41) at (-0.9,6);
\draw[->] (t00)--(c4)--(c41)--(g50) node [pos=0.4, above] {\escale{$y_5$}};

\end{tikzpicture}
\end{center}
\caption{The region $R_2$ of Lemma~\ref{lem:res_set_free_model_implies_ssp}.}\label{fig:R_2}
\end{minipage}
%%%%%%%%%%%%%%%
\begin{minipage}{\textwidth}%R^G_1
\begin{center} 
\begin{tikzpicture}[new set = import nodes]

\begin{scope}[nodes={set=import nodes}]% T_0
		\foreach \i in {0,...,4} { \coordinate (\i) at (\i*1.4cm,0) ;}
		\coordinate (top) at (4*1.5cm, 3*1.2cm) ;
		\foreach \i in {0,...,3,top} {\fill[red!20] (\i) circle (0.35cm);}
		\foreach \i in {0} { \node (t0\i) at (\i) {\nscale{$t_{0,\i}$}};}
		\foreach \i in {1,...,3} { \node (\i) at (\i) {\nscale{$t_{0,\i}$}};}
		\node (top) at (top) {\nscale{$\top$}};
		\graph {
	(import nodes);
			t00 ->["\escale{$X_0$}"]1->["\escale{$X_1$}"]2->["\escale{$X_2$}"]3->[swap, bend right=25,"\escale{$u_0$}"]top;
		};
\end{scope}
\begin{scope}[yshift=1.2cm, nodes={set=import nodes}]% T_1
		\foreach \i in {0,...,3} { \coordinate (\i) at (\i*1.4cm,0) ;}
		\foreach \i in {0,...,3} {\fill[red!20] (\i) circle (0.35cm);}
		\foreach \i in {0} { \node (t1\i) at (\i) {\nscale{$t_{1,\i}$}};}
		\foreach \i in {1,...,3} { \node (\i) at (\i) {\nscale{$t_{1,\i}$}};}
		%\node (4) at (4) {\nscale{$\top$}};
		\graph {
	(import nodes);
			t10 ->["\escale{$X_0$}"]1->["\escale{$X_2$}"]2->["\escale{$X_3$}"]3->[swap, bend right=10,"\escale{$u_1$}"]top;
		};
\end{scope}
\begin{scope}[yshift=2.4cm, nodes={set=import nodes}]% T_2
		\foreach \i in {0,...,3} { \coordinate (\i) at (\i*1.4cm,0) ;}
		\foreach \i in {0,...,3} {\fill[red!20] (\i) circle (0.35cm);}
		\foreach \i in {0} { \node (t2\i) at (\i) {\nscale{$t_{2,\i}$}};}
		\foreach \i in {1,...,3} { \node (\i) at (\i) {\nscale{$t_{2,\i}$}};}
		\graph {
	(import nodes);
			t20 ->["\escale{$X_0$}"]1->["\escale{$X_1$}"]2->["\escale{$X_3$}"]3->["\escale{$u_2$}"]top;
		};
\end{scope}
\begin{scope}[yshift=3.6cm, nodes={set=import nodes}]% T_3
		\foreach \i in {0,...,3} { \coordinate (\i) at (\i*1.4cm,0) ;}
		\foreach \i in {0,...,3} {\fill[red!20] (\i) circle (0.35cm);}
		\foreach \i in {0} { \node (t3\i) at (\i) {\nscale{$t_{3,\i}$}};}
		\foreach \i in {1,...,3} { \node (\i) at (\i) {\nscale{$t_{3,\i}$}};}
		\graph {
	(import nodes);
			t30 ->["\escale{$X_2$}"]1->["\escale{$X_4$}"]2->["\escale{$X_5$}"]3->["\escale{$u_3$}"]top;
		};
\end{scope}
\begin{scope}[yshift=4.8cm, nodes={set=import nodes}]% T_4
		\foreach \i in {0,...,3} { \coordinate (\i) at (\i*1.4cm,0) ;}
		\foreach \i in {0,...,3} {\fill[red!20] (\i) circle (0.35cm);}
		\foreach \i in {0} { \node (t4\i) at (\i) {\nscale{$t_{4,\i}$}};}
		\foreach \i in {1,...,3} { \node (\i) at (\i) {\nscale{$t_{4,\i}$}};}
		\graph {
	(import nodes);
			t40 ->["\escale{$X_1$}"]1->["\escale{$X_4$}"]2->["\escale{$X_5$}"]3->["\escale{$u_4$}"]top;
		};
\end{scope}
\begin{scope}[yshift=6cm, nodes={set=import nodes}]% T_5
		\foreach \i in {0,...,3} { \coordinate (\i) at (\i*1.4cm,0) ;}
		\foreach \i in {0,...,3} {\fill[red!20] (\i) circle (0.35cm);}
		\foreach \i in {0,...,3} { \node (t5\i) at (\i) {\nscale{$t_{5,\i}$}};}
		\foreach \i in {1,...,3} { \node (\i) at (\i) {\nscale{$t_{5,\i}$}};}
		\graph {
	(import nodes);
			t50 ->["\escale{$X_3$}"]1->["\escale{$X_4$}"]2->["\escale{$X_5$}"]3->[bend left =10,"\escale{$u_5$}"]top;
		};
\end{scope}
\begin{scope}[yshift=7.2cm, nodes={set=import nodes}]% 
		\foreach \i in {0,1} { \coordinate (\i) at (\i*4cm,0) ;}
		\foreach \i in {0,1} {\fill[red!20] (\i) circle (0.35cm);}
		\node (t60) at (0) {\nscale{$t_{6,0}$}};
		\node (1) at (1) {\nscale{$t_{7,0}$}};
		\graph {
	(import nodes);
			t60 ->["\escale{$k$}"]1->[bend left =40,"\escale{$v$}"]top;
		};
\end{scope}
\begin{scope}[yshift=7.2cm, nodes={set=import nodes}]% 
		\graph {
	(import nodes);
			t00 ->["\escale{$w_0$}"]t10 ->["\escale{$w_1$}"]t20->["\escale{$w_2$}"]t30 ->["\escale{$w_3$}"]t40 ->["\escale{$w_4$}"]t50 ->["\escale{$w_5$}"]t60 ;
		};
\end{scope}
%%%%%%%%%%%%%%%%%%%%%%%%%%%%%%%%%%%%%%%%%%%%%%%%%
\begin{scope}[nodes={set=import nodes}]% G_0
		\foreach \i in {0,...,2} { \coordinate (\i) at (-\i*1.2cm-2.5cm,0) ;}
		\foreach \i in {0,1,2} {\fill[red!20] (\i) circle (0.35cm);}
		\foreach \i in {0} { \node (g0\i) at (\i) {\nscale{$g_{0,\i}$}};}
		\foreach \i in {1,2} { \node (\i) at (\i) {\nscale{$g_{0,\i}$}};}
		\graph {
	(import nodes);
			g00 ->[swap, "\escale{$u_0$}"]1->[swap, "\escale{$k$}"]2;
		};
\end{scope}
\begin{scope}[yshift=1.2cm,nodes={set=import nodes}]% G_1
		\foreach \i in {0,...,2} { \coordinate (\i) at (-\i*1.2cm-2.5cm,0) ;}
		\foreach \i in {0,1,2} {\fill[red!20] (\i) circle (0.35cm);}
		\foreach \i in {0} { \node (g1\i) at (\i) {\nscale{$g_{1,\i}$}};}
		\foreach \i in {1,2} { \node (\i) at (\i) {\nscale{$g_{1,\i}$}};}
		\graph {
	(import nodes);
			g10 ->[swap, "\escale{$u_1$}"]1->[swap, "\escale{$k$}"]2;
		};
\end{scope}
\begin{scope}[yshift=2.4cm,nodes={set=import nodes}]% G_2
		\foreach \i in {0,...,2} { \coordinate (\i) at (-\i*1.2cm-2.5cm,0) ;}
		%\foreach \i in {0} {\fill[red!20] (\i) circle (0.35cm);}
		\foreach \i in {0} { \node (g2\i) at (\i) {\nscale{$g_{2,\i}$}};}
		\foreach \i in {1,2} { \node (\i) at (\i) {\nscale{$g_{2,\i}$}};}
		\graph {
	(import nodes);
			g20 ->[swap, "\escale{$u_2$}"]1->[swap, "\escale{$k$}"]2;
		};
\end{scope}
\begin{scope}[yshift=3.6cm,nodes={set=import nodes}]% G_3
		\foreach \i in {0,...,2} { \coordinate (\i) at (-\i*1.2cm-2.5cm,0) ;}
		%\foreach \i in {0} {\fill[red!20] (\i) circle (0.35cm);}
		\foreach \i in {0} { \node (g3\i) at (\i) {\nscale{$g_{3,\i}$}};}
		\foreach \i in {1,2} { \node (\i) at (\i) {\nscale{$g_{3,\i}$}};}
		\graph {
	(import nodes);
			g30 ->[swap, "\escale{$u_3$}"]1->[swap, "\escale{$k$}"]2;
		};
\end{scope}
\begin{scope}[yshift=4.8cm,nodes={set=import nodes}]% G_4
		\foreach \i in {0,...,2} { \coordinate (\i) at (-\i*1.2cm-2.5cm,0) ;}
		%\foreach \i in {0} {\fill[red!20] (\i) circle (0.35cm);}
		\foreach \i in {0} { \node (g4\i) at (\i) {\nscale{$g_{4,\i}$}};}
		\foreach \i in {1,2} { \node (\i) at (\i) {\nscale{$g_{4,\i}$}};}
		\graph {
	(import nodes);
			g40 ->[swap, "\escale{$u_4$}"]1->[swap, "\escale{$k$}"]2;
		};
\end{scope}
\begin{scope}[yshift=6cm,nodes={set=import nodes}]% G_5
		\foreach \i in {0,...,2} { \coordinate (\i) at (-\i*1.2cm-2.5cm,0) ;}
		%\foreach \i in {0} {\fill[red!20] (\i) circle (0.35cm);}
		\foreach \i in {0} { \node (g5\i) at (\i) {\nscale{$g_{5,\i}$}};}
		\foreach \i in {1,2} { \node (\i) at (\i) {\nscale{$g_{5,\i}$}};}
		\graph {
	(import nodes);
			g50 ->[swap, "\escale{$u_5$}"]1->[swap, "\escale{$k$}"]2;
		};
\end{scope}
\begin{scope}[yshift=6cm,nodes={set=import nodes}]% G_5
		
		\graph {
	(import nodes);
			t00 ->[pos=0.85,"\escale{$y_0$}"]g00;%->[swap, "\escale{$k$}"]2;
		};
\end{scope}
\coordinate (c0) at (-1.7,0);
\coordinate (c01) at (-1.7,1.2);
\draw[->] (t00)--(c0)--(c01)--(g10) node [pos=0.4, above] {\escale{$y_1$}};
\coordinate (c1) at (-1.5,0);
\coordinate (c11) at (-1.5,2.4);
\draw[->] (t00)--(c1)--(c11)--(g20) node [pos=0.4, above] {\escale{$y_2$}};
\coordinate (c2) at (-1.3,0);
\coordinate (c21) at (-1.3,3.6);
\draw[->] (t00)--(c2)--(c21)--(g30) node [pos=0.4, above] {\escale{$y_3$}};
\coordinate (c3) at (-1.1,0);
\coordinate (c31) at (-1.1,4.8);
\draw[->] (t00)--(c3)--(c31)--(g40) node [pos=0.4, above] {\escale{$y_4$}};
\coordinate (c4) at (-0.9,0);
\coordinate (c41) at (-0.9,6);
\draw[->] (t00)--(c4)--(c41)--(g50) node [pos=0.4, above] {\escale{$y_5$}};

\end{tikzpicture}
\end{center}
\caption{The region $R^G_1$ of Lemma~\ref{lem:res_set_free_model_implies_ssp}.}\label{fig:RG_1}
\end{minipage}
\end{figure}
%%%%%%%%%%%%%%%%%%%%%%%%%%%%%%%%%%%%%%%%%%%%%%
\begin{figure}[H]%full reduction example
\begin{minipage}{\textwidth}%R^X_{4,0}
\begin{center} 
\begin{tikzpicture}[new set = import nodes]

\begin{scope}[nodes={set=import nodes}]% T_0
		\foreach \i in {0,...,4} { \coordinate (\i) at (\i*1.4cm,0) ;}
		\coordinate (top) at (4*1.5cm, 3*1.2cm) ;
		\foreach \i in {0,1} {\fill[red!20] (\i) circle (0.35cm);}
		\foreach \i in {0} { \node (t0\i) at (\i) {\nscale{$t_{0,\i}$}};}
		\foreach \i in {1,...,3} { \node (\i) at (\i) {\nscale{$t_{0,\i}$}};}
		\node (top) at (top) {\nscale{$\top$}};
		\graph {
	(import nodes);
			t00 ->["\escale{$X_0$}"]1->["\escale{$X_1$}"]2->["\escale{$X_2$}"]3->[swap, bend right=25,"\escale{$u_0$}"]top;
		};
\end{scope}
\begin{scope}[yshift=1.2cm, nodes={set=import nodes}]% T_1
		\foreach \i in {0,...,3} { \coordinate (\i) at (\i*1.4cm,0) ;}
		\foreach \i in {0,...,3} {\fill[red!20] (\i) circle (0.35cm);}
		\foreach \i in {0} { \node (t1\i) at (\i) {\nscale{$t_{1,\i}$}};}
		\foreach \i in {1,...,3} { \node (\i) at (\i) {\nscale{$t_{1,\i}$}};}
		%\node (4) at (4) {\nscale{$\top$}};
		\graph {
	(import nodes);
			t10 ->["\escale{$X_0$}"]1->["\escale{$X_2$}"]2->["\escale{$X_3$}"]3->[swap, bend right=10,"\escale{$u_1$}"]top;
		};
\end{scope}
\begin{scope}[yshift=2.4cm, nodes={set=import nodes}]% T_2
		\foreach \i in {0,...,3} { \coordinate (\i) at (\i*1.4cm,0) ;}
		\foreach \i in {0,1} {\fill[red!20] (\i) circle (0.35cm);}
		\foreach \i in {0} { \node (t2\i) at (\i) {\nscale{$t_{2,\i}$}};}
		\foreach \i in {1,...,3} { \node (\i) at (\i) {\nscale{$t_{2,\i}$}};}
		\graph {
	(import nodes);
			t20 ->["\escale{$X_0$}"]1->["\escale{$X_1$}"]2->["\escale{$X_3$}"]3->["\escale{$u_2$}"]top;
		};
\end{scope}
\begin{scope}[yshift=3.6cm, nodes={set=import nodes}]% T_3
		\foreach \i in {0,...,3} { \coordinate (\i) at (\i*1.4cm,0) ;}
		\foreach \i in {0,...,3} {\fill[red!20] (\i) circle (0.35cm);}
		\foreach \i in {0} { \node (t3\i) at (\i) {\nscale{$t_{3,\i}$}};}
		\foreach \i in {1,...,3} { \node (\i) at (\i) {\nscale{$t_{3,\i}$}};}
		\graph {
	(import nodes);
			t30 ->["\escale{$X_2$}"]1->["\escale{$X_4$}"]2->["\escale{$X_5$}"]3->["\escale{$u_3$}"]top;
		};
\end{scope}
\begin{scope}[yshift=4.8cm, nodes={set=import nodes}]% T_4
		\foreach \i in {0,...,3} { \coordinate (\i) at (\i*1.4cm,0) ;}
		\foreach \i in {0} {\fill[red!20] (\i) circle (0.35cm);}
		\foreach \i in {0} { \node (t4\i) at (\i) {\nscale{$t_{4,\i}$}};}
		\foreach \i in {1,...,3} { \node (\i) at (\i) {\nscale{$t_{4,\i}$}};}
		\graph {
	(import nodes);
			t40 ->["\escale{$X_1$}"]1->["\escale{$X_4$}"]2->["\escale{$X_5$}"]3->["\escale{$u_4$}"]top;
		};
\end{scope}
\begin{scope}[yshift=6cm, nodes={set=import nodes}]% T_5
		\foreach \i in {0,...,3} { \coordinate (\i) at (\i*1.4cm,0) ;}
		\foreach \i in {0,...,3} {\fill[red!20] (\i) circle (0.35cm);}
		\foreach \i in {0,...,3} { \node (t5\i) at (\i) {\nscale{$t_{5,\i}$}};}
		\foreach \i in {1,...,3} { \node (\i) at (\i) {\nscale{$t_{5,\i}$}};}
		\graph {
	(import nodes);
			t50 ->["\escale{$X_3$}"]1->["\escale{$X_4$}"]2->["\escale{$X_5$}"]3->[bend left =10,"\escale{$u_5$}"]top;
		};
\end{scope}
\begin{scope}[yshift=7.2cm, nodes={set=import nodes}]% 
		\foreach \i in {0,1} { \coordinate (\i) at (\i*4cm,0) ;}
		\foreach \i in {0,1} {\fill[red!20] (\i) circle (0.35cm);}
		\node (t60) at (0) {\nscale{$t_{6,0}$}};
		\node (1) at (1) {\nscale{$t_{7,0}$}};
		\graph {
	(import nodes);
			t60 ->["\escale{$k$}"]1->[bend left =40,"\escale{$v$}"]top;
		};
\end{scope}
\begin{scope}[yshift=7.2cm, nodes={set=import nodes}]% 
		\graph {
	(import nodes);
			t00 ->["\escale{$w_0$}"]t10 ->["\escale{$w_1$}"]t20->["\escale{$w_2$}"]t30 ->["\escale{$w_3$}"]t40 ->["\escale{$w_4$}"]t50 ->["\escale{$w_5$}"]t60 ;
		};
\end{scope}
%%%%%%%%%%%%%%%%%%%%%%%%%%%%%%%%%%%%%%%%%%%%%%%%%
\begin{scope}[nodes={set=import nodes}]% G_0
		\foreach \i in {0,...,2} { \coordinate (\i) at (-\i*1.2cm-2.5cm,0) ;}
		\foreach \i in {0,1,2} {\fill[red!20] (\i) circle (0.35cm);}
		\foreach \i in {0} { \node (g0\i) at (\i) {\nscale{$g_{0,\i}$}};}
		\foreach \i in {1,2} { \node (\i) at (\i) {\nscale{$g_{0,\i}$}};}
		\graph {
	(import nodes);
			g00 ->[swap, "\escale{$u_0$}"]1->[swap, "\escale{$k$}"]2;
		};
\end{scope}
\begin{scope}[yshift=1.2cm,nodes={set=import nodes}]% G_1
		\foreach \i in {0,...,2} { \coordinate (\i) at (-\i*1.2cm-2.5cm,0) ;}
		\foreach \i in {0} {\fill[red!20] (\i) circle (0.35cm);}
		\foreach \i in {0} { \node (g1\i) at (\i) {\nscale{$g_{1,\i}$}};}
		\foreach \i in {1,2} { \node (\i) at (\i) {\nscale{$g_{1,\i}$}};}
		\graph {
	(import nodes);
			g10 ->[swap, "\escale{$u_1$}"]1->[swap, "\escale{$k$}"]2;
		};
\end{scope}
\begin{scope}[yshift=2.4cm,nodes={set=import nodes}]% G_2
		\foreach \i in {0,...,2} { \coordinate (\i) at (-\i*1.2cm-2.5cm,0) ;}
		\foreach \i in {0,1,2} {\fill[red!20] (\i) circle (0.35cm);}
		\foreach \i in {0} { \node (g2\i) at (\i) {\nscale{$g_{2,\i}$}};}
		\foreach \i in {1,2} { \node (\i) at (\i) {\nscale{$g_{2,\i}$}};}
		\graph {
	(import nodes);
			g20 ->[swap, "\escale{$u_2$}"]1->[swap, "\escale{$k$}"]2;
		};
\end{scope}
\begin{scope}[yshift=3.6cm,nodes={set=import nodes}]% G_3
		\foreach \i in {0,...,2} { \coordinate (\i) at (-\i*1.2cm-2.5cm,0) ;}
		\foreach \i in {0} {\fill[red!20] (\i) circle (0.35cm);}
		\foreach \i in {0} { \node (g3\i) at (\i) {\nscale{$g_{3,\i}$}};}
		\foreach \i in {1,2} { \node (\i) at (\i) {\nscale{$g_{3,\i}$}};}
		\graph {
	(import nodes);
			g30 ->[swap, "\escale{$u_3$}"]1->[swap, "\escale{$k$}"]2;
		};
\end{scope}
\begin{scope}[yshift=4.8cm,nodes={set=import nodes}]% G_4
		\foreach \i in {0,...,2} { \coordinate (\i) at (-\i*1.2cm-2.5cm,0) ;}
		\foreach \i in {0,1,2} {\fill[red!20] (\i) circle (0.35cm);}
		\foreach \i in {0} { \node (g4\i) at (\i) {\nscale{$g_{4,\i}$}};}
		\foreach \i in {1,2} { \node (\i) at (\i) {\nscale{$g_{4,\i}$}};}
		\graph {
	(import nodes);
			g40 ->[swap, "\escale{$u_4$}"]1->[swap, "\escale{$k$}"]2;
		};
\end{scope}
\begin{scope}[yshift=6cm,nodes={set=import nodes}]% G_5
		\foreach \i in {0,...,2} { \coordinate (\i) at (-\i*1.2cm-2.5cm,0) ;}
		\foreach \i in {0} {\fill[red!20] (\i) circle (0.35cm);}
		\foreach \i in {0} { \node (g5\i) at (\i) {\nscale{$g_{5,\i}$}};}
		\foreach \i in {1,2} { \node (\i) at (\i) {\nscale{$g_{5,\i}$}};}
		\graph {
	(import nodes);
			g50 ->[swap, "\escale{$u_5$}"]1->[swap, "\escale{$k$}"]2;
		};
\end{scope}
\begin{scope}[yshift=6cm,nodes={set=import nodes}]% G_5
		
		\graph {
	(import nodes);
			t00 ->[pos=0.85,"\escale{$y_0$}"]g00;%->[swap, "\escale{$k$}"]2;
		};
\end{scope}
\coordinate (c0) at (-1.7,0);
\coordinate (c01) at (-1.7,1.2);
\draw[->] (t00)--(c0)--(c01)--(g10) node [pos=0.4, above] {\escale{$y_1$}};
\coordinate (c1) at (-1.5,0);
\coordinate (c11) at (-1.5,2.4);
\draw[->] (t00)--(c1)--(c11)--(g20) node [pos=0.4, above] {\escale{$y_2$}};
\coordinate (c2) at (-1.3,0);
\coordinate (c21) at (-1.3,3.6);
\draw[->] (t00)--(c2)--(c21)--(g30) node [pos=0.4, above] {\escale{$y_3$}};
\coordinate (c3) at (-1.1,0);
\coordinate (c31) at (-1.1,4.8);
\draw[->] (t00)--(c3)--(c31)--(g40) node [pos=0.4, above] {\escale{$y_4$}};
\coordinate (c4) at (-0.9,0);
\coordinate (c41) at (-0.9,6);
\draw[->] (t00)--(c4)--(c41)--(g50) node [pos=0.4, above] {\escale{$y_5$}};

\end{tikzpicture}
\end{center}
\caption{The region $R^X_{4,0}$ of Lemma~\ref{lem:res_set_free_model_implies_ssp}, where $i=4$, $j=0$ and $\ell=2$.}\label{fig:RX_4_0}
\end{minipage}
\end{figure}
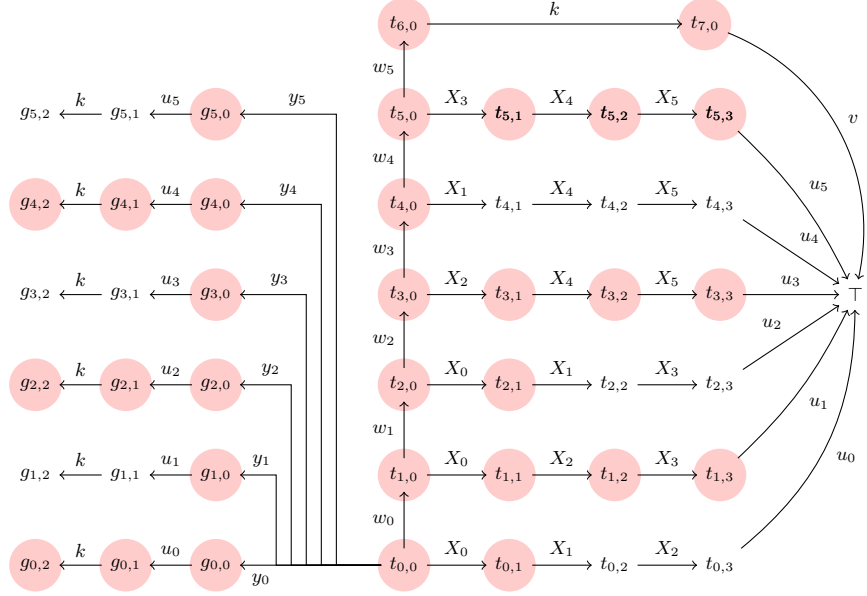

%%%%%%%%%%%%%%%%%%%%%%%%%%%%%%%%%%%%%%%%%%%%%%%%%%%%%%%%
\section{The Proof of Lemma~\ref{lem:nop_free_model_implies_ssp}.}\label{sec:nop_free_model_implies_ssp}%
%%%%%%%%%%%%%%%%%%%%%%%%%%%%%%%%%%%%%%%%%%%%%%%%%%%%%%%%

\begin{proof}[The Proof of Lemma~\ref{lem:nop_free_model_implies_ssp}.]
Let $M$ be a one-in-three model of $\varphi$.
In what follows, we assume $\free\in \tau$ and argue that $A_\varphi$ has the $\tau$-SSP.
The assumption $\free\in \tau$ is possible without loss of generality:
In particular, every presented region is a $\{\free,\swap\}$-region.
Hence, every presented support allows a $\{\swap,\res\}$-region by simply replacing $sig(e)=\free$ by $sig(e)=\res$.
Thus, by Lemma~\ref{lem:isomorphic_types} and $\{\swap,\free\}\cong \{\swap,\used\}$ and $\{\res,\swap\}\cong \{\set,\swap\}$, the $\tau$-SSP of $A_\varphi$ implies the its $\tilde{\tau}$-SSP also for the types in question where $\tilde{\tau}\cap\save_0=\emptyset$, which implies $\tilde{\tau}\cap\save_1\not=\emptyset$.

For abbreviation, let $\oplus=\{\oplus_0,\dots,\oplus_{2m-1}\}$, $\otimes=\{\otimes_0,\dots,\otimes_{14m-1}\}$, $\ominus=\{\ominus_0,\dots, \ominus_{2m-1}\}$,  $\odot=\{\odot_0,\dots, \odot_{14m-1}\}$, $V=\{v_0,\dots, v_{7m-1}\}$, $W=\{w_0,\dots, w_{7m-1}\}$, $X=\{X_0,\dots, X_{m-1}\}$, $\bot=\{\bot_0,\dots, \bot_{2m-1}\}$ and $\top=\{\top_0,\dots, \top_{14m-1}\}$, and let $\ominus'$, $\odot'$, $V'$, $W'$ and $X'$ be defined accordingly.

First of all, it is easy to see that, for all $s\in \bot\cup\top$, the SSP atom $(s,s')$ is $\tau$-solvable for all $s'\in S_{A_\varphi}\setminus\{s\}$.
Secondly, if $A$ is an arbitrary but fixed gadget of $A_\varphi$ and $s\in S_A$ and $s'\in S_{A_\varphi}\setminus S_A$, then the atom $(s,s')$ is $\tau$-solvable.
This can be seen as follows:
Let $a\in \ominus\cup\odot\cup\ominus'\cup\odot'$ be arbitrary but fixed.
(This event connects a certain gadget, say $A$, with the others.)
On the one hand, there is a region $R_1=(sup_1, sig_1)$, such that $sup_1(\iota)=0$ and, for all $e\in E_{A_\varphi}$, if $e\in \oplus\cup\otimes\cup\ominus\cup\odot\cup\ominus'\cup\odot'$, then $sig(e)=\free$, otherwise $sig(e)=\swap$.
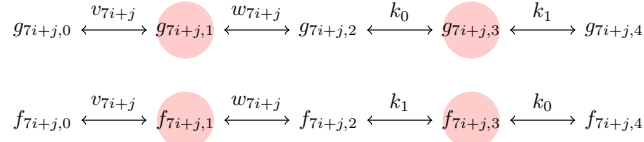
\begin{center}
\begin{tikzpicture}[new set = import nodes]
\begin{scope}[ nodes={set=import nodes}]% G_{7i+j}
		\foreach \i in {0,...,4} { \coordinate (\i) at (\i*1.9cm,0) ;}
		\foreach \i in {1,3} {\fill[red!20] (\i) circle (0.38cm);}
		\foreach \i in {0,...,4} {\node (\i) at (\i) {\nscale{$g_{7i+j,\i}$}};}
		\graph { (import nodes);

			0<->["\escale{$v_{7i+j}$}"]1<->["\escale{$w_{7i+j}$}"]2<->["\escale{$k_0$}"]3<->["\escale{$k_1$}"]4;
		};
\end{scope}
\begin{scope}[yshift=-1.2cm,nodes={set=import nodes}]% F_{7i+j}
		\foreach \i in {0,...,4} { \coordinate (\i) at (\i*1.9cm,0) ;}
		\foreach \i in {1,3} {\fill[red!20] (\i) circle (0.38cm);}
		\foreach \i in {0,...,4} { \node (\i) at (\i) {\nscale{$f_{7i+j,\i}$}};}
		\graph { (import nodes);

			0<->["\escale{$v_{7i+j}$}"]1<->["\escale{$w_{7i+j}$}"]2<->["\escale{$k_1$}"]3<->["\escale{$k_0$}"]4;
		};
\end{scope}
\node (text) at (4,-1.8cm) {\nscale{The region $R_1$ restricted to $G_{7i+j}$ and $F_{7i+j}$.}};
\end{tikzpicture}
\end{center}
%%%%%%%%%%%%%%%%%%%%%%%%%%%%%%%%%%%%%%%%%%%%%%%%%%%%%%%%%%%%%%%%%%%%
On the other hand, there is a region $R_2=(sup_2, sig_2)$, such that $sup(\iota)=0$ and, for all $e\in E_{A_\varphi}$, if  $e\in \oplus\cup\otimes\cup\ominus\cup\odot\cup\ominus'\cup\odot'$ and $e\not=a$, then $sig(e)=\free$, otherwise $sig(e)=\swap$.
Notice that for the unique gadget $A$ of $A_\varphi$ whose starting state is adjacent to $a$ holds, for all $s\in S_{A_\varphi}$,  that $sup_2(s)=sup_1(s)$ if $s\not\in S_A$, otherwise $sup_2(s)=1-sup_1(s)$.
Thus, if $s\in S_A$ and $s'\in S_{A_\varphi}\setminus S_A$, then $(s,s')$ is either solved by $R_1$ or by $R_2$.
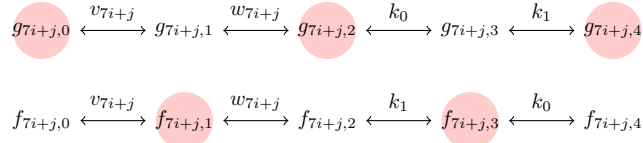
\begin{center}
\begin{tikzpicture}[new set = import nodes]
\begin{scope}[ nodes={set=import nodes}]% G_{7i+j}
		\foreach \i in {0,...,4} { \coordinate (\i) at (\i*1.9cm,0) ;}
		\foreach \i in {0,2,4} {\fill[red!20] (\i) circle (0.38cm);}
		\foreach \i in {0,...,4} {\node (\i) at (\i) {\nscale{$g_{7i+j,\i}$}};}
		\graph { (import nodes);

			0<->["\escale{$v_{7i+j}$}"]1<->["\escale{$w_{7i+j}$}"]2<->["\escale{$k_0$}"]3<->["\escale{$k_1$}"]4;
		};
\end{scope}
\begin{scope}[yshift=-1.2cm,nodes={set=import nodes}]% F_{7i+j}
		\foreach \i in {0,...,4} { \coordinate (\i) at (\i*1.9cm,0) ;}
		\foreach \i in {1,3} {\fill[red!20] (\i) circle (0.38cm);}
		\foreach \i in {0,...,4} { \node (\i) at (\i) {\nscale{$f_{7i+j,\i}$}};}
		\graph { (import nodes);

			0<->["\escale{$v_{7i+j}$}"]1<->["\escale{$w_{7i+j}$}"]2<->["\escale{$k_1$}"]3<->["\escale{$k_0$}"]4;
		};
\end{scope}
\node (text) at (4,-1.8cm) {\nscale{The region $R_2$, where $a=\odot_{14i+2j}$, restricted to $G_{7i+j}$ and $F_{7i+j}$.}};
\end{tikzpicture}
\end{center}

Justified by these observations, it only remains to consider separation atoms $(s,s')$ where $s$ and $s'$ belong to the same gadget $A$ of $A_\varphi$.

Let $i\in \{0,\dots, m-1\}$ and $j\in \{0,\dots,6\}$ be arbitrary but fixed and let's argue for the solvability of $s$ for all $s\in S_{G_{7i+j}}$.
The following $\tau$-region $R^g_{7i+j,3}=(sup, sig)$ solves $(g_{7i+j,3},s)$ for all $s\in \{g_{7i+j,0},g_{7i+j,1}, g_{7i+j,2}, g_{7i+j,4}\}$:
$sup(\iota)=0$;
for all $e\in E_{A_\varphi}$, if $e\in\{k_0,k_1\}\cup\{\ominus_{2i},\ominus_{2i}'\mid i\in \{0,\dots, m-1\}\}$, then $sig(e)=\swap$, otherwise $sig(e)=\free$.
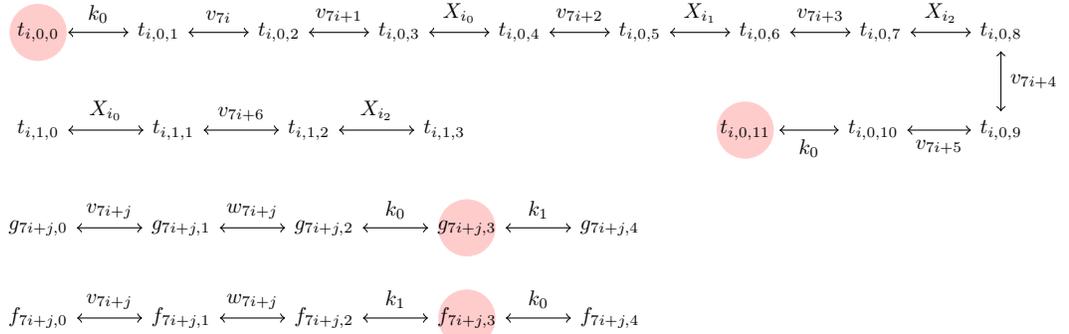
\begin{center}
\begin{tikzpicture}[new set = import nodes]
\begin{scope}%T_00
\foreach \i in {0,...,8} {\coordinate (\i) at (\i*1.6cm,0);}
\foreach \i in {9,...,11} {\pgfmathparse{\i-9}\coordinate (\i) at (12.8cm-\pgfmathresult*1.7cm,-1.3);}
\foreach \i in {0,11} {\fill[red!20] (\i) circle (0.38cm);}
\foreach \i in {0,...,11} {\node (t\i) at (\i) {\nscale{$t_{i,0,\i}$}};}
\graph { 
(t0) <->["\escale{$k_0$}"] (t1) <->["\escale{$v_{7i}$}"] (t2) <->["\escale{$v_{7i+1}$}"] (t3) <->["\escale{$X_{i_0}$}"] (t4) <->["\escale{$v_{7i+2}$}"] (t5)<->["\escale{$X_{i_1}$}"](t6)<->["\escale{$v_{7i+3}$}"](t7)<->["\escale{$X_{i_2}$}"](t8)<->["\escale{$v_{7i+4}$}"](t9)<->["\escale{$v_{7i+5}$}"](t10)<->["\escale{$k_0$}"](t11);
};
\end{scope}
\begin{scope}[ yshift=-1.3cm]%T_01

\foreach \i in {0,...,3} { \coordinate (\i) at (\i*1.8cm,0) ;}
%\foreach \i in {0,11} {\fill[red!20] (\i) circle (0.38cm);}
\foreach \i in {0,...,3} {\node (t\i) at (\i) {\nscale{$t_{i,1,\i}$}};}
\graph { 
(t0) <->["\escale{$X_{i_0}$}"] (t1) <->["\escale{$v_{7i+6}$}"] (t2) <->["\escale{$X_{i_2}$}"] (t3);%
};
\end{scope}
\begin{scope}[yshift=-2.6cm, nodes={set=import nodes}]% 
		\foreach \i in {0,...,4} { \coordinate (\i) at (\i*1.9cm,0) ;}
		\foreach \i in {3} {\fill[red!20] (\i) circle (0.38cm);}
		\foreach \i in {0,...,4} {\node (\i) at (\i) {\nscale{$g_{7i+j,\i}$}};}
		\graph { (import nodes);

			0<->["\escale{$v_{7i+j}$}"]1<->["\escale{$w_{7i+j}$}"]2<->["\escale{$k_0$}"]3<->["\escale{$k_1$}"]4;
		};
\end{scope}
\begin{scope}[yshift=-3.8cm,nodes={set=import nodes}]% 
		\foreach \i in {0,...,4} { \coordinate (\i) at (\i*1.9cm,0) ;}
		\foreach \i in {3} {\fill[red!20] (\i) circle (0.38cm);}
		\foreach \i in {0,...,4} { \node (\i) at (\i) {\nscale{$f_{7i+j,\i}$}};}
		\graph { (import nodes);

			0<->["\escale{$v_{7i+j}$}"]1<->["\escale{$w_{7i+j}$}"]2<->["\escale{$k_1$}"]3<->["\escale{$k_0$}"]4;
		};
\end{scope}
\node (text) at (5,-4.5cm) {\nscale{The region $R^g_{7i+j,3}$ restricted to $T_{i,0}, T_{i,1}, G_{7i+j}$ and $F_{7i+j}$.}};
\end{tikzpicture}
\end{center}
%%%%%%%%%%%%%%%%%%%%%%%%%%%%%%%%%%%%%%%%%%%%%%%%%%%%%%%%%%%%%%%%%%%%%%%

The following $\tau$-region $R^g_{7i+j,0}=(sup, sig)$ solves $(g_{7i+j,0},s)$ for all $s\in \{g_{7i+j,1}, g_{7i+j,2}, g_{7i+j,4}\}$:
$sup(\iota)=0$;
for all $e\in E_{A_\varphi}$, if $e\in\{k_0,k_1\} \cup \odot\cup\odot'\cup V\cup V'\cup X\cup X'$, then $sig(e)=\swap$, otherwise $sig(e)=\free$.
\begin{center}
\begin{tikzpicture}[new set = import nodes]
\begin{scope}%T_00
\foreach \i in {0,...,8} {\coordinate (\i) at (\i*1.6cm,0);}
\foreach \i in {9,...,11} {\pgfmathparse{\i-9}\coordinate (\i) at (12.8cm-\pgfmathresult*1.7cm,-1.3);}
\foreach \i in {1,3,5,7,9,11} {\fill[red!20] (\i) circle (0.38cm);}
\foreach \i in {0,...,11} {\node (t\i) at (\i) {\nscale{$t_{i,0,\i}$}};}
\graph { 
(t0) <->["\escale{$k_0$}"] (t1) <->["\escale{$v_{7i}$}"] (t2) <->["\escale{$v_{7i+1}$}"] (t3) <->["\escale{$X_{i_0}$}"] (t4) <->["\escale{$v_{7i+2}$}"] (t5)<->["\escale{$X_{i_1}$}"](t6)<->["\escale{$v_{7i+3}$}"](t7)<->["\escale{$X_{i_2}$}"](t8)<->["\escale{$v_{7i+4}$}"](t9)<->["\escale{$v_{7i+5}$}"](t10)<->["\escale{$k_0$}"](t11);
};
\end{scope}
\begin{scope}[ yshift=-1.3cm]%T_01

\foreach \i in {0,...,3} { \coordinate (\i) at (\i*1.8cm,0) ;}
\foreach \i in {1,3} {\fill[red!20] (\i) circle (0.38cm);}
\foreach \i in {0,...,3} {\node (t\i) at (\i) {\nscale{$t_{i,1,\i}$}};}
\graph { 
(t0) <->["\escale{$X_{i_0}$}"] (t1) <->["\escale{$v_{7i+6}$}"] (t2) <->["\escale{$X_{i_2}$}"] (t3);%
};
\end{scope}
\begin{scope}[yshift=-2.6cm, nodes={set=import nodes}]% 
		\foreach \i in {0,...,4} { \coordinate (\i) at (\i*1.9cm,0) ;}
		\foreach \i in {0,3} {\fill[red!20] (\i) circle (0.38cm);}
		\foreach \i in {0,...,4} {\node (\i) at (\i) {\nscale{$g_{7i+j,\i}$}};}
		\graph { (import nodes);

			0<->["\escale{$v_{7i+j}$}"]1<->["\escale{$w_{7i+j}$}"]2<->["\escale{$k_0$}"]3<->["\escale{$k_1$}"]4;
		};
\end{scope}
\begin{scope}[yshift=-3.8cm,nodes={set=import nodes}]% 
		\foreach \i in {0,...,4} { \coordinate (\i) at (\i*1.9cm,0) ;}
		\foreach \i in {0,3} {\fill[red!20] (\i) circle (0.38cm);}
		\foreach \i in {0,...,4} { \node (\i) at (\i) {\nscale{$f_{7i+j,\i}$}};}
		\graph { (import nodes);

			0<->["\escale{$v_{7i+j}$}"]1<->["\escale{$w_{7i+j}$}"]2<->["\escale{$k_1$}"]3<->["\escale{$k_0$}"]4;
		};
\end{scope}
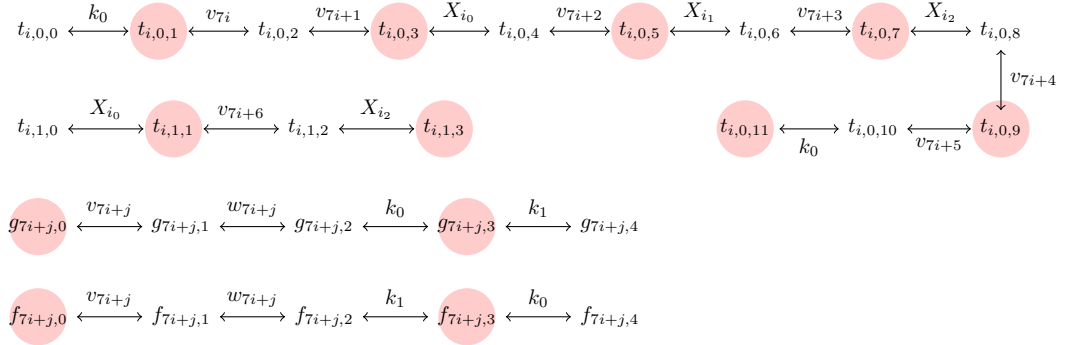
\node (text) at (5,-4.5cm) {\nscale{The region $R^g_{7i+j,0}$ restricted to $T_{i,0}, T_{i,1}, G_{7i+j}$ and $F_{7i+j}$.}};
\end{tikzpicture}
\end{center}
%%%%%%%%%%%%%%%%%%%%%%%%%%%%%%%%%%%%%%%%%%%%%%%%%%%%%%%%%%%%%%%%%%%%%%

The following $\tau$-region $R^g_{7i+j,1}=(sup, sig)$ solves $(g_{7i+j,1},s)$ for all $s\in \{g_{7i+j,2}, g_{7i+j,4}\}$:
$sup(\iota)=0$;
for all $e\in E_{A_\varphi}$, if $e\in\{k_0,k_1\}\cup V\cup V'\cup W\cup W'\cup X\cup X'$, then $sig(e)=\swap$, otherwise $sig(e)=\free$.
\begin{center}
\begin{tikzpicture}[new set = import nodes]
\begin{scope}%T_00
\foreach \i in {0,...,8} {\coordinate (\i) at (\i*1.6cm,0);}
\foreach \i in {9,...,11} {\pgfmathparse{\i-9}\coordinate (\i) at (12.8cm-\pgfmathresult*1.7cm,-1.3);}
\foreach \i in {1,3,5,7,9,11} {\fill[red!20] (\i) circle (0.38cm);}
\foreach \i in {0,...,11} {\node (t\i) at (\i) {\nscale{$t_{i,0,\i}$}};}
\graph { 
(t0) <->["\escale{$k_0$}"] (t1) <->["\escale{$v_{7i}$}"] (t2) <->["\escale{$v_{7i+1}$}"] (t3) <->["\escale{$X_{i_0}$}"] (t4) <->["\escale{$v_{7i+2}$}"] (t5)<->["\escale{$X_{i_1}$}"](t6)<->["\escale{$v_{7i+3}$}"](t7)<->["\escale{$X_{i_2}$}"](t8)<->["\escale{$v_{7i+4}$}"](t9)<->["\escale{$v_{7i+5}$}"](t10)<->["\escale{$k_0$}"](t11);
};
\end{scope}
\begin{scope}[ yshift=-1.3cm]%T_01

\foreach \i in {0,...,3} { \coordinate (\i) at (\i*1.8cm,0) ;}
\foreach \i in {1,3} {\fill[red!20] (\i) circle (0.38cm);}
\foreach \i in {0,...,3} {\node (t\i) at (\i) {\nscale{$t_{i,1,\i}$}};}
\graph { 
(t0) <->["\escale{$X_{i_0}$}"] (t1) <->["\escale{$v_{7i+6}$}"] (t2) <->["\escale{$X_{i_2}$}"] (t3);%
};
\end{scope}
\begin{scope}[yshift=-2.6cm, nodes={set=import nodes}]% 
		\foreach \i in {0,...,4} { \coordinate (\i) at (\i*1.9cm,0) ;}
		\foreach \i in {1,3} {\fill[red!20] (\i) circle (0.38cm);}
		\foreach \i in {0,...,4} {\node (\i) at (\i) {\nscale{$g_{7i+j,\i}$}};}
		\graph { (import nodes);

			0<->["\escale{$v_{7i+j}$}"]1<->["\escale{$w_{7i+j}$}"]2<->["\escale{$k_0$}"]3<->["\escale{$k_1$}"]4;
		};
\end{scope}
\begin{scope}[yshift=-3.8cm,nodes={set=import nodes}]% 
		\foreach \i in {0,...,4} { \coordinate (\i) at (\i*1.9cm,0) ;}
		\foreach \i in {1,3} {\fill[red!20] (\i) circle (0.38cm);}
		\foreach \i in {0,...,4} { \node (\i) at (\i) {\nscale{$f_{7i+j,\i}$}};}
		\graph { (import nodes);

			0<->["\escale{$v_{7i+j}$}"]1<->["\escale{$w_{7i+j}$}"]2<->["\escale{$k_1$}"]3<->["\escale{$k_0$}"]4;
		};
\end{scope}
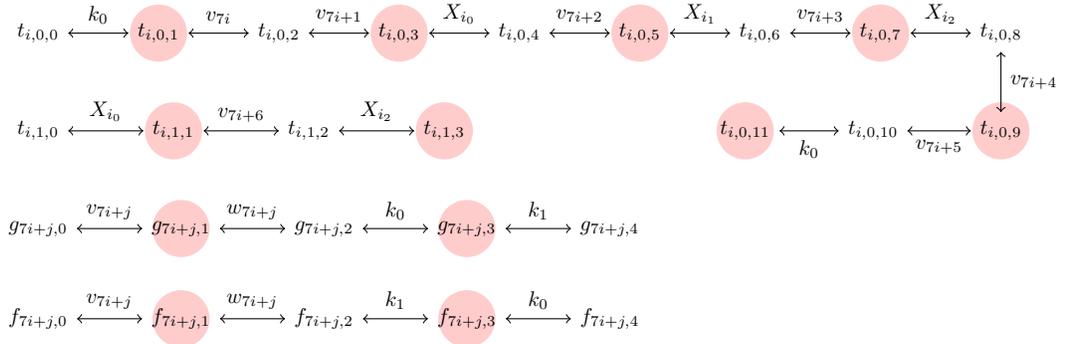
\node (text) at (5,-4.5cm) {\nscale{The region $R^g_{7i+j,1}$ restricted to $T_{i,0}, T_{i,1}, G_{7i+j}$ and $F_{7i+j}$.}};
\end{tikzpicture}
\end{center}
%%%%%%%%%%%%%%%%%%%%%%%%%%%%%%%%%%%%%%%%%%%%%%%%%%%%%%%%%%%%%%%%%%%%

Restricted to $G_{7i+j}$, it remains to solve the atom $(g_{7i+j,2}, g_{7i+j,4})$.
The following region $R^g_{7i+j,2}=(sup, sig)$ uses the existence of the one-in-three model $M$ and solves this atom:
$sup(\iota)=0$;
for all $e\in E_{A_\varphi}$, if $e\in \{k_1\}\cup V\cup V'\cup W\cup W'\cup (X\setminus M)\cup X'\cup\{\odot_{2i+1}, \odot_{2i}'\mid i\in \{0,\dots, 7m-1\}\}$, then $sig(e)=\swap$, otherwise $sig(e)=\free$.
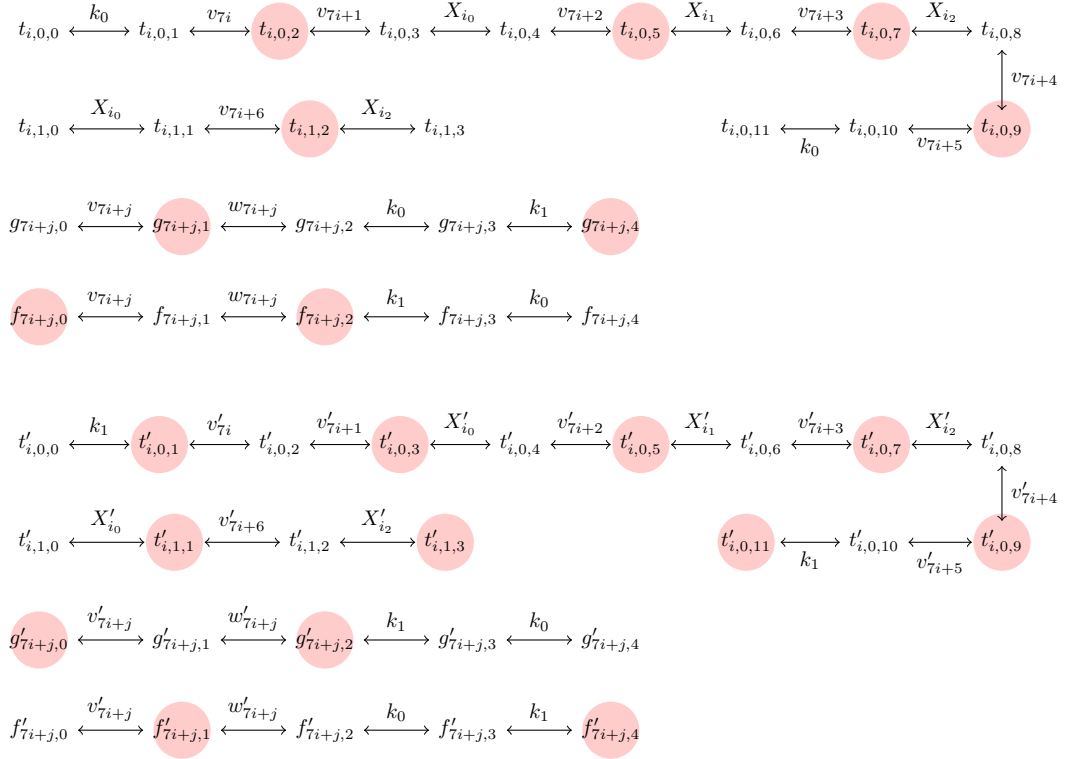
\begin{center}
\begin{tikzpicture}[new set = import nodes]
\begin{scope}%T_00
\foreach \i in {0,...,8} {\coordinate (\i) at (\i*1.6cm,0);}
\foreach \i in {9,...,11} {\pgfmathparse{\i-9}\coordinate (\i) at (12.8cm-\pgfmathresult*1.7cm,-1.3);}
\foreach \i in {2,5,7,9} {\fill[red!20] (\i) circle (0.38cm);}
\foreach \i in {0,...,11} {\node (t\i) at (\i) {\nscale{$t_{i,0,\i}$}};}
\graph { 
(t0) <->["\escale{$k_0$}"] (t1) <->["\escale{$v_{7i}$}"] (t2) <->["\escale{$v_{7i+1}$}"] (t3) <->["\escale{$X_{i_0}$}"] (t4) <->["\escale{$v_{7i+2}$}"] (t5)<->["\escale{$X_{i_1}$}"](t6)<->["\escale{$v_{7i+3}$}"](t7)<->["\escale{$X_{i_2}$}"](t8)<->["\escale{$v_{7i+4}$}"](t9)<->["\escale{$v_{7i+5}$}"](t10)<->["\escale{$k_0$}"](t11);
};
\end{scope}
\begin{scope}[ yshift=-1.3cm]%T_01

\foreach \i in {0,...,3} { \coordinate (\i) at (\i*1.8cm,0) ;}
\foreach \i in {2} {\fill[red!20] (\i) circle (0.38cm);}
\foreach \i in {0,...,3} {\node (t\i) at (\i) {\nscale{$t_{i,1,\i}$}};}
\graph { 
(t0) <->["\escale{$X_{i_0}$}"] (t1) <->["\escale{$v_{7i+6}$}"] (t2) <->["\escale{$X_{i_2}$}"] (t3);%
};
\end{scope}
\begin{scope}[yshift=-2.6cm, nodes={set=import nodes}]% 
		\foreach \i in {0,...,4} { \coordinate (\i) at (\i*1.9cm,0) ;}
		\foreach \i in {1,4} {\fill[red!20] (\i) circle (0.38cm);}
		\foreach \i in {0,...,4} {\node (\i) at (\i) {\nscale{$g_{7i+j,\i}$}};}
		\graph { (import nodes);

			0<->["\escale{$v_{7i+j}$}"]1<->["\escale{$w_{7i+j}$}"]2<->["\escale{$k_0$}"]3<->["\escale{$k_1$}"]4;
		};
\end{scope}
\begin{scope}[yshift=-3.8cm,nodes={set=import nodes}]% 
		\foreach \i in {0,...,4} { \coordinate (\i) at (\i*1.9cm,0) ;}
		\foreach \i in {0,2} {\fill[red!20] (\i) circle (0.38cm);}
		\foreach \i in {0,...,4} { \node (\i) at (\i) {\nscale{$f_{7i+j,\i}$}};}
		\graph { (import nodes);

			0<->["\escale{$v_{7i+j}$}"]1<->["\escale{$w_{7i+j}$}"]2<->["\escale{$k_1$}"]3<->["\escale{$k_0$}"]4;
		};
\end{scope}
\begin{scope}[yshift=-5.5cm]
\begin{scope}%T_00
\foreach \i in {0,...,8} {\coordinate (\i) at (\i*1.6cm,0);}
\foreach \i in {9,...,11} {\pgfmathparse{\i-9}\coordinate (\i) at (12.8cm-\pgfmathresult*1.7cm,-1.3);}
\foreach \i in {1,3,5,7,9,11} {\fill[red!20] (\i) circle (0.38cm);}
\foreach \i in {0,...,11} {\node (t\i) at (\i) {\nscale{$t_{i,0,\i}'$}};}
\graph { 
(t0) <->["\escale{$k_1$}"] (t1) <->["\escale{$v_{7i}'$}"] (t2) <->["\escale{$v_{7i+1}'$}"] (t3) <->["\escale{$X_{i_0}'$}"] (t4) <->["\escale{$v_{7i+2}'$}"] (t5)<->["\escale{$X_{i_1}'$}"](t6)<->["\escale{$v_{7i+3}'$}"](t7)<->["\escale{$X_{i_2}'$}"](t8)<->["\escale{$v_{7i+4}'$}"](t9)<->["\escale{$v_{7i+5}'$}"](t10)<->["\escale{$k_1$}"](t11);
};
\end{scope}
\begin{scope}[ yshift=-1.3cm]%T_01

\foreach \i in {0,...,3} { \coordinate (\i) at (\i*1.8cm,0) ;}
\foreach \i in {1,3} {\fill[red!20] (\i) circle (0.38cm);}
\foreach \i in {0,...,3} {\node (t\i) at (\i) {\nscale{$t_{i,1,\i}'$}};}
\graph { 
(t0) <->["\escale{$X_{i_0}'$}"] (t1) <->["\escale{$v_{7i+6}'$}"] (t2) <->["\escale{$X_{i_2}'$}"] (t3);%
};
\end{scope}
\begin{scope}[yshift=-2.6cm, nodes={set=import nodes}]% 
		\foreach \i in {0,...,4} { \coordinate (\i) at (\i*1.9cm,0) ;}
		\foreach \i in {0,2} {\fill[red!20] (\i) circle (0.38cm);}
		\foreach \i in {0,...,4} {\node (\i) at (\i) {\nscale{$g_{7i+j,\i}'$}};}
		\graph { (import nodes);

			0<->["\escale{$v_{7i+j}'$}"]1<->["\escale{$w_{7i+j}'$}"]2<->["\escale{$k_1$}"]3<->["\escale{$k_0$}"]4;
		};
\end{scope}
\begin{scope}[yshift=-3.8cm,nodes={set=import nodes}]% 
		\foreach \i in {0,...,4} { \coordinate (\i) at (\i*1.9cm,0) ;}
		\foreach \i in {1,4} {\fill[red!20] (\i) circle (0.38cm);}
		\foreach \i in {0,...,4} { \node (\i) at (\i) {\nscale{$f_{7i+j,\i}'$}};}
		\graph { (import nodes);

			0<->["\escale{$v_{7i+j}'$}"]1<->["\escale{$w_{7i+j}'$}"]2<->["\escale{$k_0$}"]3<->["\escale{$k_1$}"]4;
		};
\end{scope}
\end{scope}
\node (text) at (6.5,-10.25cm) {\nscale{The region $R^g_{7i+j,2}$ restricted to $T_{i,0}, T_{i,1}, G_{7i+j}, F_{7i+j}$ and $T_{i,0}', T_{i,1}', G_{7i+j}', F_{7i+j}'$, where $X_{i_0}\in M$.}};
\end{tikzpicture}
\end{center}
%%%%%%%%%%%%%%%%%%%%%%%%%%%%%%%%%%%%%%%%%%%%%%%%%%%%%%%%%%%%%%%%%%%%

Note that the regions $R^g_{7i+j,0},\dots, R^g_{7i+j,0}$, also show that if $s\in S_{F_{7i+j}}$, then $s$ is $\tau$-solvable.
Altogether, by the arbitrariness of $i$ and $j$ and the symmetry of a gadget $A$ and its counterpart $A'$, this shows, for all $\ell\in \{0,\dots, 7m-1\}$, that if $s\in S_{G_\ell}\cup S_{F_\ell}\cup S_{G_\ell'}\cup S_{F_\ell'}$, then $s$ is $\tau$-solvable. 

Next, we let $i\in \{0,\dots, m-1\}$ be arbitrary but fixed and argue that $s$ is $\tau$-solvable if $s\in S_{T_{i,0}}$.
The following $\tau$-region $R^t_{i,0,2}=(sup, sig)$ solves $(t_{i,0,2}, s)$ for all relevant $s\in S_{T_{i,0}}$:
$sup(\iota)=0$;
for all $e\in E_{A_\varphi}$, if $e\in \{v_{7i}, w_{7i}, v_{7i+1}, w_{7i+1}\}$, then $sig(e)=\swap$, otherwise $sig(e)=\free$.
Similarly, one proves that $(t_{i,0,9}, s)$ is $\tau$-solvable for all relevant $s\in S_{T_{i,0}}$.
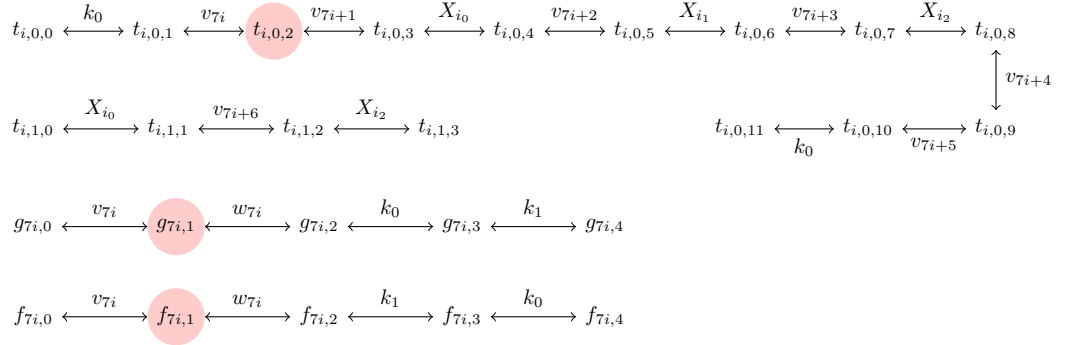
\begin{center}
\begin{tikzpicture}[new set = import nodes]
\begin{scope}%T_00
\foreach \i in {0,...,8} {\coordinate (\i) at (\i*1.6cm,0);}
\foreach \i in {9,...,11} {\pgfmathparse{\i-9}\coordinate (\i) at (12.8cm-\pgfmathresult*1.7cm,-1.3);}
\foreach \i in {2} {\fill[red!20] (\i) circle (0.38cm);}
\foreach \i in {0,...,11} {\node (t\i) at (\i) {\nscale{$t_{i,0,\i}$}};}
\graph { 
(t0) <->["\escale{$k_0$}"] (t1) <->["\escale{$v_{7i}$}"] (t2) <->["\escale{$v_{7i+1}$}"] (t3) <->["\escale{$X_{i_0}$}"] (t4) <->["\escale{$v_{7i+2}$}"] (t5)<->["\escale{$X_{i_1}$}"](t6)<->["\escale{$v_{7i+3}$}"](t7)<->["\escale{$X_{i_2}$}"](t8)<->["\escale{$v_{7i+4}$}"](t9)<->["\escale{$v_{7i+5}$}"](t10)<->["\escale{$k_0$}"](t11);
};
\end{scope}
\begin{scope}[ yshift=-1.3cm]%T_01

\foreach \i in {0,...,3} { \coordinate (\i) at (\i*1.8cm,0) ;}
%\foreach \i in {2} {\fill[red!20] (\i) circle (0.38cm);}
\foreach \i in {0,...,3} {\node (t\i) at (\i) {\nscale{$t_{i,1,\i}$}};}
\graph { 
(t0) <->["\escale{$X_{i_0}$}"] (t1) <->["\escale{$v_{7i+6}$}"] (t2) <->["\escale{$X_{i_2}$}"] (t3);%
};
\end{scope}
\begin{scope}[yshift=-2.6cm, nodes={set=import nodes}]% 
		\foreach \i in {0,...,4} { \coordinate (\i) at (\i*1.9cm,0) ;}
		\foreach \i in {1} {\fill[red!20] (\i) circle (0.38cm);}
		\foreach \i in {0,...,4} {\node (\i) at (\i) {\nscale{$g_{7i,\i}$}};}
		\graph { (import nodes);

			0<->["\escale{$v_{7i}$}"]1<->["\escale{$w_{7i}$}"]2<->["\escale{$k_0$}"]3<->["\escale{$k_1$}"]4;
		};
\end{scope}
\begin{scope}[yshift=-3.8cm,nodes={set=import nodes}]% 
		\foreach \i in {0,...,4} { \coordinate (\i) at (\i*1.9cm,0) ;}
		\foreach \i in {1} {\fill[red!20] (\i) circle (0.38cm);}
		\foreach \i in {0,...,4} { \node (\i) at (\i) {\nscale{$f_{7i,\i}$}};}
		\graph { (import nodes);

			0<->["\escale{$v_{7i}$}"]1<->["\escale{$w_{7i}$}"]2<->["\escale{$k_1$}"]3<->["\escale{$k_0$}"]4;
		};
\end{scope}
\node (text) at (5,-4.5cm) {\nscale{The region $R^t_{i,0,2}$ restricted to $T_{i,0}, T_{i,1}, G_{7i}$ and $F_{7i}$.}};
\end{tikzpicture}
\end{center}
%%%%%%%%%%%%%%%%%%%%%%%%%%%%%%%%%%%%%%%%%%%%%%%%%%%%%%%%%%%%%%%%%%%%

Let's consider the $\tau$-solvability of $(t_{i,0,3},s)$ for all relevant states $s\in S_{T_{i,0}}$.
In the following, we argue that there is a $\tau$-region $R^t_{i,0,3}=(sup, sig)$ that, for all $e\in E(T_{i,0})$, maps $e$ to swap if $e\in \{v_{7i+1}, X_{i_0}\}$ and to $\res$ otherwise, such that $sup(t_{i,0,3})=1$ and $sup(s)=0$ for all $s\in S_{T_{i,0}}\setminus\{t_{i,0,3}\}$.
Since $sig(X_{i_0})=\swap$, we have to deal with at most five further occurrences of $X_{i_0}$ in $A_\varphi$:
the one in $T_{i,1}$ and at most four further occurrences in say $T_{j,0}, T_{j,1}$ and $T_{\ell,0}, T_{\ell,1}$ for some $j,\ell\in \{0,\dots, m-1\}\setminus\{i\}$.
Thus, to define $R^t_{i,0,3}$ completely, we would have to do lots of tedious and mostly meaningless case analyses.
To avoid the latter, in the following, we argue informally for the existence of $R^t_{i,0,3}$:
The essential observation is that, by construction, if $s\edge{X_{i_0}}s'\in A_\varphi\setminus\{t_{i,0,3}\edge{X_{i_0}}t_{i,0,4}\}$, then there is a state $q\in \{s,s'\}$ and an event $v\in V$ such that $\fbedge{v}q$.
Notice that there might be events $v,v'\in V$ such that $\fbedge{v}s$ and $\fbedge{v'}s'$ as to be seen for $s=t_{i,0,3}$ and $s'=t_{i,0,4}$.
Anyway, to obtain $R^t_{i,0,3}$, we let $sup(\iota)=0$ and select exactly one arbitrary but fixed v-event and its corresponding w-event for every $s\edge{X_{i_0}}s'\in A_\varphi\setminus\{t_{i,0,3}\edge{X_{i_0}}t_{i,0,4}\}$ and map it to $\swap$ just like $X_{i_0}$ and $v_{7i+1}$.
The other events get a \free-signature.
One finds out that the resulting region is well-defined and behaves as announced.
Similarly, one proves the $\tau$-solvability for all $s\in \{t_{i,0,4},\dots, t_{i,0,8}\}$.
\begin{center}
\begin{tikzpicture}[new set = import nodes]
\begin{scope}%T_00
\foreach \i in {0,...,8} {\coordinate (\i) at (\i*1.6cm,0);}
\foreach \i in {9,...,11} {\pgfmathparse{\i-9}\coordinate (\i) at (12.8cm-\pgfmathresult*1.7cm,-1.3);}
\foreach \i in {3} {\fill[red!20] (\i) circle (0.38cm);}
\foreach \i in {0,...,11} {\node (t\i) at (\i) {\nscale{$t_{i,0,\i}$}};}
\graph { 
(t0) <->["\escale{$k_0$}"] (t1) <->["\escale{$v_{7i}$}"] (t2) <->["\escale{$v_{7i+1}$}"] (t3) <->["\escale{$X_{i_0}$}"] (t4) <->["\escale{$v_{7i+2}$}"] (t5)<->["\escale{$X_{i_1}$}"](t6)<->["\escale{$v_{7i+3}$}"](t7)<->["\escale{$X_{i_2}$}"](t8)<->["\escale{$v_{7i+4}$}"](t9)<->["\escale{$v_{7i+5}$}"](t10)<->["\escale{$k_0$}"](t11);
};
\end{scope}
\begin{scope}[ yshift=-1.3cm]%T_01

\foreach \i in {0,...,3} { \coordinate (\i) at (\i*1.8cm,0) ;}
\foreach \i in {1} {\fill[red!20] (\i) circle (0.38cm);}
\foreach \i in {0,...,3} {\node (t\i) at (\i) {\nscale{$t_{i,1,\i}$}};}
\graph { 
(t0) <->["\escale{$X_{i_0}$}"] (t1) <->["\escale{$v_{7i+6}$}"] (t2) <->["\escale{$X_{i_2}$}"] (t3);%
};
\end{scope}
\begin{scope}[yshift=-2.6cm, nodes={set=import nodes}]% 
		\foreach \i in {0,...,4} { \coordinate (\i) at (\i*1.9cm,0) ;}
		\foreach \i in {1} {\fill[red!20] (\i) circle (0.38cm);}
		\foreach \i in {0,...,4} {\node (\i) at (\i) {\nscale{$g_{7i+1,\i}$}};}
		\graph { (import nodes);

			0<->["\escale{$v_{7i+1}$}"]1<->["\escale{$w_{7i+1}$}"]2<->["\escale{$k_0$}"]3<->["\escale{$k_1$}"]4;
		};
\end{scope}
\begin{scope}[yshift=-3.8cm,nodes={set=import nodes}]% 
		\foreach \i in {0,...,4} { \coordinate (\i) at (\i*1.9cm,0) ;}
		\foreach \i in {1} {\fill[red!20] (\i) circle (0.38cm);}
		\foreach \i in {0,...,4} { \node (\i) at (\i) {\nscale{$f_{7i+1,\i}$}};}
		\graph { (import nodes);

			0<->["\escale{$v_{7i+1}$}"]1<->["\escale{$w_{7i+1}$}"]2<->["\escale{$k_1$}"]3<->["\escale{$k_0$}"]4;
		};
\end{scope}
\node (text) at (5,-4.5cm) {\nscale{The region $R^t_{i,0,3}$ restricted to $T_{i,0}, T_{i,1}, G_{7i+1}$ and $F_{7i+1}$.}};
\end{tikzpicture}
\end{center}
%%%%%%%%%%%%%%%%%%%%%%%%%%%%%%%%%%%%%%%%%%%%%%%%%%%%%%%%%%%%%%%%%%%%

Restricted to $T_{i,0}$, it remains to argue that $t_{i,0,0}, t_{i,0,1}, t_{i,0,10}$ and $t_{i,0,11}$ are pairwise separable.
Fortunately, the corresponding SSP atoms are already solved by regions defined above.
More exactly, the region $R^g_{7i+j,0}$ solves $(t_{i,0,0}, t_{i,0,1})$, $(t_{i,0,0}, t_{i,0,11})$, $(t_{i,0,1}, t_{i,0,11})$ and $(t_{i,0,10}, t_{i,0,11})$;
moreover, the region $R^g_{7i+j,3}$ solves $(t_{i,0,0}, t_{i,0,10})$ and $(t_{i,0,1}, t_{i,0,11})$.
Hence, $s$ is $\tau$-solvable for all $s\in S_{T_{i,0}}$.

Finally, it is easy to see that, on the one hand, the SSP atoms corresponding to $T_{i,1}$ are similarly solvable as the ones of $T_{i,0}$ and, on the other hand, the SSP atoms of $T_{i,0}'$ and $T_{i,1}'$ are similarly solvable by symmetry.
Hence, since $i$ was arbitrary, this proves that if $s\in S_{T_{i,0}}\cup S_{T_{i,1}}\cup S_{T_{i,0}'}\cup S_{T_{i,1}'}$, then $s$ is $\tau$-solvable for all $i\in \{0,\dots, m-1\}$.
\end{proof}

%%%%%%%%
\end{appendix}%
%%%%%%%%

%%%%%%%%%
\end{document}